\documentclass[10pt]{article}
\usepackage{amsmath}
\usepackage{dsfont}
\usepackage{amsfonts}
\usepackage{amssymb}
\usepackage{amsthm,xcolor}
\usepackage{enumerate}

\usepackage{tikz}

\usepackage{xcolor}

\theoremstyle{plain}
\begingroup
\newtheorem*{theorem*}{Theorem}
\newtheorem*{maintheorem}{Main Theorem}
\newtheorem{theorem}{Theorem}[section]
\newtheorem{lemma}[theorem]{Lemma}
\newtheorem{proposition}[theorem]{Proposition}
\newtheorem{corollary}[theorem]{Corollary}

\newtheorem{rmk}[theorem]{Remark}

\newtheorem{remark}[theorem]  {Remark} 
\usepackage{dsfont}
\endgroup

\theoremstyle{remark}
\begingroup
\endgroup

\mathchardef\emptyset="001F

\numberwithin{equation}{section}

\newcommand{\op}[1]{{\rm{#1}}}

\newcommand{\R}{\mathbb R}

\newcommand{\yg}{|y|^\gamma}

\newcommand{\N}{\mathbb{N}}

\newcommand{\calP}{{\mathcal{P}}}

\newcommand{\verti}[1]{\ensuremath{\left\lvert #1 \right\rvert}}

\newcommand{\diff}{\ensuremath{\mathrm{d}}}

\title{Equality of the Jellium and Uniform Electron Gas next-order asymptotic terms for Coulomb and Riesz potentials}

\author{
Codina Cotar
\footnote{
University College London, Statistical Science Department,
London,
United Kingdom,
\texttt{c.cotar@ucl.ac.uk}}
\, and
Mircea Petrache
\footnote{Pontificia Universidad Catolica de Chile, Santiago, Chile
\texttt{decostruttivismo@gmail.com}}}
\begin{document}
 \maketitle
\textit{Abstract:} We consider two sharp next-order asymptotics problems, namely the asymptotics for the minimum energy for optimal point configurations and the asymptotics for the many-marginals Optimal Transport, in both cases with Riesz costs with inverse power-law long-range interactions. The first problem describes the ground state of a Coulomb or Riesz gas, while the second appears as a semiclassical limit of the Density Functional Theory energy modelling a quantum version of the same system. Recently the second-order term in these expansions was precisely described, and corresponds respectively to a Jellium and to a Uniform Electron Gas model. 

The present work shows that for inverse-power-law interactions with power $s\in [d-2,d)$ in $d$ dimensions, the two problems have the same minimum value asymptotically. For the Coulomb case in $d=3$, our result verifies the physicists' long-standing conjecture regarding the equality of the second-order terms for these two problems. Furthermore, our work implies that, whereas minimum values are equal, the minimizers may be different. Moreover, provided that the crystallization hypothesis in $d=3$ holds, which is analogous to Abrikosov's conjecture in $d=2$, then our result verifies the physicists' conjectured $\approx 1.4442$ lower bound on the famous Lieb-Oxford constant. Our work also rigorously confirms some of the predictions formulated in \cite{CHM}, \cite{PerXc}, regarding the optimal value of the Uniform Electron Gas second-order asymptotic term.

We also show that on the whole range $s\in(0,d)$, the Uniform Electron Gas second-order constant is continuous in $s$.

\vspace{0.2in}


\section{Introduction}
We consider here two minimization problems, and compare their asymptotics. In both problems we will consider the energy of $N$-points configurations, in which the pairwise interaction $\textsf{c}(x-y)$ between points $x,y\in\mathbb R^d$, depends upon the power-law potentials
\begin{equation}\label{value c}
\textsf{c}(x):=\frac{1}{|x|^s}\text{ where }s>0.
\end{equation}
The first minimization which we consider is the minimization of energy under a suitable external "confining" potential $V:\mathbb R^d\to\mathbb R$ which is assumed to be bounded below, lower semicontinuous, such that $\{x: V(x)<\infty\}$ has non-zero $\textsf{c}$-capacity, and such that $V(x)\to\infty$ as $x\to\infty$. We will define the $N$-particle energy by
\begin{equation}\label{gas}
E_{N,s,V}(x_1,\ldots,x_N):=\sum_{1\le i\neq j\le N}\textsf{c}(x_i-x_j) + N\sum_{i=1}^NV(x_i)\text{ for }\ x_1,\ldots,x_N\in\mathbb R^d,
\end{equation}
and we consider the $N\to\infty$ asymptotics of 
\begin{equation}\label{gasmin}
\mathcal E_{N, s}(V):=\inf\left\{E_{N,s,V}(x_1,\ldots,x_N): x_1,\ldots,x_N\in\mathbb R^d\right\}
\end{equation}
One of the main motivations for studying the problem \eqref{gasmin} has its origin in numerical approximation questions, where it becomes important to study different measures of uniformity of large point configurations, see \cite{bhsreview,ks98} for reviews of the related literature, and see also the forthcoming monograph \cite{borohasa}. Other motivations come from the study of interactions of vortices in supercondictivity (see \cite{SandSerbook, Ser} and the references therein). Important open problems intimately linked to the above are Smale's $7$th problem, which in the case where the points are constrained to a submanifold in $\mathbb R^d$ requires to understand how to construct in polynomial time $N$-point configurations which are optimal to high accuracy (see \cite{smalenextcentury}), and the mathematical understanding of large-$N$ Abrikosov crystallization phenomena, especially in $2$ dimensions, as first predicted in \cite{Abrikosov}. Tackling these long-open problems is the main motivation leading to the study of large-$N$ asymptotics of the form \eqref{gasmin}.

The second problem which we consider is an $N$-marginal optimal transport (OT) problem with cost similarly given by pairwise interactions as above. For given $\mu\in\mathcal P(\mathbb R^d)$ we consider
\begin{equation}\label{otmin}
\mathcal F_{N, s}(\mu)=\inf\left\{\int_{(\mathbb{R}^d)^N}\sum_{i,j=1,i\neq j}^N\textsf{c}(x_i-x_j)\diff \gamma_N(x_1,\ldots,x_N): \gamma_N\in\mathcal P^N_{sym}(\mathbb{R}^d), \gamma_N\mapsto \mu\right\}.
\end{equation}
Here $\mathcal P^N_{sym}(\mathbb{R}^d)\subset \mathcal P((\mathbb R^d)^N)$ is the subset of probability measures which are invariant under the permutation of the $N$ factors of the cartesian product $(\mathbb R^d)^N$ and the notation $\gamma_N\mapsto\mu$ means that $\gamma_N$ has one-body density $\mu$ (physics terminology) or equivalently equal
$\mathbb{R}^d$-marginals $\mu$ (probability terminology), i.e. 
\begin{equation} \label{marginals}
    \gamma_N(\underbrace{\mathbb{R}^{d}\times\ldots\times\mathbb{R}^d}_\textrm{i-1 times}\times A_i \times \underbrace{\mathbb{R}^{d}\times\ldots\times\mathbb{R}^d}_\textrm{N-i times}) = \int_{A_i}\diff {\mu(x)}\, \mbox{ for all Borel}~A_i\subseteq \mathbb{R}^d
    \mbox{ and all }i=1,\ldots,N.
\end{equation}
The motivation for the stuy of the optimal transport problem \eqref{otmin} comes from Density Functional Theory (DFT). The functional $\mathcal F_{N,s}(\mu)$ appearing therein, in the particular case $s=1, d=3$, turns out to be a natural semiclassical limit to the famous \emph{Hohenberg-Kohn (HK)} functional from quantum mechanics, originally introduced by Hohenberg-Kohn in \cite{HK64}, and rigorously proved by Levy and Lieb in \cite{Le79}, \cite{Li83}. This semi-classical limit connection was shown only recently by \cite{CotFriKlu11}, was later further extended to $N=3$ in \cite{BinddePasc17}, and was only a short while ago proved in full generality of $N\ge 2$ by \cite{CotFriKlu17} and \cite{Lewin17}.

Concerning the $N\to\infty$ behavior of the problems \eqref{gasmin} and \eqref{otmin}, we are going first to note that the leading order term is in both cases a ``mean field'' term, and the only difference is that for the first problem we have imposed an external potential in order to confine the minimizing configurations, whereas in the second problem we did not (and the configurations are confined via the marginal constraint $\gamma_N\mapsto\mu$). 

The following results concerning power-law potentials appear respectively in \cite{cfp15} and \cite{frostman, landkof}: 
\begin{theorem}[leading-order asymptotics {\cite{cfp15}, \cite{frostman,landkof}}] \label{leading term}
\[
 \mathcal I_s(\mu):=\int_{\mathbb R^d}\int_{\mathbb R^d}\textsf{c}(x-y)\diff \mu(x)\ \diff \mu(y), \quad \mathcal I_{s,V}(\mu):=\mathcal I_s(\mu)+\int_{\mathbb R^d}V(x)\diff\mu(x).
\]
Then the functional $\mathcal I_{s,V}$ has a unique minimizer $\mu_V\in\mathcal P(\mathbb R^d)$ and we have, as $N\to\infty$, the following asymptotics:
\begin{itemize}
 \item [(a)] $\mathcal E_{N,s}(V) = N^2\mathcal I_{s,V}(\mu_V)+ o(N^2)$ (see \cite{frostman, landkof}).
 \item [(b)] $\mathcal F_{N,s}(\mu) = N^2\mathcal I_s(\mu) + o(N^2)$ (see  \cite{cfp15}).
\end{itemize}
\end{theorem}
In fact results of the above kind are very robust, and the above asymptotics extend much beyond the range of power-law pairwise interaction potentials for which we have stated the above theorem. For example, for the first problem in the \textit{hypersingular} case $s>d$ the integral defining $\mathcal I_s$ is not finite anymore for $\textsf{c}(x)=|x|^{-s}$ and the scaling of the leading term is $\sim N^{s/d}$ (or $N^2\log N$ for $s=d$), as proved in \cite{ks98,hs05}. For an in-depth discussion of recent results and open problems in this very popular and fast-moving vast field, see for example \cite{Ser, bhs2}.

\medskip

It seems that in both of the above problems, a next-order asymptotics is available only for powers with long-range interactions, i.e. in the regime $s<d$, and therefore we restrict to this case. We also restrict here to the case $s>0$, whereas the log-case, corresponding to the value $s=0$, will be briefly discussed in Remark \ref{rmklog}. In this case we have the following result, where we recall known results taken from \cite{PS} and \cite{cotpet} (for the second problem see also \cite{LewLiebSeir17} for results in the Coulomb case $s=1,d=3$). Furthermore, for works leading to the below Jellium result we point the reader to \cite{ss1d,ss2d,rs}, and for works leading to the Uniform Electron Gas result we point to \cite{CotFriKlu11, cfp15, p15}. For a variety of optimal transport results involving the Coulomb cost, or more generally repulsive costs, see the recent extensive survey \cite{diMarGerNe15}.
\begin{theorem}[Next-order terms for the two problems, \cite{PS}, \cite{cotpet}]\label{thnextorder}$ $
\begin{itemize} 
 \item [(a)] (Next-order term for the Coulomb and Riesz gases ground state \cite[Thm. 1]{PS}) Assume that $\textsf{c}$ is as in \eqref{value c} with either $d\ge 3$ and $d-2\le s<d$, or $d=2$ and $d-2<s<d$, and that $V$ is such that the equilibrium measure $\mu_V$ from Theorem \ref{leading term} exists and satisfies some suitable regularity assumptions (in particular has a density $\rho_V$, cf. Remark \ref{condprob} (a) below). Then we have the expansion
\begin{equation}\label{nextordergas}
\mathcal E_{N,s}(V)= N^2 \mathcal I_{s,V}(\mu_V)+ N^{1+\frac{s}{d}}  \textsf{C}_{\mathrm{Jel}}(s,d) \int_{\mathbb R^d} \rho_V^{1+\frac{s}{d}}(x)\, \diff x +o(N^{1+\frac{s}{d}}),
\end{equation}
where the number $\textsf{C}_{\mathrm{Jel}}(s,d)<0$ depends only on $s,d$, and is characterized as the minimum of a ``Jellium energy'' functional $\mathcal W$, as described below in  \eqref{weta} and \eqref{defc1}.
 \item [(b)] (Next-order term for the optimal transport problem with Coulomb and Riesz gases costs \cite[Thm. 1.2.1]{cotpet}) Assume that $\textsf{c}$ is as in \eqref{value c} with $0<s<d$ and that $\mu$ has a density $\rho$ such that the integrals below are finite (cf. Remark \ref{condprob} (b) below). Then we have the expansion
\begin{equation}\label{nextorderot}
\mathcal F_{N,s}(\mu)= N^2 \mathcal I_s(\mu)+ N^{1+\frac{s}{d}}  \textsf{C}_{\mathrm{UEG}}(s,d) \int_{\mathbb R^d} \rho^{1+\frac{s}{d}}(x)\, \diff x +o(N^{1+\frac{s}{d}}),
\end{equation}
where the number $\textsf{C}_{\mathrm{UEG}}(s,d)<0$ depends only on $s,d$, and is characterized as the ground state energy of a ``Uniform Electron Gas'' energy, as described below in \eqref{reformulc2}.
\end{itemize}
\end{theorem}
\begin{remark}
\label{condprob}
We make the following assumptions on the densities appearing in the two problems.
\begin{itemize}
\item [(a)] The assumptions on $\mu_V$ needed for the first part of the above lemma are discussed in \cite{PS} and in \cite{lebleserfaty}, and require that $\rho_V\in C^{0,\beta}$ and bounded, that the boundary $\partial\,\mathrm{supp}(\rho_V)$ of $\mathrm{supp}(\rho_V)$ is $C^1$-regular and that $\rho_V(x)=O(\mathrm{dist}(x,\partial\,\mathrm{supp}(\rho_V)^\alpha)$ as $x\to\,\partial\mathrm{supp}(\rho_V)$, for some $\beta\in(0,1]$ and for some $\alpha\in[0,\tfrac{2\beta d}{2d-s}]$, with $\beta=\alpha$ for $\alpha<1$, and $\beta =1$ for $\alpha\ge 1$.
\item [(b)] We note that for $\mu$ with density $\rho\in L^1(\mathbb{R}^d)$
\begin{multline}
\label{finmu}
\sup_{x\in \mathbb R^d}\int_{\mathbb R^d} \textsf{c}(x-y)\rho(y)\diff y<\infty\quad \text{ for }\rho\in L^{\frac{d}{d-s}, 1}(\mathbb{R}^d)\quad\mbox{and}\\
\quad\int_{\mathbb R^d}\int_{\mathbb R^d} \textsf{c}(x-y)\rho(x)\rho(y)\diff x\ \diff y<\infty\quad\text{ for }\quad \rho\in L^{\frac{2d}{2d-s},2}(\mathbb{R}^d),
\end{multline}
where the spaces $L^{p,q}(\mathbb{R}^d)$ are the Lorentz spaces (see Appendix Section B from \cite{cotpet} for an extended discussion on this).
\end{itemize}
\end{remark}
Our main aim in this paper is to link the above two problems via the following result:
\begin{maintheorem}\label{mainthm}
Let either $d\ge 3$ and $d-2\le s<d$, or $d=2$ and $d-2<s<d$. With the notations \eqref{nextordergas}, \eqref{nextorderot}, we have $\textsf{C}_{\mathrm{Jel}}(s,d)=\textsf{C}_{\mathrm{UEG}}(s,d)$.
\end{maintheorem}
Note that the next-order constants $\textsf{C}_{\mathrm{Jel}}(s,d)$ and $\textsf{C}_{\mathrm{UEG}}(s,d)$ \textit{do not depend} on the density of the leading-order measure (respectively $\rho_V$ and $\rho$, in our two problems). Thus the constants $\textsf{C}_{\mathrm{Jel}}(s,d), \textsf{C}_{\mathrm{UEG}}(s,d)$, characterize the \textit{microscale behavior} of the two problems, up to a scaling factor which depends only on the densities $\rho_V(x), \rho(x)$, respectively.

\medskip

We will see below in Lemma \ref{reformulc1c2} that, after intepreting the interaction $\textsf{c}(x-y)$ as some kind of generalized electrostatic potential interaction, the constant $\textsf{C}_{\mathrm{Jel}}(s,d)$ can be reformulated as the minimum of a \textit{Jellium energy} $E_{\mathrm{Jel}}$ on configurations covering the whole $\mathbb R^d$, and that the constant $\textsf{C}_{\mathrm{UEG}}(s,d)$ can be reformulated as a \textit{Uniform Electron Gas} energy $E_{\mathrm{UEG}}$, again for configurations covering the whole $\mathbb R^d$.

\medskip

For $s=1,d=3,$ the Main Theorem solves a controversy/conjecture recently formulated in \cite{lewinlieb} (for more explanations, see also Section \ref{conclusionsec} below and the introduction to \cite{LewLiebSeir17}). More precisely, unlike what was conjectured in the above works, and what is implied by only looking at the case of crystals, we find that the Jellium and Uniform Electron Gas energies are the same, even in the Coulomb case. This is not contradicting the calculation in \cite[App. B]{lewinlieb} if for example it is true that the Jellium minimizer is to good approximation crystalline, while the Uniform Electron Gas minimizer is essentially non-crystalline. By comparison with the periodic (or more generally, homogeneous) case, we see that this must be due to important boundary effects, peculiar to the Coulomb case. See \S \ref{conclusionsec} and \S \ref{boundaryeffects} for more details.

\medskip

Our main theorem furthermore says that in the range of exponents $0\le d-2< s<d$, the two minimization problems \eqref{gasmin} and \eqref{otmin} have quantitatively the same microscale behavior: As a consequence of our constructions, the corresponding generalized Riesz-interaction type Jellium and Uniform Electron Gas energy functionals have \emph{asymptotically the same minimizers}, as detailed in Remark \ref{eqmin} below. The situation in case $d=1$ is discussed in Remark \ref{rmkd1} below.
\begin{rmk} 
\label{Cristconj}
In $d=2$ the question of whether or not the $N$-point configurations in a cube $K_R=[-R/2,R/2]^d$, which minimize $E_{N,s,V}$, asymptotically converge to \emph{crystalline configurations forming a triangular lattice} as the cube $K_R$ invades $\R^d$ for $R^d=N\to\infty$ is the celebrated Abrikosov conjecture, which aims to explain and rigorously prove the result predicted in \cite{Abrikosov}. A similar wide open conjecture holds in $d=3$, where the $\mathrm{BCC}$ lattice is conjectured to be the minimizer for $0<s<3/2$, which includes the Coulomb case, and the $\mathrm{FCC}$ lattice is conjectured to be the minimizer for $s>3/2$. We note here that as $s\to\infty$ the related energy minimization approximates a best-packing problem, for which the optimizer has been proved by Hales in a computer-aided proof \cite{hales} to be the $\mathrm{FCC}$ lattice. In $d=8$ and $d=24$, it has been very recently shown by Cohn, Kumar, Miller,
Radchenko and Viazovska \cite{CKMRV} by using linear programming bounds that for interaction kernels of Gaussian or superposition-of-Gaussian type the minimizer is achieved on the $\mathrm{E}_8$ lattice and respectively on the leech lattice. However, their results are highly specific to $d=8,24$, and do not extend to any other dimension. In particular, due to the above mentioned dual behavior for $s>3/2$ and $s<3/2$ in dimension $3$, a similar universal behavior is precluded, see \cite{sastr}. In high dimensions, there is more and more evidence that the minimizers are not lattices, although this is very much
speculative at the moment.
 
 The value of $\textsf{C}_\mathrm{UEG}(1,3)$ has long been conjectured in the physics community to be equal to $\textsf{C}_\mathrm{Jel}(1,3)$. Based on the $d=3$ analogue of the Abrikosov crystallization conjecture, 
 it was conjectured long ago that $\textsf{C}_\mathrm{UEG}(1,3)\approx -1.4442$ (see \cite{wigner1934}, \cite{CHM}, \cite{PerXc},\cite{BGMW}, \cite{OC}, and \cite[\S 1.6]{GV}). In \cite{LiNarn75} the constant $\textsf{C}_\mathrm{Jel}(1,3)$ was rigorously bounded below by $-1.45$.
 
\end{rmk}

We also note the following  second result: we find the property of $\textsf{C}_\mathrm{UEG}(s,d)$ of being continuous in $s$ across all the range of exponents $0<s<d$. This will be proved below in Section \ref{contsec}. The same continuity in $s$ in the range $0<s<d$ can be proved by methods from \cite{hsss} for our generalized Riesz-type Jellium problem in a \emph{periodic} setting, but in a non-periodic setting this has not been proved so far; furthermore, the Jellium problem has not been studied in the cases $0<s<d-2$.
\begin{proposition}
\label{continuityc2}
Let $d\ge 2$. Then the value of $\mathsf{C}_\mathrm{UEG}(s,d)$ is continuous as a function of $s$, for $s\in(0,d)$.
\end{proposition}
\begin{rmk}[The case $d=1$]\label{rmkd1}
We note that for $d=1$ the next-order term in the Coulomb and Riesz gas problem has an expression like \eqref{nextordergas} and is also part of \cite{PS}, whereas previous results are in \cite{ss1d} and \cite{rs}.

The next-order expansion as in \eqref{nextorderot} for the optimal transport problem can be exactly derived by a very elegant computation for Coulomb and Riesz costs (as explained to us by Simone Di Marino \cite{DimarPer}), by means of the explicit ``monotone rearrangement'' description of the optimal transport plan from \cite{ColdePascdiMar13}. 

Due to the above explicit computations, one can find again that $\mathsf C_\mathrm{Jel}(s,1)=\mathsf C_\mathrm{UEG}(s,1)$ for $0<s<1$. Also note that while in the present work we do not consider the cases $s\le 0$, we should note that for $d=1$ the scaling exponent $N^{1+s/d}=N^{1+s}$ changes sign at the degenerate case $s=d-2=-1$, and the asymptotics considered here in that case completely change behaviour. 
\end{rmk}

\begin{rmk}[the case $s=0$ of logarithmic interactions]\label{rmklog}
In the case $s=0$ it is natural to extend the class of power-law potentials \eqref{value c} by defining $\mathsf c(x)=-\log|x|$. Concerning this case, we hereby briefly summarize what is known about the problem at hand, and what results we expect to hold.

\textbf{ The optimal transport problem:} Let $\mu\in\mathcal {P}(\mathbb R^d)$ with density $\rho$, such that $\rho\in L^{1+\epsilon}(\mathbb{R}^2)$ for some $0<\epsilon<\infty$ and $\int_{\mathbb R^d}\log (2+|x|)\rho(x)\diff x<\infty$. Then for $d=2$ and for all $N\ge 2$, we expect to have for some $c_{\mathrm{LO}}(\log,d)<0$ which does not depend on $N$ and $\mu$
\begin{equation}
\label{lotermlog}
\frac{c_{\mathrm{LO}}(\log,d)}{\epsilon}\left(1+\int_{\mathbb{R}^2} \rho^{1+\epsilon}(x)\diff x\right)
\le N^{-1}\left[\mathcal F_{N,\log}(\mu)-N^2 \mathcal I_{\log}(\mu)+\frac{N\log N}{d}\right]\le 0.
\end{equation}
By \cite[Prop. 3.8]{LNSS} we know that the above formula holds for $d=2$, but it seems to be unknown for $d\ge 3$. An explicit characterization of the next order term for $d=1$, of similar form as in \eqref{upbu1log} below, can be derived by the method introduced by Simone di Marino in \cite{DimarPer}. For $d=2$, the characterization of the next order term is still an open problem for general $\mu\in\calP(\mathbb{R}^2)$ satisfying the assumptions above.  
Here we conjecture that there exists $-\infty<\mathsf C_{\mathrm{UEG}}(\log,d)<0$, depending only on $d$, such that
\begin{equation}
 \label{upbu1log}
\mathcal F_{N,\log}(\mu)= N^2 \mathcal I_{\log}(\mu)-\frac{N\log N}{d}+N\left(\mathsf C_{\mathrm{UEG}}(\log,d)-\frac{1}{d}\int_{\mathbb{R}^d}\rho(x)\log\rho(x)\,\diff x\right)+o(N).
\end{equation}
Note that the above can be explicitly shown to hold if $\mu\in\calP(\mathbb{R}^2)$ is a uniform measure with density $\rho$ supported on a Borel set $\Sigma\subset\mathbb{R}^2$ (satisfying a so-called $\phi$-regular boundary assumption), by applying a similar subadditivity argument as in Lemma 2.5 from \cite{LewLiebSeir17}. Furthermore, the constant $\mathsf C_{\mathrm{UEG}}(\log,d)$ can be proved to satisfy
\[\mathsf C_{\mathrm{UEG}}(\log,d)=\lim_{N\rightarrow\infty}\frac{\mathcal F_{N,\log}(1_{[0,1]^2})-N^2 \mathcal I_{\log}(1_{[0,1]^2})+\frac{N\log N}{d}}{N}.\]

\medskip

\textbf{The Coulomb and Riesz gases problem:} For the Coulomb and Riesz gases problem, it has been shown in \cite{ss1d}, \cite{ss2d}, that for $d=1,2$, and with the same regularity assumptions on $\mu_V$ as before, the following similar characterization to the one conjectured in \eqref{upbu1log} holds in the logarithmic case for some $\mathsf C_{\mathrm{Jel}}(\log,d)<0$
\begin{equation}
\label{logCoulRiesz}
\mathcal E_{N,\log}(V)= N^2 \mathcal I_{\log,V}(\mu_V)-\frac{N\log N}{d}+N\left(\mathsf C_{\mathrm{Jel}}(\log,d)-\frac{1}{d}\int_{\mathbb{R}^d}\rho_V(x)\log\rho_V(x)\,\diff x\right)+o(N).
\end{equation}
Just as in \eqref{nextordergas},  $\mathsf C_{\mathrm{Jel}}(\log,d)$ is characterised as the minimum of a "Jellium" energy functional $\mathcal{W}$.
\medskip

It would be interesting to compare the two constants $\mathsf C_{\mathrm{Jel}}(\log,d)$ and $\mathsf C_{\mathrm{UEG}}(\log,d)$, which we expect to be equal at least for $d=1$.

\end{rmk}

\subsection{Some facts concerning the comparison of $C_{\mathrm{Jel}}(d-2,d)$ and $C_{\mathrm{UEG}}(d-2,d)$}
\label{conclusionsec}

Before we start our discussion, we need to introduce some more definitions.

\subsubsection{Jelllium and Uniform Electron Gas Energies}
\label{jelueg00}
To begin with, we define for a bounded domain $\Omega$ with $|\Omega|=N$ and $c:\mathbb{R}^d\times\mathbb{R}^d\rightarrow \mathbb{R}^d\cup\{+\infty\}$ \textit{the generalized Jellium problem}
\begin{eqnarray}
\label{jelueg001}
E_{\mathrm{Jel},\mathsf{c}}(\Omega)&:=&\min\left\{E_{\mathrm{Jel},\mathsf{c}}(\Omega, \vec{x}):\ x_1,\ldots,x_N\in\Omega\right\},
\end{eqnarray}
where
\begin{eqnarray}\label{defjellium00}
E_{\mathrm{Jel},\mathsf{c}}(\Omega, \vec{x})&:=&\sum_{1\le i,j\le N\atop i\neq j}\textsf{c}(x_i-x_j) - 2\sum_{1\le i\le N}\int_{\Omega}\textsf{c}(x_i-y) \diff y + \int_{\Omega}\int_{\Omega} \textsf{c}(x-y)\diff x \diff y.
\end{eqnarray}
In the particular case where $\mathsf{c}$ satisfies (\ref{value c}), we will use the notations $E_{\mathrm{Jel},s}(\Omega, \vec{x})$ and $E_{\mathrm{Jel},s}(\Omega)$. The $s=1,d=3,$ is the \textit{classical Jellium problem}, and it is of great interest in physics, where it has been extensively studied.
 
 We observe here that, by the potential-theoretic reasoning used in the proof of \cite[Thm. 5]{PS}, due to the controlled shape of the domain $\Omega$ the minimizing points $\vec x=(x_1,\ldots,x_N)$ on the right hand side of \eqref{defjellium00} automatically satisfy $x_i\in \Omega, i=1,\ldots,N$. This justifies why the minimization problem above and in \cite[App. B]{lewinlieb} is restricted to points in $\Omega$. For connections of the formulation on the right hand side of  \eqref{defjellium00} to the ``Wigner minimization problem'', see \cite[\S2.6]{BlLe15}.
 
 Note that in Lemma \ref{reformulc1c2} (a) below we prove that for $d\ge 3$ and $d-2\le s<d$, and $d=2$ and $d-2<s<d$, we have
$$\lim_{N\to\infty}\frac{E_{\mathrm{Jel},s}(\Omega)}{N}=C_{\mathrm{Jel}}(s,d),$$
where $C_{\mathrm{Jel}}(s,d)$ is the Coulomb and Riesz gases next-order term constant from (\ref{nextordergas}). This allows us to work thereafter with $E_{\mathrm{Jel}}(\Omega)$ rather than with (\ref{gasmin}), when comparing $C_{\mathrm{Jel}}(s,d)$ and $C_{\mathrm{UEG}}(s,d)$.

We move  now to the definition of the \textit{generalized Uniform Electron Gas} problem. More precisely, we define for $\Omega$ with $|\Omega|=N$
\begin{equation}\label{reformulcueg00}
 E_{\mathrm{UEG},\mathsf{c}}(\Omega):= \min\left\{\int E_{\mathrm{UEG},\mathsf{c}}(\Omega,\vec x)d\gamma_N(\vec x):\ \gamma_N\mapsto\frac{1_\Omega}{N}\right\}={\cal F}_{N,\mathsf{c}}\left(\frac{1_\Omega}{N}\right)-\int_{\Omega}\int_{\Omega} \textsf{c}(x-y)\diff x \diff y,
\end{equation}
where
$$ E_{\mathrm{UEG},\mathsf{c}}(\Omega,\vec x):=\sum_{1\le i,j\le N\atop i\neq j}\textsf{c}(x_i-x_j)- \int_{\Omega}\int_{\Omega} \textsf{c}(x-y)\diff x \diff y.$$
When $\mathsf{c}$ is of form as in (\ref{value c}), we will use the notation $E_{\mathrm{UEG},s}(\Omega,\vec x)$ and $E_{\mathrm{UEG},s}(\Omega)$, with $s=1,d=3,$ corresponding again to the \textit{classical Uniform Electron Gas} problem.

The above-defined $E_{\mathrm{UEG},s}(\Omega,\vec x)$ will play an important part in the comparison discussion following next, and also later on in the proofs.

\subsubsection{The comparison}
\label{thecomp00}

Note that similar quantities to the ones from Section \ref{jelueg00} for the Jellium, respectively Uniform Electron Gas, energies (and to the more general coresponding quantities from Section \ref{Jelintro}) appear also in \cite[App. B]{lewinlieb}, where they are used in $d=3$ in the simplified situation where the minimizing configurations are assumed to be exact lattices (as is conjectured in $d=2$ in the till-now open Abrikosov conjecture detailed below). This is different to our case, where we work with a generic Jellium-energy minimizing sequence. Furthermore, in \cite[App. B]{lewinlieb} they work with a competitor to $E_{\mathrm{UEG},s}(\Omega)$, rather than with a minimizer, as explained below.

We first need to define new constants $C_{\mathrm{Jel}}^{\mathrm{lattice}}(s,d), C_{\mathrm{UEG}}^{\mathrm{lattice}}(s,d)$, corresponding to simplified versions of the minimization problems from the left of \eqref{defjellium00} and of \eqref{reformulcueg00}, in which the configurations allowed are lattice-like only. We define
$$\mathcal L_1:=\{\mbox{lattices of density }1\}=\left\{R\mathbb Z^d:\ R\in Gl(d), |\mathrm{det}(R)|=1\right\},
$$
and for $\Lambda\in\mathcal L_1$, we denote by $Q_\Lambda$ the unit cell of $\Lambda$ and for a set $K\subset\mathbb R^d$, let
\[
K^\Lambda:=(\Lambda\cap K) + Q_\Lambda.
\]
Define then 
$$
C_{\mathrm{Jel}}^{\mathrm{lattice}}(s,d):=\lim_{N\to\infty}\frac{\min\left\{E_{\mathrm{Jel},s}(K, \vec{x}):\ K\mbox{ is a cube},\ \vec x= \Lambda\cap K,\ \Lambda\in\mathcal L_1,\ \#(K\cap \Lambda)=N,\ \right\}}{N}
$$
and
$$C_{\mathrm{UEG}}^{\mathrm{lattice}}(s,d):=\lim_{N\to\infty}\frac{\min\left\{\int E_{\mathrm{UEG},s}(K^\Lambda,\vec x)\diff\gamma_{\Lambda,K}(\vec x):\ K\mbox{ is a cube},\ \Lambda\in\mathcal L_1,\ \#(K\cap\Lambda)=N\right\}}{N}.
$$
In the above, $\gamma_{\Lambda,K}$ is defined by
\begin{equation}\label{defgammanr100}
\gamma_{\Lambda,K}:=\int_{Q_\Lambda} \delta_{\mathsf{sym}}((\Lambda\cap K)+y)\ \diff y,
\end{equation}
which can be shown to have marginal $1_{K^\Lambda}$. Furthermore, the set $(\Lambda\cap K)+y$ is the translation of $\Lambda+K$ by $y$, and
\begin{equation}\label{symmetrize001}
\delta_\mathsf{sym}(\{x_1,\ldots,x_N\}):=\frac{1}{N!}\sum_{\sigma\in\mathsf{Perm}_N}\delta_{(x_{\sigma(1)},\ldots,x_{\sigma(N)})},
\end{equation}
where $\mathsf{Perm}_N$ is the permutations group with $N$ elements.
We recall here for the benefit of the reader that if $L$ is a lattice then its density is the average number of points of the lattice per unit volume of the ambient space, or also $1/V_L$ where $V_L$ is the volume of a fundamental domain of $L$. Then the above objects were the ones treated in \cite[App. B]{lewinlieb} for the case $s=1,d=3$. In that paper the minimization over $\mathcal L_1$ is not included, and the computation is done on a single lattice $L$. Further, it is assumed therein that the measure $1_{\Omega_L}(x)dx$ has barycenter zero and zero dipole and quadrupole moments, where $\Omega_L$ is the fundamental domain of the considered lattice $L$. These facts are justified since in $s=1,d=3$ it is folklore knowledge that the BCC lattice is the minimizer in $\mathcal L_1$, and thus it suffices to work on this particular lattice, whose unit cell satisfies the aforementioned symmetry properties.

\medskip

The following relationships hold between the above quantities and the ones we work with here:
\begin{enumerate}[(1)]
\item In general, since $\mathcal L_1$ is a very restrictive (finite-dimensional) family of configurations, there holds
\begin{equation}\label{obviouslyy}
C_{\mathrm{Jel}}(s,d)\le C_{\mathrm{Jel}}^{\mathrm{lattice}}(s,d), \quad\quad C_{\mathrm{UEG}}(s,d)\le C_{\mathrm{UEG}}^{\mathrm{lattice}}(s,d)
\end{equation}
\item The generalization of Abrikosov's conjecture states that 
\begin{equation}\label{abrikoso}
\text{Abrikosov conjecture (generalized):}\quad C_{\mathrm{Jel}}(s,d) = C_{\mathrm{Jel}}^{\mathrm{lattice}}(s,d).
\end{equation}
The conjecture, initially formulated in $d=2$ for log-interactions, is open for all interesting dimensions $d\ge 2$, except $d=8,24$ \cite{CKMRV}, and all exponents $0<s<d$ as well as for log-interactions. Furthermore, as explained in Remark \ref{Cristconj} above, there is growing evidence that in high dimensions the Jellium minimizers are not lattices.
\item In \cite[App. B]{lewinlieb} it is proved by a direct computation in the case of the BCC lattice for $d=3, s=d-2=1$ (but the proof generalizes to the following cases), that
\begin{equation}\label{lewinliebprove}
\left\{\begin{array}{lr}
       C_{\mathrm{Jel}}^{\mathrm{lattice}}(s,d) = C_{\mathrm{UEG}}^{\mathrm{lattice}}(s,d)&\text{ for }d-2<s<d,\\[3mm]
       C_{\mathrm{Jel}}^{\mathrm{lattice}}(d-2,d) = C_{\mathrm{UEG}}^{\mathrm{lattice}}(d-2,d) - \mathsf{gap}(d)&\text{ for }s=d-2,
       \end{array}
\right.
\end{equation}
where $\mathsf{gap}(d)>0$ is a constant that can be computed explicitly if one knows and uses the conjectured minimizers for the problems in the second line of \eqref{lewinliebprove} (which is the case for $d=3$, but not for the high dimensions as explained in Remark \ref{Cristconj} above). This in particular implies that the arguments adopted in \cite[App. B]{lewinlieb} may even not have any connection in high dimensions to the comparison of the constants $C_{\mathrm{UEG}}(s,d)$ and $C_{\mathrm{Jel}}(s,d)$.
\item It is unknown for all $d\ge2, 0<s<d$, whether or not the following is true:
\begin{equation}\label{ueggapnogap}
C_{\mathrm{UEG}}(s,d) \stackrel{???}{=} C_{\mathrm{UEG}}^{\mathrm{lattice}}(s,d).
\end{equation}
Note that the minimization problem defining $C_{\mathrm{UEG}}^{\mathrm{lattice}}(s,d)$ is rather different than the original one defining $C_{\mathrm{UEG}}(s,d)$ because the constraint $\gamma_N\mapsto\frac{1_\Omega}{N}$ as in \eqref{reformulcueg00} might possibly be favouring non-lattice configurations, especially in cases where boundary terms become important such as for $s=d-2$.

\end{enumerate}

\subsubsection{The role of boundary effects}\label{boundaryeffects}

\textit{Existence of noncrystalline minimizers:} For $s=1,d=3$, our Main Theorem together with the above considerations implies that, either the BCC-lattice is not asymptotically close to minimizing $E_{\mathrm{Jel},s}(\Omega)$, or for $\gamma_N$ minimizing \eqref{reformulcueg00}, the support of $\gamma_N$ contains far-from-crystalline configurations. 

In other words, if we believe that for $s=1,d=3$ the $E_{\mathrm{Jel},s}$-minimizing configurations are asymptotically crystalline, then minimizers of the \eqref{reformulcueg00} are far from lattice-like in terms of energy, due to \emph{boundary effects}. 

\textit{Boundary effects:} That boundary effects are the culprit of the above non-crystalline behavior can be checked by noticing the following fact. Consider the minimization of $E_{\mathrm{Jel},s}$ or of \eqref{reformulcueg00} on sequences of homogeneous spaces $X_N$ of volume $N$ such as (a) the torii $X_N=\mathbb T_{\Lambda,N}:=\R^d/(R\Lambda)$ for $\Lambda\subset \R^d$ a unimodular lattice and $R^d=N$, or (b) the volume-$N$ round spheres $X_N=\mathbb S^d_N$, or (c) volume-$N$ rescalings of projective spaces $X_N=\mathbb CP^d_N$ or $X_N=\R P^d_N$. Then we can use the compactness and homogeneity of the spaces $X_N$ in order to find $\gamma_N\mapsto \mathcal L^d_{X_N}/N$, which are minimizers of \eqref{reformulcueg00} supported on minimizers of $E_{\mathrm{Jel},s}$. 

Note that in this situation, the domains differ from cubes $K$ in the essential property of having boundary points. Charges near boundary points of $K$ essentially interact only with points on one side of the boundary, whereas no such anisotropic points appear in homogeneous spaces.

In the case of homogeneous spaces, we can construct $\gamma_N$ directly as a superposition of Jellium-minimizing configurations, and therefore minimizers of the Jellium and Uniform Electron Gas problems don't differ as for the problems on $K_R$ from Theorem \ref{thnextorder}. For example, in the case of torii $\mathbb T_{\Lambda,N}$ for $\Lambda=\mathbb Z^d$, we could just use the translations, and define $\gamma_N$ as a superposition of translations of a fixed $E_\mathrm{Jel}$-minimizing configuration: in this case it is plausible that both problems on the torus produce asymptotically lattice-like minimizers, as the boundary effects disappear, in striking contrast to the non-periodic case of cubes $K_R$ considered here. 

We thank Doug Hardin, Mathieu Lewin and Ed Saff for mentioning to us this last observation about flat torii on separate occasions, and we point out work \cite{hsss}, related to the investigation of Jellium energies on flat torii.

We also note that, as discussed in \cite[App. B]{lewinlieb}, numerous Wigner-type problems considered in the literature involve nontrivial boundary effects. This has to be kept in mind when comparing our continuity result for $C_\mathrm{Jel}(d,s)$ with other such problems. For example, doing the analytic continuation in $s$ of the Epstein Zeta function, as done for example in $d=3$ in \cite{BBS}, produces a residue at $s=d-2$, and a discontinuity of the Zeta-function extension at this value. Again this is precisely due to the relevance of boundary effects, which produces an essential difference between different formulations.

\subsection{Link to the Lieb-Oxford bound}

Let $d=3,s=1$. The functional
\[
E_{N} ^{xc}[\mu]:=\mathcal F_{N, 1}(\mu)-N^2 \mathcal I_1(\mu), ~\mu\in \mathcal P(\mathbb R^3),
\]
is called \textit{exchange-correlation energy} and appears in numerical calculations of the electronic structure that use Density-Functional Theory (DFT). Amongst probability measures $\mu$ such that $d\mu(x)=\rho(x)dx$ for $\sqrt\rho\in H^1(\R^3)$,  $E_{N} ^{xc}[\rho]:=E_N^{xc}[\mu]$ has a universal lower bound, given, as shown in \cite{LO}, by
\[
E_{N} ^{xc}[\mu]\ge -\textsf{C}_\mathrm{LO}N^{4/3} \int\rho^{4/3}(x) dx,
\]
where \textit{the Lieb-Oxford constant} is given by 
\[
\textsf{C}_\mathrm{LO}:= \lim_{N\to\infty} \sup\left\{\frac{ - E_{N} ^{xc}[\mu] } { N^{4/3} \int
\rho^{4/3}(x) dx } : \mu \in  \mathcal P(\mathbb R^3), \sqrt\rho \in H^1(\R^3) \right\}>0.
\]
The exact value of $\textsf{C}_\mathrm{LO}$, while still not known, has been the subject of great attention due to its wide-spread use in the construction of approximate exchange-correlation functionals. It has been recently argued by means of physical arguments that $\textsf{C}_\mathrm{LO}\le 1.4442$ in \cite{RPCP} and it has long been conjectured that $\textsf{C}_\mathrm{LO}\approx 1.4442$. It is easy to see that we have 
\[
-\textsf{C}_\mathrm{UEG}(1,3)\le \textsf{C}_\mathrm{LO},
\]
which, if we assume the crystallization conjecture, implies in view of our Main Theorem that in particular $\textsf{C}_\mathrm{LO}\ge -\textsf{C}_\mathrm{Jel}(1,3)\approx 1.4442$. 

Note that has also been conjectured in \cite{OC}, \cite{RPCP}, that $-\textsf{C}_\mathrm{UEG}(1,3)=\textsf{C}_\mathrm{LO}$. We observe now that for the case of the Uniform Electron Gas, when $\mu$ is a uniform measure with density of form $1_{\Omega}/|\Omega|$ (see also (\ref{reformulcueg00}) below), and in view of Proposition 2.3 from \cite{cotpet} (as stated also in Proposition \ref{subadd3} below) which gives in this case via a standard subadditivity theorem that
$$ \lim_{N\to\infty} \frac{ E_{N} ^{xc}[\mu] } { N^{4/3} \int
\rho^{4/3}(x) dx}=\inf_N \frac{ E_{N} ^{xc}[\mu] } { N^{4/3} \int
\rho^{4/3}(x) dx},$$ 
one immediately obtains 
$$\textsf{C}_\mathrm{Jel}(1,3)=\textsf{C}_\mathrm{UEG}(1,3)=-\textsf{C}_\mathrm{LO}.$$

\subsection{Strategy of the proofs and plan of the paper}

\textbf{Strategy of the proofs} 

In Section \ref{recapnot}, we give some main definitions and provide some known results on the Coulomb and Riesz gases problem for the exponents $\max\{0,d-2\}\le s<d$; even though most of the results in this section are known and standard for the Coulomb and Riesz gases community, they may be new to the optimal transport, optimization, and computational chemistry communities, who may be interested in the main results of our paper, hence our introducing these notions in some detail. In Section \ref{reinterpret},  which presents a series of novel results in the area, we reinterpret $\textsf{C}_{\mathrm{Jel}}(s,d)$ and $\textsf{C}_{\mathrm{UEG}}(s,d)$ in terms of the explicit Jellium energy $E_{\mathrm{Jel},s}$, respectively of the Uniform Electron Gas energy $E_{\mathrm{UEG},s}$, both first introduced in Section \ref{jelueg00}. This allows us in particular to equate later on in (\ref{reformulc1}) the next-order term constants for (\ref{nextordergas}) and (\ref{jelueg001}), which equality of constants seems to not have been previously mentioned nor proved in the literature.

The proof of equality of constants $\textsf{C}_{\mathrm{Jel}}(s,d)=\textsf{C}_{\mathrm{UEG}}(s,d)$ will be done in two steps. We will first prove it for $0\le d-2<s<d$, and then later on, we will use this result, together with the continuity in $s$ of $\textsf{C}_{\mathrm{UEG}}(s,d)$ shown in Proposition \ref{continuityc2}, to extend the proof to the crucial $s=d-2$ case.

In Section \ref{framework}, we establish the main steps for proving that $\textsf{C}_{\mathrm{Jel}}(s,d)=\textsf{C}_{\mathrm{UEG}}(s,d)$ for $0\le d-2<s<d$. To begin with, in Section \ref{easyineq} we prove the easier direction of the inequality $\textsf{C}_{\mathrm{Jel}}(s,d)\le\textsf{C}_{\mathrm{UEG}}(s,d)$. In Section \ref{gapbound0} we introduce the framework for proving the harder direction $\textsf{C}_{\mathrm{Jel}}(s,d)\ge\textsf{C}_{\mathrm{UEG}}(s,d)$ by first establishing in Lemma \ref{comparisoneindejel} a series of \textit{key} comparison equalities holding for $d-2<s<d$ (but crucially not for $s=d-2$) between the Jellium energy $E_{\mathrm{Jel},s}$ and the Uniform Electron Gas energy $E_{\mathrm{UEG},s}$, for measures which are constructed from competitors to the Jellium problem. To the best of our knowledge these equalities are new and of independent interest. Next comes the novel inequality comparison between the large $N$ limit of the scaled by $N$ Uniform Electron Gas energy $E_{\mathrm{UEG},s}$ and the corresponding limit of the scaled by $N$ next-order term for an optimal transport problem with "the wrong marginal", constructed from competitors to the Jellium problem which have certain helpful properties. Furthermore, this "wrong marginal" is constructed in such a way such that it is close enough to the "right marginal", which has density of form $1_\Omega/|\Omega|$, where $|\Omega|=N$. Since the marginal of this optimal transport problem with "wrong marginal" depends in a non-trivial way on $N$, we cannot apply directly to it the results from Theorem \ref{leading term}. To circumvent this issue, by employing similar subadditivity arguments as those from \cite{cotpet} and making use of the closeness of the "wrong" marginal and the "right marginal", we obtain in Lemma \ref{boundbelowmuR1} from Section \ref{reinstate0} an equality between the large $N$ limit of the scaled by $N$ next-order terms for their corresponding OT problems.  Finally, to the OT problem with the "right marginal" we can then apply Theorem \ref{leading term}, and conclude that it has the next-order constant $\textsf{C}_{\mathrm{UEG}}(s,d)$.
 
In Section \ref{contsec} we prove the continuity of the map $s\rightarrow \mathsf C_{\mathrm{UEG}}(s,d), 0<s<d$. The proof is based on the Moore-Osgood Theorem of interchanging the double limits between $N$ and $s$ for $\lim_{s\to s_0}\lim_{N\to\infty}{E}^{\mathrm{xc}}_{N,s}(\mu)/N^{1+s/d}$, and requires two tools: continuity in $s\in I$ of ${E}^{\mathrm{xc}}_{N,s}(\mu)/N^{1+s/d}$ at fixed $N$ (proved  in Lemma \ref{convot} below), and uniform convergence in $N\to\infty$ of ${E}^{\mathrm{xc}}_{N,s}(\mu)/N^{1+s/d}$ with respect to the parameter $s\in (0,d)$, shown in Corollary 5.1 from \cite{cotpet}. Since for $d\ge2$ we know by our Main Theorem that $\mathsf C_{\mathrm{Jel}}(s,d)=\mathsf C_{\mathrm{UEG}}(s,d)$ for $s\in(d-2,d)$, this translates the question of whether or not $\mathsf C_\mathrm{Jel}(d-2,d)=\mathsf C_\mathrm{UEG}(d-2,d)$ to the question of whether or not $\mathsf C_\mathrm{Jel}(s,d)$ has a discontinuity at $s=d-2$.

In Section \ref{contsecjel} we prove the Main Theorem for the delicate Coulomb case $s=d-2$. The proof combines the result of the Main Theorem for $0\le d-2<s<d$, with the continuity result of Proposition \ref{continuityc2}, and with a new method, introduced and proved in Lemma \ref{subaddjell} below, which gives an almost-subadditivity formula for the minimum Jellium energy. Lemma \ref{subaddjell} is proved by means of an application of the Fefferman-Gregg decomposition introduced in Proposition 1.6 from \cite{cotpet}, coupled with exploiting the point separation of Jellium minimizers, which crucially holds uniformly in $s$ for all $d-2\le s<d$, and it is proved in Lemma \ref{unifbound} from Appendix \ref{usefulprop} below. Furthermore, this almost-subadditivity formula allows to develop in separate work a novel and robust alternative method to the screening method of Sandier and Serfaty \cite{SandSerbook}, \cite{ss1d}, \cite{ss2d}, of deriving the optimal next-order upper bound for the Jellium and Coulomb and Riesz gas problems.

\textbf{Possible alternative proof strategies} 

A possible alternative strategy to prove our main result would be to proceed through a treatment of the periodic problem on a sequence bigger and bigger torii such as $\mathbb R^d/(N^{1/d}\mathbb Z)^d$. In case of the Jellium problem this sequence of periodic problems could be connected to the problem on the cube $[-N^{1/d}/2, N^{1/d}/2]^d$, and thus to the constant $\mathsf C_{\mathrm{Jel}}(s,d)$ from \cite[Prop. 1.5]{PS}. This connection however is not immediate, requiring considerable work due to the presence of such boundary effects as described in Section \ref{boundaryeffects}. Moreover, the same link in the case of the Optimal Transport problem from the definition of $\mathsf C_{\mathrm{UEG}}(s,d)$ would require a similar effort. Thus the current approach seems to be shorter and more direct. 

More precisely, the main issue we face is that the distance between points in a periodic problem is different than distance on a cube. This was dealt by PDE-based screening results in \cite{PS}. In the case of the Optimal Transport $\mathcal F_{N,s}$-problems a comparison between the two types of distances could go through a substitute of the screening for that problem, i.e. we could repeat most of the steps from \cite{cotpet} in the periodic case, and we could use the new type of screening phenomena resulting from the Fefferman-Gregg decomposition and localization results. As said above, however, this way to the main result seems however to be longer and more technical.

\section{The Jellium energy and the constant $\textsf{C}_{\mathrm{Jel}}(s,d)$}
\label{Jelintro}

\subsection{Setting and known results}

\label{recapnot}

In this section we recall the definition of, and the notations pertinent to, the functional $\mathcal W$ appearing in Theorem \ref{thnextorder}. We also recall here some related known theorems, which will be used later in the proofs. We consider the full range of exponents $\max\{0,d-2\}\le s<d$ done in a unified framework, as previously done in \cite{PS}. Later on, in Lemma \ref{comparisoneindejel} and in Section \ref{conclstrict}, we will need to restrict to the $0\le d-2<s<d$ case, and come back to using the results for $s=d-2$ in Section \ref{contsecjel}.

\subsubsection{Extension to $\mathbb R^{d+1}$}\label{sec:extension}

Note that the function $\textsf{c}$ in \eqref{value c} is the fundamental solution for the operator $(-\Delta)^{\frac{d-s}{2}}$ on $\mathbb R^d$, in the cases $0<d-2\le s<d$. For $\max\{0,d-2\}<s<d$, $(-\Delta)^{\frac{d-s}{2}}$ is a nonlocal operator, which does not allow to directly gain energy control on the solutions of the form $f(x)=(\mu*\textsf c)(x)$ for $\mu\in\mathcal M(\mathbb R^d)$, just by studying the associated boundary value problem on subdomains.  As originally noticed by Caffarelli and Silvestre \cite{caffarelli2007extension} and used in \cite{PS}, if $d-2<s<d$ and $s>0$, then we may add one space variable $y\in \mathbb R$ to the space $\mathbb R^d$, and consider the local but inhomogeneous operator of the form $\mathrm{div}(\yg \nabla\cdot)$ with the weight
\begin{equation}\label{choicegamma}
\gamma=s-d+2-k\in]-1,1[,
\end{equation}
where for $s=d-2$ we use $k=1$ and $\gamma=0$, while for $s>d-2$ we use $k=1$ and $\gamma$ satisfying (\ref{choicegamma}). Thus, in either case we extend our space from $\R^d$ to $\R^{d+k}$. Points in the extended space $\mathbb R^{d+k}$ will be denoted by $X$, with $X=(x,y)$, $x\in \mathbb R^d$ and $y\in \mathbb R^k$, and by $\textsf{c}(X)=|X|^{-s}$ the same power-like potential in the extended coordinates. 
Denote by 
\begin{equation}\label{iota}
\iota:\mathbb R^d\to\mathbb R^{d+k},\quad \iota(x):=(x,0),\quad\bar\mu:=\iota_\#\mu\quad\text{ for }\mu\in\mathcal M(\mathbb R^d).
\end{equation}
Also note that if $\bar\mu$ is the unique measure on $\mathbb R^d\times\{0\}^k$ such that for $\pi_d:\mathbb R^{d+k}\to\mathbb R^d$ given by $\pi_d(x,y):=x$, there holds $(\pi_d)_\#\bar\mu=\mu$. Then we have the following useful representation:
\begin{equation}\label{fundsolrd+1}
H^{\bar\mu}(X):=\int_{\mathbb R^d}\frac{1}{|X-(x',0)|^s}d\mu(x')\quad \text{ then}\quad-\mathrm{div}(|y|^\gamma \nabla H^{{\bar\mu}})=c_{d,s}\bar\mu,
\end{equation}
where
\begin{equation}\label{cds}
\quad c_{d,s}:=\left\{\begin{array}{ll}2s\frac{2\pi^{\frac{d}2}\Gamma\left(\frac{s+2-d}2\right)}{\Gamma\left(\frac{s+2}2\right)}&\mbox{ for }s>d-2\ge 0,\\
(d-2)\frac{2\pi^{\frac{d}2}}{\Gamma(d/2)}&\mbox{ for }s=d-2>0,\\
2\pi&\mbox{ for }s=0, d=2.
\end{array}\right.
\end{equation}
\subsubsection{Point configurations}
It will be convenient to identify \textit{point configurations} in $\mathbb R^d$ either with integer-valued positive atomic Radon measures, which will be denoted by $\nu\in\mathcal M^+(\mathbb R^d)$, (or equivalently, locally finite sums of Dirac masses), or with multisets $\textsf{set}(\nu)\subset\mathbb R^d$ of locally finite cardinality. The two notations are linked by the formula
\begin{equation}\label{configurations}
 \nu = \sum_{p\in\textsf{set}(\nu)}\delta_p.
\end{equation}
Therefore if $\nu$ is a point configuration, in the first interpretation it makes sense to write $\nu(A)$ for $A\subset\mathbb R^d$ a Borel set, and then this is an integer representing the number of points $\#(\textsf{set}(\nu)\cap A)$ in the second interpretation. For $N$-point configurations we use the following notations: 
\begin{equation}\label{pointconfig}
 \vec x=(x_1,\ldots,x_N)\in(\mathbb R^d)^N, \quad \nu_{\vec x}:=\sum_{i=1}^N \delta_{x_i}.
\end{equation}
We will denote the set of all configurations $\nu$ defined as above by $\textsf{Config}$, and the set of configurations with support in a Borel set $A\subset\mathbb R^d$ by $\textsf{Config}(A)$.

\medskip

Note also that $\nu\in\textsf{Config}(A)$ directly implies, with notation \eqref{iota}, that $\bar\nu\in\textsf{Config}(A\times\{0\})\subset\textsf{Config}(A\times\mathbb R)$.

\medskip

If we denote by $\tau_p$ the translation by a vector $p\in\mathbb R^d$, i.e. 
\begin{equation}\label{tau_p}
 \tau_p(x):= x+p
\end{equation}
then the $p$-translated configurations are just defined by pushforward $(\tau_p)_\#\nu$, which is defined by requiring that, for all test functions $f\in C^0_c(\mathbb R^d)$, there holds $\int f\ \diff((\tau_p)_\#\nu) =\int f\circ\tau_p \diff\nu$.
\subsubsection{Compatible and screened electric fields}
Fixing $\nu\in\textsf{Config}$, we say that a vector field $E\in L_{loc}^p(\mathbb R^{d+1}, \mathbb R^{d+1})$ with $p\in(1,2)$ is compatible with $\nu$ (with neutralizing background charge uniform and of intensity $1$) if, with the same notations as in \eqref{fundsolrd+1} and \eqref{choicegamma}, there holds, recalling the notation \eqref{iota}, and where $\mathcal L^d$ is the Lebesgue measure on $\mathbb R^d$,
\begin{equation}\label{compatible}
 -\mathrm{div}(|y|^\gamma E) =c_{s,d}\left(\bar\nu - \bar{\mathcal L^d}\right),
\end{equation}
where in the above we denoted by $\bar{\mathcal L^d}:=\iota_\# \mathcal L^d$. We note that in this case we must have $p<p_{\rm{max}}:=\min\left\{\frac2{\gamma+1}, \frac{d+1}{s+1}\right\}$. For a general $\nu$ there may not exist a compatible $E$, and if a compatible $E$ exists then it is never unique, because for any field $E_0$ such that $-\mathrm{div}(|y|^\gamma E_0) =0$, we see that $E+E_0$ is still a solution to \eqref{compatible}, and we may always find choices of $E_0$ such that this field will still be in $L^p_{loc}$ as well.  

\medskip

For $K\subset\mathbb R^d$ a closed set we also say that \textit{$E$ is compatible with $\nu$ in $K$}, or equivalently, we write $E\in \mathsf{Comp}_{K,\nu}$ (with the further abbreviation $E\in\mathsf{Comp}_\nu$ meaning $E\in\mathsf{Comp}_{\mathbb R^d,\nu}$), if 
\begin{equation}\label{compatibleinK}
 -\mathrm{div}(|y|^\gamma E) = c_{s,d}\left(\bar\nu - \bar{\mathcal L}^d\right) \text{ in }K\times \mathbb R,
\end{equation}
and we say that $E$ is \emph{screened on $K$} (or equivalently, we write $E\in \textsf{Scr}_K$) if 
\begin{equation}\label{screenK}
E\cdot\vec n = 0 \text{ on }\partial K\times\mathbb R\text{ where }\vec n\text{ is the outer normal}.
\end{equation}
Note that if $E\in\mathsf{Comp}_{K,\nu}$ and $E\in\mathsf{Scr}_K$, then automatically $\nu(K)=\bar\nu(K\times\mathbb R)=|K|$ (in particular $|K|\in \N$ in this case), because then $E$ satisfies \eqref{compatibleinK}, and on the other hand the condition \eqref{screenK} together with Stokes' theorem implies that the total divergence of $|y|^\gamma E$ in $K\times \mathbb R^k$ must be zero. Therefore a \textit{necessary condition} for $E\in\mathsf{Comp}_{K,\nu}$ and $E\in\mathsf{Scr}_K$, to contemporarily hold is that $\nu(K)=|K|$. Furthermore note that $E\in \mathsf{Comp}_{K,\nu}$ and $E\in \mathsf{Scr}_K$ imply that if we extend $E$ by the zero vector field outside $K$ (denoting the result still by $E$ by abuse of notation), then we have $E\in\mathsf{Comp}_\nu$ as well. We emphasise that charges have to precisely balance only if we want to be able to screen them. If no screening is required then it is not necessary to have integer total mass for the \textit{background absolutely continuous charge} which here is $1_K(x)\diff x$ (whose total mass is then $|K|$).

\medskip
While we do not know an explicit alternative condition on $\nu$ ensuring the existence of $L^p_{loc}$ fields $E$ that satisfy \eqref{compatible}, in \cite{PS}, \cite{PRN} and \cite{lebleserfaty}, there are several tools for constructing and manipulating solutions to \eqref{compatibleinK} and \eqref{screenK}, some of which will be recalled and used later in the paper.
We also will be helped in our stuy by just considering periodic configurations and electric fields. Therefore we introduce the following notations, in which we assume $r>0$:
\begin{equation}\label{periodicfields}
 \textsf{Per}_{r}:=\left\{E\in L^p(\mathbb R^{d+1}, \mathbb R^{d+1}):\ \exists \nu\in\textsf{Config}, E\in\mathsf{Comp}_\nu, \forall v\in(r\mathbb Z)^d, E\circ \tau_{v} = E\right\}.
\end{equation}
Note that if $E\in \mathsf{Per}_r$ and $\nu\in\mathsf{Config}$ is such that $E\in\mathsf{Comp}_\nu$ then automatically $\nu$ is $(r\mathbb Z)^d$-periodic as well. By further imposing that the vector fields are screened on a periodicity cell and that the configuration has zero barycenter on that cell (in the sense below), we obtain the more restricted class
\begin{equation}\label{perscreenbar}
 \textsf{Per}_{r}^{scr,0}:=\textsf{Per}_{r}\cap\left\{E:\ \left.\begin{array}{l}\exists v\in\mathbb R^d, E\in\textsf{Scr}_{K_{r}+v},\\[3mm]\exists\nu\in\mathsf{Config}\text{ such that }\int_{K_{r}+v} x d\nu(x)=0,\text{ and }E\in\mathsf{Comp}_\nu,\\[3mm] \min_{p\in\mathsf{set}(\nu)}\mathrm{dist}(p, \partial K_r+v)>0\end{array}\right.\right\}.
\end{equation}
\subsubsection{Truncation}\label{sec:truncation}
In order to control the blow-up behavior of $\textsf{c}$ we need the following notation for the kernel truncation introduced in \cite{PS}:
\begin{equation}\label{truncation}
 \textsf{c}_\eta(X):=\min\left\{\textsf{c}(X), \textsf{c}(\eta)\right\}, \quad \textsf{f}_\eta:= \textsf{c} - \textsf c_\eta, \quad \delta_0^{(\eta)}:=-\frac{1}{c_{s,d}}\mathrm{div}\left(|y|^\gamma \textsf c_\eta\right).
\end{equation}
For $\eta\in(0,1)$ this allows to truncate also compatible electric fields: if $E$ satisfies \eqref{compatible} and $E_K$ satisfies \eqref{compatibleinK} then we define
\begin{equation}\label{truncE}
 E_\eta(X):=E(X) - \sum_{x\in \nu}\nabla \textsf f_\eta(X-(x,0)),\quad E_{K,\eta}(X):=E_K(X) - \sum_{x\in \nu\cap K}\nabla \textsf f_\eta(X-(x,0)).
\end{equation}
\subsubsection{Definition of $\mathcal W$ and of $\textsf C_{\mathrm{Jel}}(s,d)$}
If $0\le d-2<s<d$ or $s=d-2>0$, if $E\in\mathsf{Comp}_\nu$ for some $\nu\in\mathsf{Config}$, and using the notation
\begin{equation}
K_R:=\left[-\frac{R}2, \frac{R}2\right)^d,
\end{equation}
we define the renormalized energy of $E$, denoted $\mathcal W(E)$, as follows:
\begin{equation}\label{weta}
 \mathcal W_\eta(E):=\frac{1}{c_{s,d}}\limsup_{R\to\infty} \left(\frac{1}{R^d}\int_{K_R\times\mathbb R}|y|^\gamma |E_\eta(X)|^2\diff X - c_{s,d}\textsf c(\eta)\right),\qquad \mathcal W(E):=\lim_{\eta\to 0}\mathcal W_\eta(E).
\end{equation}
From now on, the space $\mathbb R^{d+1}$ will always have coordinates $X=(x,y), x\in\mathbb R^d, y\in\mathbb R$. Moreover the integrals over subdomains of $\mathbb R^{d+1}$ will be always against the Lebesgue measure, unless otherwise specified, and the volume element ``$\diff x$'' will be omitted in the notation. 

\medskip

By minimizing over possible choices of $E\in\mathsf{Comp}_\nu$ at fixed $\nu\in\mathsf{Config}$, we may define the following functional $\mathbb W$:
\begin{equation}\label{w}
 \mathbb W(\nu):=\frac{1}{c_{s,d}}\inf\left\{\mathcal W(E): E\text{ satisfies \eqref{compatible}}\right\},
\end{equation}
with the convention $\inf\emptyset=+\infty$. If finite, the infimum in \eqref{w} is uniquely achieved (see \cite[Lem. 2.3]{lebleserfaty}). Then the constant $\mathsf C_{\mathrm{Jel}}(s,d)$ figuring in \eqref{nextordergas} is then characterized in \cite{PS} as follows:
 \begin{equation}\label{defc1}
 \textsf{C}_{\mathrm{Jel}}(s,d)= \min_{\nu\in\mathsf{Config}}\mathbb W(\nu)
 \end{equation}
More explicitly, using the above definitions \eqref{weta}, \eqref{w} together with \eqref{defc1}, we also have
\begin{equation}\label{defc1_long}
\textsf C_{\mathrm{Jel}}(s,d)={\frac{1}{c_{s,d}}}\min_{\substack{\nu\in\mathsf{Config}\\ E\in\mathsf{Comp}_\nu}}\lim_{\eta\to 0}\left(\lim_{R\to\infty}\frac{1}{|K_R|}\int_{K_R\times \mathbb R}|y|^\gamma |E_\eta|^2 - c_{s,d}\textsf{c}(\eta)\right).
\end{equation}
Note that in \cite{PS} the above definition is extended to the slightly larger range $\max\{0,d-2\}\le s<d$, however we restrict to the cases $0\le d-2<s<d$, which are the ones of interest here.
\subsubsection{Some useful results concerning the Jellium problem}
In this subsection we recall some useful lemmas which will help us simplify the definition of $\mathsf{C}_{\mathrm{Jel}}(s,d)$. Note that the configurations that we consider will have an integer number of charges, therefore we will work everywhere in the paper under the constraint $R^d\in\mathbb{N}$, $R_1^d\in\mathbb{N}$, and also where appropriate, with $(R_1/2)^d\in\mathbb{N}$. 
The first result is a simplified restatement of \cite[Prop. 6.1]{PS}, which describes the procedure of modifying a couple $(\nu,E)$ such that $\nu\in\mathsf{Config}, E\in\mathsf{Comp}_{{K_R},\nu}$, over a cube $K_R$ and near $\partial K_R$, to a couple $(\nu^{\mathrm{scr}}_R, E^{\mathrm{scr}}_R)$, such that $E^{\mathrm{scr}}_R$ satisfies ($\mathsf{Scr}_{K_R}$), by controlling the new energy introduced in the procedure:
\begin{proposition}[Screening, {\cite[Prop. 6.1]{PS}}]\label{screening}
Assume $\max\{0,d-2\}\le s<d$. We consider $R$ such that $R^d=N\in\mathbb N$, and denote $K_R:=[-R/2,R/2]^d$. Consider a constant $0<\epsilon<1/4$ and define $\check K_R:=[-(1-\epsilon)R/2, (1-\epsilon)R/2]^d$. There exist constants $\eta_0,C>0$, depending only on $s,d$, such that for all $\nu\in\mathsf{Config}$ and $E\in \mathsf{Comp}_{{K_R,}\nu}$ if we denote
\begin{equation}\label{Me}
M:=\frac1{R^d}\int_{\check K_R\times[-\epsilon^2R,\epsilon^2R]}|y|^\gamma|E_\eta|^2,\quad e:=\frac1{\epsilon^4R^d}\int_{\check K_R\times\left(\mathbb R\setminus [-\epsilon^2R/2, \epsilon^2R/2]\right)}|y|^\gamma|E_\eta|^2,
\end{equation}
then there exists $\underline R$ depending only on $s,d,M,e,\epsilon$, such that if $R>\underline R$ then there exists $\nu^{\mathrm{scr}}_R\in\mathsf{Config}$ and $E^{\mathsf{scr}}_R\in\mathsf{Comp}_{{K_R},\nu^{\mathrm{scr}}_R}$ such that the following hold: %
\begin{itemize}
\item $E^{\mathrm{scr}}_R\in\mathsf{Scr}_{K_R}$ and $\nu^{\mathrm{scr}}_R(K_R)=
|K_R|=N$.
\item $E^{\mathrm{scr}}_R=E$ on $\check K_R\times\mathbb R$ and $\left.\nu^{\mathrm{scr}}_R\right|_{\check K_R}=\left.\nu\right|_{\check K_R}$. 
\item The minimum distance between different charges belonging to $\mathsf{set}\left(\left.\nu^{\mathrm{scr}}_R\right|_{K_R\setminus\check K_R}\right)$ as well as the minimum distance between such charges and $\partial K_R$ are bounded below by $\eta_0$.
\item For $0<\eta\le \eta_0$ there holds
\begin{equation}\label{energyscr}
\int_{K_R\times\mathbb R}|y|^\gamma|E^{\mathrm{scr}}_{R,\eta}|^2\leq (1+C\epsilon)\int_{\check K_R\times\mathbb R}|y|^\gamma|E_\eta|^2 + C(1+M)\epsilon R^d\mathsf{c}(\eta) + Ce \epsilon R^d.
\end{equation}
\end{itemize}
\end{proposition}

We will also use later the following fundamental elementary result taken from \cite[Lem. 3.10]{lebleserfaty}, which can be summarized by saying that ``imposing the $\mathsf{Scr}_K$-constraint increases the energy'':
\begin{lemma}[{\cite[Lem. 3.10]{lebleserfaty}}]\label{leblemma}
Let $K\subset\mathbb R^d$ be compact with piecewise $C^1$ boundary, denote by $d\mu_K(x):=1_K(x)\diff x$ and assume that $E\in\mathsf{Scr}_K$, and that $E\in \mathsf{Comp}_{{K,}\nu}$ with $\nu\in\mathsf{Config}(K)$.

Then for any $\eta\in(0,1)$, with the notation \eqref{truncE} used both for $E$ and for $\nabla H^{\bar\nu-\bar\mu_K}$ defined via \eqref{fundsolrd+1}, we have
\begin{equation}\label{localsmall}
\int_{\mathbb R^{d+1}}|y|^\gamma|\nabla H^{\bar\nu -\bar\mu_K}_\eta|^2\le \int_{K\times\mathbb R}|y|^\gamma|E_\eta|^2. 
\end{equation}
\end{lemma}
Furthermore we recall the following result from \cite[Lem. 7.1]{PS}, which will allow to control the $e$-term defined in \eqref{Me} and which appears in \eqref{energyscr}:
\begin{lemma}[Decay of minimizers {\cite[Lem. 7.1]{PS}}]\label{decayy}
If $\nu\in\mathsf{Config}$ is a minimizer of $\mathbb W$ and $E\in\mathsf{Comp}_\nu$ is a minimizer of $\mathcal W(\nu)$ then the following holds:
\begin{equation}\label{eqdecayy}
\lim_{t\to\infty}\lim_{R\to\infty}\frac1{R^d}\int_{K_R\times(\mathbb R\setminus[-t,t])}|y|^\gamma|E|^2 =0.
\end{equation}
\end{lemma}
As a technical tool in the proof below, we also note the following well-known result:
\begin{lemma}[Gluing vector fields without creating distributional divergence]\label{gluinglemma}
Let $V^+,V^-$ be $L^p$-vector fields on $\mathbb R^{d+1}$ with $p>1$, let $\Omega^\pm\subset\mathbb R^{d+1}$ be two domains such that $\partial\Omega^+\cap \partial \Omega^-=\Sigma$ is a submanifold of $\mathbb R^{d+1}$, and define
\[
V:=1_{\Omega^+}V^+ + 1_{\Omega^-}V^-.
\]
Let $D_\pm:=\mathrm{div}V^{\pm}$ be the distributional divergences of $V^{\pm}$ and assume that there exists an open neighborhood $U_\Sigma\supset\Sigma$ such that the restrictions $D_\pm|_{U_\Sigma}$ are represented by $L^1_{loc}$-functions. Then
\[
\mathrm{div}V= 1_{\Omega^+}\mathrm{div}V^+ +1_{\Omega^-}\mathrm{div}V^-
\]
holds if and only if $V^+\cdot\vec n=V^-\cdot\vec n$ on $\Sigma$, where $\vec n$ is a normal vector field along $\Sigma$. 
\end{lemma}
The proof of this lemma follows by a classical argument, for example by using the Hodge decomposition of $V^\pm, V$, \cite[Thm. 2.4.14]{schwarz}, and by testing the above equations for the terms coming from this decomposition against compactly supported test functions. Since the arguments are classical, we skip them.

\medskip

The above preparations allow to prove the following result, which is a more detailed and explicit version of \cite[Prop. 1.4, point 4.]{PS} which we need below in the proof of Lemma \ref{reformulc1c2}. Indeed, we construct below the reflected fields explicitly and we express how the contributions from a neighborhood of the boundary do not affect the energy. The proof of \cite[Prop. 1.4, point 4.]{PS} instead was referring to an analogy to previous works \cite{ss2d} and \cite{rs}, in which again the construction was not performed in full detail, and was not performed for the case $d-2<s<d$ which we include here. We find that for the benefit of the reader it is of interest to write a self-contained proof here. This result says that the global minimum of $\mathcal W$ can be arbitrarily well approximated by the minimization restricted to $(R_1\mathbb Z)^d$-periodic configurations if $R_1$ is chosen large enough.
\begin{lemma}[see also {\cite[Prop. 1.4, point 4]{PS}}]\label{screenedperiodic}
Let either $d\ge 3$ and $d-2\le s<d$, or $d=2$ and $d-2<s<d$. Then
\begin{equation}\label{minperscr}
\mathsf C_{\mathrm{Jel}}(s,d) = {\frac{1}{c_{s,d}}}\lim_{\substack{{R_1}\to\infty\\(R_1/2)^d\in\mathbb{N}}} \min\left\{\mathcal W(E_{R_1}^{\mathrm{scr},0}): {E_{R_1}^{\mathrm{scr},0}}\in\mathsf{Per}_{R_1}^{\mathrm{scr},0}\right\}.
\end{equation}
\end{lemma}
\begin{proof}
The strategy of proof goes as follows. We note that the inequality ``$\leq$'' in \eqref{minperscr} directly follows from the definitions, because the right hand side of \eqref{minperscr} is a minimization over a smaller class than the one on the left hand side. So we need only prove the ``$\geq$'' inequality. Recall that $\mathcal W$ is defined as a $R\to\infty$ limit of averaged energy per unit of volume on cubes, so we can select a large $R$ which gives a cube on which the minimizer of $\mathcal W$ ``almost reaches'' the numerical value of the minimum of $\mathcal W$. Then we modify it near the boundary using Proposition \ref{screening}, in order to produce a screened vector field with ``almost the same'' energy. Then we use a reflection trick to obtain another vector field and another charge configuration, which are still screened, have unchanged almost optimal energy per volume, and furthermore are balanced, as required in the right hand side of \eqref{minperalmost}. During the whole procedure we introduce an error of average energy per unit volume which can be estimated and proved to be arbitrarily small, and this will prove the desired inequality in \eqref{screening}.

\medskip

\textbf{Step 1.} Since the right hand side is a minimization over fewer configurations than the definition of $\mathsf C_{\mathrm{Jel}}(s,d)$ in \eqref{w}, the inequality ``$\le$'' is clear. We will now fix $\epsilon_0>0$ and prove that 
\begin{equation}\label{minperalmost}
\lim_{\substack{R_1\to\infty\atop (R_1/2)^d\in\N}} \min\left\{\mathcal W(E_{R_1}^{\mathrm{scr},0}): E\in\mathsf{Per}_{R_1}^{\mathrm{scr},0}\right\}\le {c_{s,d}}\mathsf C_{\mathrm{Jel}}(s,d)+\epsilon_0.
\end{equation}
If we prove this for arbitrarily small $\epsilon_0>0$ then this concludes the proof. Therefore, keeping in mind \eqref{defc1_long}, there exist $\eta_0>0$ and $R_0=R_{0,\eta_0}$ such that for all $R_1>R_0$ with $(R_1/2)^d\in\mathbb{N}$ and all $0<\eta<\eta_0$ we can find a minimizing configuration $\nu\in\mathsf{Config}(K_{R_1/2})$ on the cube $K_{R_1/2}$, a minimizing field $E\in\mathsf{Comp}_{{K_{R_1/2}},\nu}$ such that there holds 
\begin{equation}\label{almostminc1}
\frac1{|K_{R_1/2}|}\int_{K_{R_1/2}\times\mathbb R}|y|^\gamma|E_\eta|^2-c_{s,d}\mathsf c(\eta)<{c_{s,d}} \mathsf C_{\mathrm{Jel}}(s,d)+\frac{\epsilon_0}{2}.
\end{equation}
Then denote by
\begin{equation}\label{notM}
\frac1{|K_{R_1/2}|}\int_{K_{R_1/2}\times\mathbb R}|y|^\gamma|E_\eta|^2:=\bar M
\end{equation}
and note that by comparing to \eqref{Me} for the case where $R$ is now replaced by $R_1/2$, we find that $M\le\bar M$ because the energy integrand is positive and the domain of integration in \eqref{notM} is larger than the one in the corresponding \eqref{Me}.

\medskip

\textbf{Step 2.} We now fix a small parameter $\epsilon>0$ to be determined later. By Lemma \ref{decayy} applied to $\nu\in\mathsf{Config}(K_{R_1/2})$, up to increasing $R_1$ we can insure that the second quantity in \eqref{Me} is small, say $e\le 1$. Then we apply the Proposition \ref{screening} on the cube $K_{R_1/2}$ and we find ${E^{\mathrm{scr}}_{R_1/2}}\in \mathsf{Scr}_{K_{R_1/2}}$ such that 
\begin{eqnarray}\label{screenedbound1}
\frac1{|K_{R_1/2}|}\int_{K_{R_1/2}\times\mathbb R}|y|^\gamma|E^{\mathrm{scr}}_{R_1/2,\eta}|^2&\le& (1+C\epsilon)\frac1{|K_{R_1/2}|}\int_{K_{R_1/2}\times\mathbb R}|y|^\gamma|E_\eta|^2 + C(1+M)\epsilon \mathsf c(\eta) + C\epsilon\nonumber\\
&=&(1+C\epsilon)\bar M + C(1+M)\epsilon \mathsf c(\eta) + C\epsilon\nonumber\\
&\le&\bar M+C(\mathsf c(\eta)+1)(\bar M +1)\epsilon,
\end{eqnarray}
where in the last inequality we used $M\le \bar M$. If we fix $\epsilon<(2C(\mathsf c(\eta)+1)(\bar M +1))^{-1}\epsilon_0$, then by \eqref{almostminc1} , \eqref{notM} and \eqref{screenedbound1} we find
\begin{equation}\label{screenedbound2}
\frac1{|K_{R_1/2}|}\int_{K_{R_1/2}\times\mathbb R}|y|^\gamma|E^{\mathrm{scr}}_{R_1/2,\eta}|^2 -c_{s,d}\mathsf c(\eta)\le {c_{s,d}}\mathsf C_{\mathrm{Jel}}(s,d) + \epsilon_0.
\end{equation}
%
\begin{figure}[htp]
\centering
\begin{tikzpicture}[scale=0.25]

\xdefinecolor{dblue}{RGB}{0, 0, 155}
\xdefinecolor{lightblue}{RGB}{102, 224, 255}
\xdefinecolor{llightblue}{RGB}{179, 240, 255}
\xdefinecolor{chargg}{RGB}{255, 83, 26}
\xdefinecolor{lchargg}{RGB}{255, 159, 128}
\xdefinecolor{degreen}{RGB}{0,55,0}
\xdefinecolor{egreen}{RGB}{0, 105, 0}
\xdefinecolor{legreen}{RGB}{0, 175, 0}

\newcommand*{\defcoords}{%
\coordinate (q1) at (-4,-4);
\coordinate (q2) at (4,-4);
\coordinate (q3) at (4,4);
\coordinate (q4) at (-4,4);

\coordinate (l1) at (-9,0);
\coordinate (l2) at (9,0);
\coordinate (l3) at (0,-9);
\coordinate (l4) at (0,9);

\coordinate (g1) at (-2,3);
\coordinate (g2) at (-2,-1);
\coordinate (g3) at (-3,-2.5);
\coordinate (g4) at (1,-1.5);
\coordinate (g5) at (1,1.5);
\coordinate (g6) at (3,-0.5);
\coordinate (g7) at (2,-2);

\coordinate (a1) at (-2.5,0.5);
\coordinate (a2) at (-1.5,1.5);
\coordinate (b1) at (2,2);
\coordinate (b2) at (3.5,3.5);
\coordinate (c1) at (1,0.5);
\coordinate (c2) at (2.5,1.5);
\coordinate (d1) at (1,-2.5);
\coordinate (d2) at (3,-3.5);
\coordinate (e1) at (-2,-2);
\coordinate (e2) at (-1,-2.5);
}
\begin{scope}[shift={(-24,0)}]
\defcoords
\path[fill=lightblue] (q1) -- (q2) -- (q3) -- (q4) -- (q1);
\draw [dblue] (q2) node[anchor=south west] {$K_{R_1/2}$};
\draw [degreen] (q3) node[anchor=north west] {$E^{\mathrm{scr}}_{R_1/2}$};

\draw [very thick,->] (l1) -- (l2);
\draw [very thick,->] (l3) -- (l4);

\draw [thick, egreen, ->] (a1)--(a2);
\draw [thick, egreen, ->] (b1)--(b2);
\draw [thick, egreen, ->] (c1)--(c2);
\draw [thick, egreen, ->] (d1)--(d2);
\draw [thick, egreen, ->] (e1)--(e2);

\draw [fill=chargg] (g1) circle [radius=0.2];
\draw [fill=chargg] (g2) circle [radius=0.2];
\draw [fill=chargg] (g3) circle [radius=0.2];
\draw [fill=chargg] (g4) circle [radius=0.2];
\draw [fill=chargg] (g5) circle [radius=0.2];
\draw [fill=chargg] (g6) circle [radius=0.2];
\draw [fill=chargg] (g7) circle [radius=0.2];
\end{scope}

\begin{scope}
\begin{scope}[shift={(4,4)}]
\defcoords
\path[fill=lightblue] (q1) -- (q2) -- (q3) -- (q4) -- (q1);
\draw [dblue] (q1) node[anchor=south east] {$K_{R_1/2}+z$};
\draw [degreen] (q4) node[anchor=north east] {$E^{\mathrm{scr}}_{R_1/2}\circ \tau_{-Z}$};

    \draw [thick, egreen, ->] (a1) -- (a2);
    \draw [thick, egreen, ->] (b1)--(b2);
    \draw [thick, egreen, ->] (c1)--(c2);
    \draw [thick, egreen, ->] (d1)--(d2);
    \draw [thick, egreen, ->] (e1)--(e2);

\draw [fill=chargg] (g1) circle [radius=0.2];
\draw [fill=chargg] (g2) circle [radius=0.2];
\draw [fill=chargg] (g3) circle [radius=0.2];
\draw [fill=chargg] (g4) circle [radius=0.2];
\draw [fill=chargg] (g5) circle [radius=0.2];
\draw [fill=chargg] (g6) circle [radius=0.2];
\draw [fill=chargg] (g7) circle [radius=0.2];
\end{scope}
\defcoords
\draw [very thick,->] (l1) -- (l2);
\draw [very thick,->] (l3) -- (l4);
\end{scope}

\begin{scope}[shift={(24,0)}]


\begin{scope}[shift={(4,4)}]
\defcoords
\path[fill=lightblue] (q1) -- (q2) -- (q3) -- (q4) -- (q1);

    \draw [thick, egreen, ->] (a1) -- (a2);
    \draw [thick, egreen, ->] (b1)--(b2);
    \draw [thick, egreen, ->] (c1)--(c2);
    \draw [thick, egreen, ->] (d1)--(d2);
    \draw [thick, egreen, ->] (e1)--(e2);

\draw [fill=chargg] (g1) circle [radius=0.2];
\draw [fill=chargg] (g2) circle [radius=0.2];
\draw [fill=chargg] (g3) circle [radius=0.2];
\draw [fill=chargg] (g4) circle [radius=0.2];
\draw [fill=chargg] (g5) circle [radius=0.2];
\draw [fill=chargg] (g6) circle [radius=0.2];
\draw [fill=chargg] (g7) circle [radius=0.2];
\end{scope}


\begin{scope}[xscale=-1, yscale=1]
\begin{scope}[shift={(4,4)}]
\defcoords
\path[fill=llightblue] (q1) -- (q2) -- (q3) -- (q4) -- (q1);
\draw [degreen] (q3) node[anchor=north east] {$E^{\mathrm{scr}}_{R_1}$};

    \draw [thick, legreen, ->] (a1) -- (a2);
    \draw [thick, legreen, ->] (b1)--(b2);
    \draw [thick, legreen, ->] (c1)--(c2);
    \draw [thick, legreen, ->] (d1)--(d2);
    \draw [thick, legreen, ->] (e1)--(e2);

\draw [fill=lchargg] (g1) circle [radius=0.2];
\draw [fill=lchargg] (g2) circle [radius=0.2];
\draw [fill=lchargg] (g3) circle [radius=0.2];
\draw [fill=lchargg] (g4) circle [radius=0.2];
\draw [fill=lchargg] (g5) circle [radius=0.2];
\draw [fill=lchargg] (g6) circle [radius=0.2];
\draw [fill=lchargg] (g7) circle [radius=0.2];
\end{scope}
\end{scope}


\begin{scope}[xscale=-1, yscale=-1]
\begin{scope}[shift={(4,4)}]
\defcoords
\path[fill=llightblue] (q1) -- (q2) -- (q3) -- (q4) -- (q1);
\draw [dblue] (q3) node[anchor=south east] {$K_{R_1}$};

    \draw [thick, legreen, ->] (a1) -- (a2);
    \draw [thick, legreen, ->] (b1)--(b2);
    \draw [thick, legreen, ->] (c1)--(c2);
    \draw [thick, legreen, ->] (d1)--(d2);
    \draw [thick, legreen, ->] (e1)--(e2);

\draw [fill=lchargg] (g1) circle [radius=0.2];
\draw [fill=lchargg] (g2) circle [radius=0.2];
\draw [fill=lchargg] (g3) circle [radius=0.2];
\draw [fill=lchargg] (g4) circle [radius=0.2];
\draw [fill=lchargg] (g5) circle [radius=0.2];
\draw [fill=lchargg] (g6) circle [radius=0.2];
\draw [fill=lchargg] (g7) circle [radius=0.2];
\end{scope}
\end{scope}


\begin{scope}[xscale=1, yscale=-1]
\begin{scope}[shift={(4,4)}]
\defcoords
\path[fill=llightblue] (q1) -- (q2) -- (q3) -- (q4) -- (q1);
    \draw [thick, legreen, ->] (a1) -- (a2);
    \draw [thick, legreen, ->] (b1)--(b2);
    \draw [thick, legreen, ->] (c1)--(c2);
    \draw [thick, legreen, ->] (d1)--(d2);
    \draw [thick, legreen, ->] (e1)--(e2);

\draw [fill=lchargg] (g1) circle [radius=0.2];
\draw [fill=lchargg] (g2) circle [radius=0.2];
\draw [fill=lchargg] (g3) circle [radius=0.2];
\draw [fill=lchargg] (g4) circle [radius=0.2];
\draw [fill=lchargg] (g5) circle [radius=0.2];
\draw [fill=lchargg] (g6) circle [radius=0.2];
\draw [fill=lchargg] (g7) circle [radius=0.2];
\end{scope}
\end{scope}


\defcoords
\draw [very thick,->] (l1) -- (l2);
\draw [very thick,->] (l3) -- (l4);
\end{scope}

\end{tikzpicture}

\caption{We depict here schematically the transformations applied for constructing the symmetrized field $E_{R_1}^{\mathrm{scr}}$ from the proof of Lemma \ref{screenedperiodic}. The charges associated to the vector fields at the different steps, are represented by small balls. In order to simplify the drawing we forget here about the $y$-coordinate, and the picture refers to a slice at $y=0$, for dimension $d=2$.}\label{figure}
\end{figure}
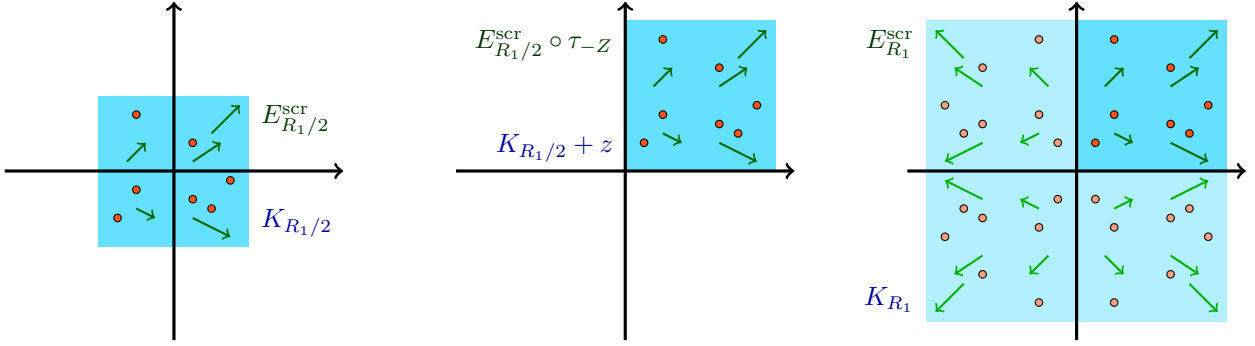

\textbf{Step 3.} While many of the desired properties are alreay verified by the charges in $K_{R_1/2}$, the balancing condition from \eqref{perscreenbar} is generally not true for the charges $\nu|_{R_1/2}$, and we perform a symmetrization procedure in order to ensure it (see Figure \ref{figure}). Note that when in \eqref{minperscr} or in \eqref{defc1_long} we take $R_1\to\infty$, we need the sequence of cubes to cover the whole plane in the limit, and thus we choose cubes with center at the origin. 

For a graphical description of what the construction here is doing in practice, see Figure \ref{figure}. Recall that $K_{R_1/2}=[-R_1/4,R_1/4)^d$ was the cube of sidelength $R_1/2$ with center at $0$. On the other hand we will perform a construction of fields and configurations by reflection, which is notationally lighter if we start with a cube which has a \emph{corner} in the origin. Thus we first translate $K_{R_1/2}$ by ${Z:=(z,0)}:=(R_1/4,\ldots,R_1/4 {,0)\in\mathbb R^{d+1}}$, so that $K_{R_1/2}{\times\mathbb R}+Z=[0,R_1/2)^d{\times\mathbb R}$. Then $E^{\mathrm{scr}}_{R_1/2}\circ \tau_{-Z}(x,y):=E^{\mathrm{scr}}_{R_1/2}(x - z,y)$ for $(x,y)\in K_{R_1/2}\times\mathbb R$ is a translation of $E^{\mathrm{scr}}_{R_1/2}$, satisfying $E^{\mathrm{scr}}_{R_1/2}\circ\tau_{-Z}\in \mathsf{Scr}_{K_{R_1/2}+z}$. We now extend $E^{\mathrm{scr}}_{R_1/2}\circ \tau_{-Z}$ to the larger cube $K_{R_1}{\times\mathbb R}=[-R_1/2,R_1/2)^d\times\mathbb R$ by reflection across the hyperplanes $H_j:=\{(x_1,\ldots,x_d,y)\in\mathbb R^{d+1}:\ x_j=0\}$ for $j=1,\ldots,d$. More precisely, if $[v]_j$ is the $j$-th component of a vector $v\in\mathbb R^{d+1}$, then we define, for all $(x_1,\ldots,x_{d+1})\in\mathbb R^{d+1}$, 
\begin{equation}\label{defescr}
\left[E^{\mathrm{scr},0}_{R_1}(x_1,\ldots,x_d,x_{d+1})\right]_j:=\left\{
\begin{array}{ll}
\left[E^{\mathrm{scr}}_{R_1/2}\circ\tau_{-Z}(|x_1|,\ldots,|x_d|,x_{d+1})\right]_j&\text{ if }x_j\ge 0,\ 1\le j\le d,\\
-\left[E^{\mathrm{scr}}_{R_1/2}\circ\tau_{-Z}(|x_1|,\ldots,|x_d|,x_{d+1})\right]_j&\text{ if }x_j\le 0,\ 1\le j\le d,\\
\left[E^{\mathrm{scr}}_{R_1/2}\circ\tau_{-Z}(|x_1|,\ldots,|x_d|,x_{d+1})\right]_j&\text{ for all }x_j,\text{ if }j=d+1.
\end{array}\right.
\end{equation}
Equivalently, if $U_j:\mathbb R^d\to \mathbb R^d$ is the reflection across the hyperplane {$\tilde H_j:=\{(x_1,\ldots,x_d): x_j=0\}$}, i.e. 
\[U_j(x_1,\ldots, x_d):=(x_1,\ldots,x_{j-1},-x_j,x_{j+1},\ldots,x_d),\]
then, denoting $U_j^1:=U_j, U_j^0:=\mathrm{Id}_{\mathbb R^d}$, we note that we can cover $K_{R_1}$ by $2^d$ reflected copies of $K_{R_1/2}+z$ parametrized by multi-indices $\sigma:=(\sigma_1,\ldots, \sigma_d)\in\{0,1\}^d$ as follows: 
\begin{equation}\label{decomp}
K_{R_1}=\bigcup_{(\sigma_1,\ldots,\sigma_d)\in\{0,1\}^d}\left(U_1^{\sigma_1}\circ\ldots\circ U_d^{\sigma_d}\right)(K_{R_1/2}+z):=\bigcup_{\sigma\in\{0,1\}^d}U^{(\sigma)}(K_{R_1/2}+z).
\end{equation}
The decomposition \eqref{decomp} follows by noticing that each of the hyperplanes $\tilde H_j$ is a symmetry plane of $K_{R_1}$ and the union of these hyperplanes subdivides $K_{R_1}$ into $2^d$ pieces, each of which is related to the piece $K_{R_1/2}+z$ (constituted of those points of $K_{R_1}$ for which all coordinates are nonnegative) via a different composition of reflections $U^{(\sigma)}$, corresponding to the $2^d$ indices $(\sigma_1,\ldots,\sigma_d)\in\{0,1\}^d$.
Next, note that for $\sigma\neq\sigma'\in\{0,1\}^d$, the cubes $U^{(\sigma)}(K_{R_1/2}+z)$ and $U^{(\sigma')}(K_{R_1/2}+z)$ overlap only on their boundaries. Note that each cube is a subset of $\mathbb R^d\times\{0\}\subset\mathbb R^{d+1}$, and our reflections change the $\mathbb R^d$-coordinate only, while we do not reflect the last coordinate; thus the operations $U^{(\sigma)}\otimes \mathrm{Id}_{\mathbb R}$ which are the identity on the $y$-coordinate are the desired extensions of these reflections to $\mathbb R^{d+1}$. With the above notation we define $E^{\mathrm{scr},0}_{R_1}$ for all $\sigma\in\{0,1\}^d$ separately on each cube $\left[U^{(\sigma)}(K_{R_1/2}+z)\right]\times\mathbb R=(U^{(\sigma)}\otimes \mathrm{Id}_{\mathbb R})(K_{R_1/2}\times\mathbb R+Z)$ as follows:  
\begin{equation}\label{reflectnot}
E^{\mathrm{scr},0}_{R_1}(x,y):=\left(U^{(\sigma)}\otimes \mathrm{Id}_{\mathbb R}\right)\cdot E^{\mathrm{scr}}_{R_1/2}\circ\tau_{-Z}((U^{(\sigma)})^{-1}x,y) \quad\text{ for }\quad(x,y)\in (U^{(\sigma)})(K_{R_1/2}+z))\times\mathbb R.
\end{equation}
We claim that the resulting field $E^{\mathrm{scr},0}_{R_1}$ is screened on $K_{R_1}$, i.e. it satisfies \eqref{screenK} with $K:=K_{R_1}$. Indeed, it is screened on the boundary of each one of the regions \eqref{reflectnot}, therefore \eqref{screenK} is also valid on $K_{R_1}\times\mathbb R$. Next, we claim that the compatibility equation \eqref{compatibleinK} holds for $E_{R_1}^{\mathrm{scr},0}$ on $K_{R_1}\times\mathbb R$ for $E_{R_1}^{\mathrm{scr},0}$, with $\bar\nu$ replaced by a configuration $\nu_{R_1}$ which has as restriction $\nu_{R_1}|_{K_{R_1}}$ the symmetrization of $\left.(\tau_Z)_\#\nu\right|_{K_{R_1/2}}$. In other words our claim is that
\begin{equation}\label{nur1}
E^{\mathrm{scr},0}_{R_1}\in\mathsf{Comp}_{K_{R_1},\nu_{R_1}}\quad\text{ for }\quad\nu_{R_1}{|_{K_{R_1}}}=\sum_{\sigma\in\{0,1\}^d}\left[\left(U^{(\sigma)}\otimes \mathrm{Id}_{\mathbb R}\right)\circ\tau_Z\right]_\#\left(\nu|_{K_{R_1/2}}\right).
\end{equation}
To show this, in view of \eqref{reflectnot} we only note that \emph{no extra distributional divergence is created on the interfaces between the regions \eqref{reflectnot}}. This can be shown by repeatedly applying Lemma~\ref{gluinglemma} to the vector fields defining $E_{R_1}^{\mathrm{scr},0}$, and over the domains $U^{(\sigma)}(K_{R_1/2}+\zeta)$ as above. In order to verify the that the hypotheses of Lemma~\ref{gluinglemma} hold, we note two considerations about the field $E_{R_1/2}^{\mathrm{scr}}$ out of which $E_{R_1}^{\mathrm{scr},0}$ is constructed above (as a mnemonic, the reader can refer to Figure \ref{figure}). Firstly, the normal component of $E_{R_1/2}^{\mathrm{scr}}$ is zero over the boundary of each of the $(U^{(\sigma)}\otimes \mathrm{Id}_{\mathbb R})(K_{R_1/2}\times\mathbb R+Z)$ by the property that $E_{R_1/2}^{\mathrm{scr}}$ is screened. Secondly, the divergence of the pieces defined in \eqref{reflectnot} is $L^p$-integrable near the interfaces. Up to applying a translation by $Z$ and a suitable reflection, we may verify this property for $E_{R_1/2}^{\mathrm{scr}}$ only. Since $E_{R_1/2}^{\mathrm{scr}}$ is constructed via Proposition \ref{screening}, its divergence satisfies a compatibility equation like \eqref{compatibleinK} in $K_{R_1/2}$ and thus it is given by the measure $1_{K_{R_1/2}}(x)dx -\nu_{R_1/2}$. Moreover, by the third bullet of the screening Proposition \ref{screening}, this atomic part stays outside a neighborhood of $\partial K_{R_1/2}$ used to construct $E_{R_1/2}^{\mathrm{scr}}$, as desired. Thus Lemma~\ref{gluinglemma} can be applied, and concludes the proof of \eqref{compatibleinK} for $E^{\mathrm{scr},0}_{R_1}$.

Due to the symmetrization applied for defining $E^{\mathrm{scr},0}_{R_1}$, each atom of the measure $\bar\nu_{R_1}$ on $\mathbb R^{d+1}$ comes with symmetrized charges with respect to the coordinate hyperplanes $H_j$, and thus, taking a pushforward under projection of $\mathbb R^{d+1}$ on the first $d$ coordinates, each charge of $\nu_{R_1}$ comes with the symmetrized charges with respect to the hyperplane $\tilde H_j$ (once more, we refer the reader to Figure \ref{figure}). In particular we find that $\left.\nu_{R_1}\right|_{K_{R_1}}$ satisfies the zero-barycenter requirement from
($\mathsf Per_{R_1}^{\mathrm{scr},0}$):
\begin{equation}\label{baryzero}
\int_{K_{R_1}} x d\nu_{R_1}(x)= \sum_{\sigma\in\{0,1\}^d}\sum_{x\in\mathsf{set}\left(\nu|_{K_{R_1/2}}\right)}R^{(\sigma)}(x+z) = \sum_{x\in\mathsf{set}(\nu)}\prod_{j=1}^d\left(\mathrm{Id}_{\mathbb R^d}+R_j\right)(x+z)=0,
\end{equation}
because for each $j=1,\ldots,d,$ we have that $\mathrm{Id}_{\mathbb R^d}+R_j$ annihilates the vector $e_j$ and transforms $e_i$ into $2e_i$ for the other basis vectors, and therefore the product of these transformations over all $j=1,\ldots,d,$ is zero on all basis vectors, and thus vanishes.

\medskip

We may now  extend $E^{\mathrm{scr},0}_{R_1}$ by $(R_1\mathbb Z)^d$-periodicity. Then, $\nu_{R_1}$ being the $(R_1\mathbb Z)^d$-periodic extension of $\nu_{R_1}|_{K_{R_1}}$, we then have $E^{\mathrm{scr},0}_{R_1}\in\mathsf{Comp}_{\nu_{R_1}}$, and thus, together with property \eqref{baryzero}, we obtain
$E^{\mathrm{scr},0}_{R_1}\in\mathsf Per_{R_1}^{\mathrm{scr},0}$ and so 
$E^{\mathrm{scr},0}_{K_{R_1}}$ is a competitor for the minimisation problem on the left of \eqref{minperalmost}. On the other hand by symmetry, and using the notation \eqref{truncE} concerning the subscript $\eta$, we find
\begin{equation}\label{newaverage}
\frac1{|K_{R_1}|}\int_{K_{R_1}\times \mathbb R}|y|^\gamma|E^{\mathrm{scr},0}_{K_{R_1},\eta}|^2 = \frac1{|K_{R_1/2}|}\int_{K_{R_1/2}\times\mathbb R}|y|^\gamma|{E_{R_1/2,\eta}^{\mathrm{scr}}}|^2.
\end{equation}
This follows by noting that, due to the definition \eqref{reflectnot}, the restriction of $E^{\mathrm{scr},0}_{K_{R_1}}$ over each of the sets $(U^{(\sigma)}(K_{R_1/2}+z))\times\mathbb R$ corresponding to the decomposition \eqref{decomp} of $K_{R_1}$ is related to ${E_{R_1/2}^{\mathrm{scr}}}|_{K_{R_1/2}\times\mathbb R}$ by an isometry. If the balls of radius $\eta$ centered at the points $\mathsf{set}(\nu_{R_1/2})$, corresponding to the truncation \eqref{truncE}, do not cross the boundary of $K_{R_1/2}$, then automatically these same isometries also relate $E^{\mathrm{scr},0}_{K_{R_1},\eta}$ to ${E_{R_1/2,\eta}^{\mathrm{scr}}}|_{K_{R_1/2}\times\mathbb R}$ (again, we refer the reader to Figure \ref{figure}, in which this time the red balls can be thought of as representing the $\eta$-balls used in the regularization). Also recall that in the screening procedure encoded by Proposition \ref{screening}, which was used to construct $E_{R_1/2}^{\mathrm{scr}}$, the charges end up being at a controlled distance from the boundary, so the above condition is verified by $E_{R_1/2}^{\mathrm{scr}}$ for $0<\eta<\eta_0$, where $\eta_0$ depends only on $s,d$. Therefore over each of the $2^d$ subcubes from \eqref{decomp} the field $E^{\mathrm{scr},0}_{K_{R_1},\eta}$ has the same energy as ${E_{R_1/2,\eta}^{\mathrm{scr}}}|_{K_{R_1/2}\times\mathbb R}$, ensuring \eqref{newaverage} as desired, for all $\eta>0$ small enough that the charges of $E_{R_1/2}^{\mathrm{scr}}$ are more than $\eta$-distant from the boundary. The existence of such $\eta$ is ensured by Proposition \ref{screening}.

\medskip

Due to \eqref{newaverage} and to \eqref{screenedbound2}, we now have that 
$$\frac1{|K_{R_1}|}\int_{K_{R_1}\times\mathbb R}|y|^\gamma|E^{\mathrm{scr,0}}_{R_1,\eta}|^2 -c_{s,d}\mathsf c(\eta)\le {c_{s,d}}\mathsf C_{\mathrm{Jel}}(s,d) + \epsilon_0,$$  
from which we obtain by means of \eqref{weta} that
$$\mathcal W(E^{\mathrm{scr,0}}_{R_1,\eta})\le {c_{s,d}}\mathsf C_{\mathrm{Jel}}(s,d) + \epsilon_0.$$
In view of the above and since $E^{\mathrm{scr},0}_{R_1}\in\mathsf Per_{R_1}^{\mathrm{scr},0}$ , our construction proves \eqref{minperalmost} which concludes the proof. 
\end{proof}
\subsection{The Uniform Electron Gas energy and the constant $\mathsf{C}_{\mathrm{UEG}}(s,d)$}
\label{unifjelnewcomp}
We will need to introduce next some more notation. In our statements and proofs below, we define for a probability measure $\mu\in {\cal P}(\mathbb{R}^d)$ for a cost function $c$ as in \eqref{value c},  and for $0<s<d$
\begin{equation}\label{defexc}
E_{N,s}^{\mathrm{xc}}(\mu):=\mathcal F_{N,s}(\mu)-N^2\int_{\mathbb{R}^d}\int_{\mathbb{R}^d} c(x,y)\diff\mu(x)\diff\mu(y).
\end{equation} 
Translating the result stated in Theorem \ref{thnextorder} (b) in the notation from \eqref{defexc}, we have that for all $\mu\in {\cal P}(\mathbb{R}^d)$ of the form $d\mu(x)=\rho(x)\diff x$ such that $\rho\in L^1(\mathbb{R}^d)\cap L^{1+s/d}(\mathbb R^d)$, and for $0<s<d$, there holds 
\begin{eqnarray}\label{defc20}
\textsf{C}_{\mathrm{UEG}}(s,d)&=&\frac{\lim_{N\to\infty} N^{-1-s/d}E_{N,s}^{xc}(\mu)}{\int_{\mathbb R^d}\rho^{1+s/d}(x)\diff x}.
\end{eqnarray}

The equation \eqref{defc20} says that the constant $\textsf{C}_{\mathrm{UEG}}(s,d)$ is independent of the marginal $\mu$. The particular case with $\mu$ being a uniform measure with density $\rho_{\mathrm{UEG}}(x):=1_{[0,1]^d}(x)$ is known as the \textit{Uniform Electron Gas} (UEG), as already explained in Section \ref{jelueg00}. 

In what follows and due to the fact that for the screening procedure the sets $K_R=\left[-\frac{R}{2},\frac{R}{2}\right)^d$ are used, we will work mostly with a scaled variant of the UEG of the same type. More precisely, we will work with $\rho(x):=\left[-\frac{1}{2},\frac{1}{2}\right)^d$, in which situation \eqref{defc20} reduces to
\begin{equation}\label{defc2}
\textsf{C}_{\mathrm{UEG}}(s,d)=\lim_{N\to\infty} N^{-1-s/d}E_{N,s}^{\mathrm{xc}}\left(1_{\left[-\frac{1}{2},\frac{1}{2}\right)^d}\diff x\right)=\lim_{N\to\infty} N^{-1-s/d}E_{N,s}^{\mathrm{xc}}\left(1_{[0,1]^d}\diff x\right).
\end{equation}

\subsection{Reinterpreting $\mathsf C_{\mathrm{Jel}}(s,d), \mathsf C_{\mathrm{UEG}}(s,d)$, in terms of the simpler energies expressions $E_{\mathrm{Jel}}$, $E_{\mathrm{UEG}}$}
\label{reinterpret}

\medskip

We next find the link between the above definition of $\textsf{C}_{\mathrm{Jel}}(s,d)$, as given in \eqref{defc1} and \eqref{defc1_long}, and a limit for the Jellium type energy $E_{\mathrm{Jel},s}(\mu,\nu)$ which quantity coincides for uniform marginals with the one given in Section \ref{jelueg00}. 
We will also re-express $\mathsf C_{\mathrm{UEG}}(s,d)$ in terms of a Uniform Electron Gas energy $E_{\mathrm{UEG},s}(\mu,\nu)$, which generalizes the corresponding definition from Section \ref{jelueg00}. In (\ref{reformulc1}) we present the equality between the next-order term constants for (\ref{nextordergas}) and (\ref{jelueg001}). In (\ref{valuec1}), we reinterpret $\textsf{C}_{\mathrm{Jel}}(s,d)$ in terms of $E_{\mathrm{Jel},s}(\mu,\nu)$ computed at configurations corresponding to minimizing values $\nu$ for the original Coulomb and Riesz gases problem. To the best of our knowledge, the results from Lemma \ref{reformulc1c2} (a) and (b) are new and of independent interest. The proofs of Lemma \ref{reformulc1c2} (a) and (b) follow by means of Lemma \ref{screenedperiodic} and by adapting ideas from the screening method as explained and used in \cite{PS}, and as briefly described in the proofs below. The result in (\ref{valuec1}) will be crucial in the proof of the Main Theorem for $d-2<s<d$ given in Section \ref{conclstrict}.

In the definitions below, $\nu\in {\mathcal M}^+(\mathbb{R}^d)$ is as defined in \eqref{pointconfig}. Although we define these $ E_{\mathrm{Jel}}(\mu,\nu),E_{\mathrm{UEG}}(\mu,\nu),$ energies for general measures $\mu$ on $\mathbb{R}^d$ and for the measures $\bar\mu$ on the extended space $\mathbb R^{d+1}$, they will be used mostly for very special measures of the form $d\mu(x)=\alpha 1_K(x)\diff x$, for $\alpha>0, K\subset\mathbb R^d$ compact. Alternatively, we may define the $ E_{\mathrm{Jel}}(\bar\mu,\bar\nu)$ energy for general measures $\bar\nu,\bar\mu$ on the extended space $\mathbb R^{d+1}$, as defined in \eqref{iota}; this definition can be shown to be equivalent to the first one, as explained in \eqref{defjellium} below.

More precisely, if $\mu$ is a positive measure on $\mathbb{R}^d$ such that $\int_{\mathbb{R}^d}\int_{\mathbb{R}^d} \textsf{c}(x-y)\diff\mu(x)\diff\mu(y)<\infty$, and $\nu\in\mathsf{Config}$ such that the points in $\mathsf{set}(\nu)$ are distinct and finitely many, we define for $\mathsf{c}:\mathbb{R}^d\times\mathbb{R}^d\rightarrow\mathbb{R}\cup\{0\}$ the following simpler version of our energies, which we note here that can be equivalently written either in terms of the measures $\mu,\nu$ on the original space $\mathbb R^d$ or in terms of $\bar\mu, \bar\nu$ in the extended space $\mathbb R^{d+1}$:
\begin{eqnarray}\label{defjellium}
E_{\mathrm{Jel},\mathsf{c}}(\mu, \nu)&:=&\sum_{p\neq q\in\mathsf{set}(\nu)}\textsf{c}(p-q) - 2\sum_{p\in\mathsf{set}(\nu)}\int_{\mathbb R^d}\textsf{c}(p-y) \diff\mu(y) + \int_{\mathbb R^d}\int_{\mathbb R^d} \textsf{c}(x-y)\diff\mu(x) \diff\mu(y)\nonumber\\
&=&\sum_{(p,0)\neq (q,0)\in\mathsf{set}(\bar\nu)}\textsf{c}(p-q) - 2\sum_{(p,0)\in\mathsf{set}(\bar\nu)}\int_{\mathbb R^d}\textsf{c}(p-y) \diff\mu(y) + \int_{\mathbb R^d}\int_{\mathbb R^d} \textsf{c}(x-y)\diff\mu(x) \diff\mu(y)\nonumber\\[3mm]
&=&E_{\mathrm{Jel},\mathsf{c}}(\bar\mu, \bar\nu).
\end{eqnarray}
To prove the above we may use the definition \eqref{iota} and the fact that $\iota:\mathbb R^d\to \mathbb R^{d+1}$ preserves distances. The above equivalent formulation as $E_{\mathrm{Jel}}(\bar\mu,\bar\nu)$ will be used later in order to relate $E_{\mathrm{Jel}}(\mu,\nu)$ to the description of $\mathsf C_{\mathrm{Jel}}(s,d)$, as done in \cite{PS}, and as used in \eqref{renormen} here. Further, we define
\begin{equation}\label{defeind}
E_{\mathrm{UEG},\mathsf{c}}(\mu, \nu):=\sum_{p\neq q\in\mathsf{set}(\nu)}\textsf{c}(p-q) -\int_{\mathbb R^d}\int_{\mathbb R^d} \textsf{c}(x-y)\diff\mu(x) \diff\mu(y).
\end{equation}
We now denote as follows the ``scalar product outside the diagonal'', respectively the usual scalar product,  between two finite measures $\mu,\nu$ and weighted by $\textsf{c}(x)=1/|x|^s$
\begin{equation}\label{notationduality}
\langle\mu,\nu\rangle_s^*:=\left\langle\mu\otimes\nu|_{\{(x,y):\ x\neq y\}},\textsf{c}(x-y)\right\rangle,\qquad \langle\mu,\nu\rangle_s:=\left\langle\mu\otimes\nu,\textsf{c}(x-y)\right\rangle.
\end{equation}
Then, noting that $\langle\mu_1,\mu_2\rangle_s^*=\langle\mu_1,\mu_2\rangle_s$, if at least one of $\mu_1,\mu_2$ is absolutely continuous, we find the reexpressions
\begin{equation}\label{reexpressejeueg}
E_{\mathrm{Jel},s}(\mu,\nu)=\langle\mu-\nu, \mu-\nu\rangle_s^*,\quad \mbox{and}\quad E_{\mathrm{UEG},s}(\mu,\nu)= \langle\nu,\nu\rangle_s^*-\langle\mu,\mu\rangle_s.
\end{equation}
Moreover, we have from \eqref{reexpressejeueg} that
\begin{equation}\label{whatever}E_{\mathrm{UEG},s}(\mu, \nu_{\vec x})=\left\langle \nu_{\vec x}, \nu_{\vec x}\right\rangle_s^* - \left\langle \mu, \mu\right\rangle_s ~~~\mbox{and}~~~ E_{\mathrm{Jel},s}(\mu, \nu_{\vec x})=\left\langle \nu_{\vec x} - \mu, \nu_{\vec x} - \mu\right\rangle_s^*.          \end{equation}

Keeping in mind $\langle\cdot, \cdot\rangle_s^*=\langle\cdot, \cdot\rangle_s$, we obtain
\begin{eqnarray}\label{noselfint}
E_{\mathrm{UEG},s}(\mu, \nu_{\vec x}) - E_{\mathrm{Jel},s}(\mu, \nu_{\vec x})&=&\left\langle \nu_{\vec x}, \nu_{\vec x}\right\rangle_s^* - \left\langle \mu, \mu\right\rangle_s^* -\left\langle \nu_{\vec x} - \mu, \nu_{\vec x} - \mu\right\rangle_s^*\nonumber\\
&=&2\left\langle \mu, \nu_{\vec x} - \mu \right \rangle_s^*=2\left\langle \mu, \nu_{\vec x}-\mu\right\rangle_s.
\end{eqnarray}
If $K\subset\mathbb R^d$ is a compact set, we also abbreviate, by abuse of notation, as follows: 
\begin{equation}\label{defjelliumind1}
 E_{\mathrm{Jel},s}(K,\nu):= E_{\mathrm{Jel},s}\left(1_{K}(x)\diff x,\nu\right), \quad  E_{\mathrm{UEG},s}(K,\nu):= E_{\mathrm{UEG},s}\left(1_K(x)\diff x, \nu\right).
\end{equation}
\begin{lemma}\label{reformulc1c2}
We have the following reformulations for the above definitions of $\textsf{C}_{\mathrm{Jel}}(s,d), \textsf{C}_{\mathrm{UEG}}(s,d)$: 

\begin{itemize}
\item [(a)] Let either $d\ge 3$ and $d-2\le s<d$, or $d=2$ and $d-2<s<d$. Then we have
\begin{equation}\label{reformulc1}
\textsf{C}_{\mathrm{Jel},s}(s,d)= \lim_{R^d=N\to\infty}\frac{\min\left\{E_{\mathrm{Jel},s}(K_R,\nu_{\vec x}):\ \vec x= (x_1,\ldots, x_N)\in(\mathbb R^d)^N\right\}}{N}.
\end{equation}
\item [(b)] Let either $d\ge 3$ and $d-2\le s<d$, or $d=2$ and $d-2<s<d$. For each $R_1>0$ large enough and such that $(R_1/2)^d\in\mathbb N$, let $\nu_{R_1} \in \mathsf{Config}$ be $(R_1\mathbb Z)^d$-periodic configurations corresponding to minimizers $E^{\mathrm{scr},0}_{R_1}$ of $\mathcal W$ on $\mathsf{Per}^{\mathrm{scr},0}_{R_1}$ as in \eqref{minperscr}, such that $E^{\mathrm{scr},0}_{R_1}\in\mathsf{Comp}_{\nu_{R_1}}$. Considering below only $R^d=N\in\mathbb N$ such that $R/R_1\in\mathbb N$, there holds
\begin{equation}\label{valuec1} 
\textsf{C}_{\mathrm{Jel},s}(s,d)=\lim_{\substack{R_1\to\infty\\(R_1/2)^d\in\mathbb N}}\liminf_{\substack{R^d=N\to\infty\\ R/R_1\in\mathbb N}} \frac{E_{\mathrm{Jel},s}(K_R,\nu_{{R_1}}|_{K_R})}{N}.
\end{equation}

\item [(c)] For $0<s<d$ we have
\begin{equation}\label{reformulc2}
{\textsf{C}_{\mathrm{UEG}}(s,d)}=\lim_{R^d=N\to\infty} \frac{\min\left\{\int E_{\mathrm{UEG},s}(K_R,\nu_{\vec x})\diff\gamma_N(\vec x):\ \gamma_N\mapsto\frac{\mu_N}{N}\right\}}{N},
\end{equation}
where $d\mu_N(x)=1_{K_R}(x)\diff x=1_{\big[-\frac{N^{1/d}}{2},\frac{N^{1/d}}{2}\big)^d}(x)\diff x$.
\end{itemize}
\end{lemma}
\begin{proof}
\par \textbf{Step 1:} \textit{Proof of ``$\leq$'' in \eqref{reformulc1}}. We first note that it can be proved that the limit
\begin{equation}
\label{existlim}
\lim_{N\rightarrow\infty}\frac{\min\left\{E_{\mathrm{Jel},s}(K_{R},\nu_{\vec x}):\ \vec x= (x_1,\ldots, x_N)\in(\mathbb R^d)^N\right\}}{N}
\end{equation}
exists, as can be shown by means of the almost-subaddativity formula from Lemma \ref{subaddjell} below, and it is finite by Lemma \ref{LNequiv} in the Appendix. We also observe here that alternatively, the (\ref{existlim}) immediately follows for a subsequence $(N_k)_{k\ge 1}$ from the upper/lower bounds in Lemma \ref{LNequiv}. The subsequence result is sufficient for our purposes in view of the fact that the limit representations for $\mathsf C_{\mathrm{Jel}}(s,d)$ in \eqref{defc1_long} and \eqref{minperscr} have been derived in terms of the full sequence.

Note also that the minimization problem on the r.h.s. of \eqref{reformulc1} was extensively treated for the Coulomb case ($s=1,d=3$) in \cite{LiebLeb72} and \cite{LiNarn75}. See also, for example, \cite{BlLe15} for a comprehensive discussion of other minimzation problems of this type.

Next, in order to prove the inequality ``$\le$'' in \eqref{reformulc1}, that is
\[
\textsf{C}_{\mathrm{Jel}}(s,d)\le \lim_{R^d=N\to\infty}\frac{\min\left\{E_{\mathrm{Jel},s}(K_R,\nu_{\vec x}):\ \vec x= (x_1,\ldots, x_N)\in(\mathbb R^d)^N\right\}}{N},
\]
we consider for each $N\ge 2$ a minimizer ${\vec x}_N$ to the minimization problem from the right-hand-side in \eqref{reformulc1}, and we construct a competitor for the left-hand-side minimization satisfied 
by $\mathsf C_{\mathrm{Jel}}(s,d)$, by using the formulation \eqref{defc1_long}. We can use the equation \eqref{fundsolrd+1} and the decay of our kernels to show the following, with the notation \eqref{truncE} used here for $E=\nabla H^{\bar\nu_{{\vec x}_N} - \bar\mu_N}$:
\begin{eqnarray}\label{renormen}
E_{\mathrm{Jel},s}(K_R, \nu_{{\vec x}_N})&\stackrel{\text{\eqref{reexpressejeueg}}}{=}& {\int_{(\mathbb R^d)^2\setminus\{(x,x):  x\in\mathbb R^d\}} \textsf{c}(x-y) \diff\left(\nu_{{\vec x}_N} - \mu_N\right)(x)\diff\left(\nu_{{\vec x}_N} - \mu_N\right)(y)}\nonumber\\
&\stackrel{\text{\cite[p.17--18]{PS}}}{=}& \lim_{\eta\to 0} \left(\frac{1}{c_{s,d}}\int_{\mathbb R^{d+1}}|y|^\gamma |\nabla H^{\bar\nu_{{\vec x}_N}-\bar\mu_N}_\eta|^2 - N \textsf{c}(\eta)\right).
\end{eqnarray}
The last passage follows by the discussion detailed in \cite[p.17--18]{PS} applied to \eqref{reexpressejeueg}, taken with the choices ${\vec x}_N\in(\mathbb R^d)^N$, $R^d=N$ and $d\mu_N(x)=1_{K_R}(x)\diff x$. That discussion goes as follows. Firstly, we use equivalence of the two expressions stated in \eqref{defjellium}. Then we re-express $E_{\mathrm{Jel}}(K_R,\nu_{{\vec x}_N})$ in $\mathbb R^{d+1}$-coordinates. The fact that $\mathsf c$ satisfies equation \eqref{fundsolrd+1} and an integration by parts (which is justified by the decay properties of our kernels) then directly implies \eqref{renormen}. See \cite{PS} for more details.

\medskip

From this point on, we may perform precisely the same procedure as for the construction of a screened periodic competitor $E^{\mathrm{scr}}_{R_1/2}$ in the proof of Lemma \ref{screenedperiodic}, but applied to the cube $K_R$ and to the field $E_\eta$ associated via \eqref{truncE} to $E=\nabla H^{\bar\nu_{{\vec x}_N} - \bar\mu_N}$. Precisely like in the proof leading to \eqref{minperalmost}, this allows to produce a competitor for the minimization problem \eqref{defc1_long} that re-expresses $\mathsf C_{\mathrm{Jel}}(s,d)$, and whose energy is increased by at most $\epsilon_0>0$ with respect to the one of the right hand side of \eqref{renormen}. By the arbitrarity of such $\epsilon_0$ this suffices to prove the inequality ``$\le$'' in \eqref{reformulc1}. 

\par \textbf{Step 2:} \textit{Proof of ``$\le$'' in \eqref{valuec1}}.
We may use that $\nu_{R_1}|_{K_R}$ is a competitor for the minimzation problem on the right-hand side of \eqref{reformulc1}, due to $\nu_{R_1}|_{K_R}$ having $N$ charges in $K_R$, from which
\[
\frac{E_{\mathrm{Jel},s}(K_R,{\nu_{R_1}|_{K_R}})}{N}\ge\frac{\min\left\{E_{\mathrm{Jel},s}(K_R,\nu_{\vec x}):\ \vec x=(x_1,\ldots,x_N)\in(\mathbb R^d)^N\right\}}{N}.
\]
Via the result of Step 1, the above in turn leads to
\begin{eqnarray}
\liminf_{\substack{R^d=N\to\infty\\ R/R_1\in\mathbb N}} \frac{E_{\mathrm{Jel},s}(K_R,\nu_{{R_1}}|_{K_R})}{N}&\ge&\liminf_{\substack{R^d=N\to\infty\\ R/R_1\in\mathbb N}}\frac{\min\left\{E_{\mathrm{Jel},s}(K_R,\nu_{\vec x}):\ \vec x=(x_1,\ldots,x_N)\in(\mathbb R^d)^N\right\}}{N}\label{step2}\\
&\stackrel{\text{\eqref{existlim}}}{=}&\lim_{R^d=N\to\infty}\frac{\min\left\{E_{\mathrm{Jel},s}(K_R,\nu_{\vec x}):\ \vec x=(x_1,\ldots,x_N)\in(\mathbb R^d)^N\right\}}{N}\ge\textsf{C}_{\mathrm{Jel}}(s,d).\nonumber
\end{eqnarray}
Taking the $\liminf$ in the above over $R_1$, with $(R_1/2)^d\in\mathbb N$, this proves the ``$\leq$'' inequality for \eqref{valuec1}.
%
%
\par \textbf{Step 3.} \textit{Proof of ``$\geq$'' in \eqref{valuec1}.}  We prove the ``$\geq$'' part of \eqref{valuec1}, namely
\begin{equation}\label{to-do}
\textsf{C}_{\mathrm{Jel}}(s,d)\ge\lim_{\substack{R_1\to\infty\\(R_1/2)^d\in\mathbb N}}\liminf_{\substack{R^d=N\to\infty\\ R/R_1\in\mathbb N}} \frac{E_{\mathrm{Jel},s}(K_R,\nu_{{R_1}}|_{K_R})}{N}.
\end{equation}
From this, the ``$\ge$'' inequality for \eqref{reformulc1} also directly follows, by using \eqref{step2}.

\par To prove \eqref{to-do}, we may use directly the outcome of Lemma \ref{screenedperiodic} as follows. Fix $R_1$ with $(R_1/2)^d\in\mathbb N$, and choose a vector field $E^{\mathrm{scr},0}_{R_1}\in\mathsf{Per}_{R_1}^{\mathrm{scr},0}$ which is a minimizer for the minimization problem on the right of \eqref{minperscr}.  Then $E^{\mathrm{scr},0}_{R_1}$ is $(R_1\mathbb Z)^d$-periodic. We now apply the definition \eqref{weta} to this vector field with the goal to link the definition of $C_{\mathrm{Jel}}(s,d)$ to the limit on the right-hand side in \eqref{to-do}.

As noted after formula \eqref{screenK} we deduce that $E^{\mathrm{scr},0}_{{R_1}}\in\mathsf{Comp}_{\nu_{R_1}}$, for some $(R_1\mathbb Z)^d$-periodic configuration $\nu_{R_1}$, such that $\nu_{R_1}|_{K_{R_1}}$ is composed of $R_1^d$ charges, and $\nu_{R_1}|_{K_R}$ is composed of $N$ charges. Then due to Lemma \ref{leblemma} we find (where we recall that the notation ${\overline{\nu_{R_1}|_{K_R}}}$ and ${{\bar\mu}_N}$ means that the measures are now considered on the extended space)
\begin{equation}\label{leblemmconseq10}
\frac1{N}\frac{1}{c_{s,d}}\int_{K_R\times\mathbb R}|y|^\gamma|E^{\mathrm{scr},0}_{{R_1},\eta}|^2-\mathsf c(\eta)\ge \frac1{N}\frac{1}{c_{s,d}}\int_{\mathbb R^{d+1}}|y|^\gamma|\nabla H^{{\overline{\nu_{R_1}|_{K_R}}}-\bar\mu_N}_\eta|^2-\mathsf c(\eta).
\end{equation}
By taking then the limit $\eta\rightarrow 0$ on both sides of the above, \eqref{leblemmconseq10} and \eqref{renormen} give
\begin{eqnarray}\label{leblemmconseq1}
\lim_{\eta\rightarrow 0}\left(\frac1{N}\frac{1}{c_{s,d}}\int_{K_R\times\mathbb R}|y|^\gamma|E^{\mathrm{scr},0}_{{R_1},\eta}|^2-\mathsf c(\eta)\right)&\ge& \frac{E_{\mathrm{Jel},s}(K_R,{\nu_{R_1}|_{K_R}})}{N}.
\end{eqnarray}
%

We now note that by periodicity and similarly to \eqref{newaverage} from the proof of Lemma \ref{screenedperiodic} we also have that the values of $E^{\mathrm{scr},0}_{R_1,\eta}$ on each of the cubes of sidelength $R_1$ which tile $K_R$ equal the values over $K_{R_1}$ up to composing with a translation, therefore taking the average of the energy over $K_R$ or over $K_{R_1}$ gives the same result, in the sense that
\begin{equation}\label{perioduse}
\frac1{|K_R|}\frac{1}{c_{s,d}}\int_{K_R\times\mathbb R}|y|^\gamma|E^{\mathrm{scr},0}_{R_1,\eta}|^2=\frac1{|K_{R_1}|}\frac{1}{c_{s,d}}\int_{K_{R_1}\times\mathbb R}|y|^\gamma|E^{\mathrm{scr},0}_{{R_1},\eta}|^2.
\end{equation}
Using now the definition \eqref{weta} of $\mathcal W$ and \eqref{leblemmconseq1}, we find 
\begin{eqnarray}\label{leblemmconseq2}
{\frac{1}{c_{s,d}}}\mathcal W(E^{\mathrm{scr},0}_{R_1})&\stackrel{\text{\eqref{weta}}}{=}&\frac{1}{c_{s,d}}\lim_{\eta\rightarrow 0}\limsup_{R\to\infty} \left(\frac{1}{|K_R|}\int_{K_R\times\mathbb R} |y|^\gamma|E^{\mathrm{scr},0}_{{R_1},\eta}|^2 -c_{s,d} \textsf c(\eta)\right)\nonumber\\
&\ge&\lim_{\eta\rightarrow 0}\liminf_{\substack{R^d=N\rightarrow\infty\\ R/R_1\in\mathbb N}}\left(\frac1{|K_R|}\frac{1}{c_{s,d}}\int_{K_R\times\mathbb R}|y|^\gamma|E^{\mathrm{scr},0}_{R_1,\eta}|^2- \textsf c(\eta)\right)\nonumber\\
&\stackrel{\text{\eqref{perioduse}}}{=}&\lim_{\eta\rightarrow 0}\liminf_{\substack{R^d=N\rightarrow\infty\\ R/R_1\in\mathbb N}}\left(\frac1{|K_{R_1}|}\frac{1}{c_{s,d}}\int_{K_{R_1}\times\mathbb R}|y|^\gamma|E^{\mathrm{scr},0}_{{R_1},\eta}|^2-\textsf c(\eta)\right)\nonumber\\
&=&\lim_{\eta\to0}\left(\frac1{|K_{R_1}|}\frac{1}{c_{s,d}}\int_{K_{R_1}\times\mathbb R}|y|^\gamma|E^{\mathrm{scr},0}_{{R_1},\eta}|^2-\textsf c(\eta)\right)\nonumber\\
&\stackrel{\text{\eqref{perioduse}}}{\ge}&\liminf_{\substack{R^d=N\rightarrow\infty\\ R/R_1\in\mathbb N}}\lim_{\eta\rightarrow 0}\left(\frac1{|K_{R}|}\frac{1}{c_{s,d}}\int_{K_{R}\times\mathbb R}|y|^\gamma|E^{\mathrm{scr},0}_{{R_1},\eta}|^2-\textsf c(\eta)\right)\nonumber\\
&\stackrel{\text{\eqref{leblemmconseq1}}}{\ge}& \liminf_{\substack{R^d=N\to\infty\\ R/R_1\in\mathbb N}}\frac{E_{\mathrm{Jel},s}(K_R,{\nu_{R_1}|_{K_R}})}{N},
\end{eqnarray}
where for the third equality in the above we used that the term in the second equality only depends on $R_1$. 

Therefore, via Lemma \ref{screenedperiodic} now for $R_1$ such that $(R_1/2)^d\in\mathbb N$ large enough, the field $E^{\mathrm{scr},0}_{R_1}$ has energy as defined on the left of \eqref{leblemmconseq2}, arbitrarily close to the value of $\mathsf C_{\mathrm{Jel}}(s,d)$, and thus taking the $\limsup$ of $R_1\to\infty$ with $(R_1/2)^d\in\mathbb N$ on both sides of \eqref{leblemmconseq2}, concludes the proof of the inequality ``$\ge$'', as desired.

\par \textbf{Step 4:} \textit{Proof of \eqref{reformulc2}.}
\par By scaling, if $S_\alpha(x)=\alpha x$ is a dilation, then we find by \cite[Lemma 2.4 (b)]{cotpet} (and in view of the Monge-Kantorovich duality formulation of the problem, proved in \cite{depascaleduality}), that for any $\mu\in\mathcal{P}(\mathbb{R}^d)$, we have 
\[E_{N,s}^{\mathrm{xc}}((S_\alpha)_\#\mu)=\alpha^{-s}E_{N,s}^{\mathrm{xc}}(\mu).\] 
We use next the notation 
\begin{equation}
\label{scaledrhogen}
\rho_N(x):=\rho(N^{-1/d}x), ~~~\mu_N:=N(S_{N^{1/d}})_\#\mu,
\end{equation}
in which case $d\mu_N(x)=\rho_N(x)\diff x$. Thus, assuming next that $\rho(x)=1_{\left[-\frac{1}{2},\frac{1}{2}\right)^d}(x)$, we obtain
\begin{equation}
\label{scaledrho}
\rho_N(x)=1_{\left[-\frac{N^{1/d}}{2},\frac{N^{1/d}}{2}\right)^d}(x)=1_{K_{R}}(x),
\end{equation}
and we find from \eqref{defc2} that
\begin{eqnarray}
\textsf{C}_{\mathrm{UEG}}(s,d)&=&\lim_{N\to\infty} N^{-1-s/d} E_{N,s}^{\mathrm{xc}}(\mu)=\lim_{N\to\infty} N^{-1} E_{N,s}^{\mathrm{xc}}\left(\frac{\mu_N}{N}\right)\nonumber\\
&=&\lim_{N\to\infty} {\frac{1}{N}}\left({\mathcal F}_{N,s}\left(\frac{\mu_N}{N}\right) - \int_{\mathbb R^d\times\mathbb R^d}\frac{\rho_N(x)\rho_N(y)}{|x-y|^s}\diff x\ \diff y\right).\label{rescaled}
\end{eqnarray}

Then, comparing to \eqref{defeind} and with the notation $\nu_{\vec x}:=\{x_1,\ldots, x_N\}$ for $\vec x=(x_1,\ldots,x_N)\in (\mathbb R^d)^N,$ we have
\begin{equation}\label{eind}
{\mathcal F}_{N,s}\left(\frac{\mu_N}{N}\right) - \int_{\mathbb R^d\times\mathbb R^d}\frac{\rho_N(x)\rho_N(y)}{|x-y|^s}\diff x\ \diff y = \inf\left\{\int E_{\mathrm{UEG},s}(K_R,\nu_{\vec x})\diff\gamma_N(\vec x):\ \gamma_N\mapsto\frac{\mu_N}{N}\right\},
\end{equation}
from which \eqref{reformulc2} directly follows in view of \eqref{scaledrho} and \eqref{rescaled}.
\end{proof}

\section{Steps in the proof of the equality $\textsf{C}_{\mathrm{Jel}}(s,d)=\textsf{C}_{\mathrm{UEG}}(s,d)$}
\label{framework}
In what follows below, we will prove that $\textsf{C}_{\mathrm{Jel}}(s,d)=\textsf{C}_{\mathrm{UEG}}(s,d)$, starting first with the proof of the \textit{easier} direction $\textsf{C}_{\mathrm{Jel}}(s,d)\le\textsf{C}_{\mathrm{UEG}}(s,d)$ of the inequality, and then giving the \textit{more difficult} direction $\textsf{C}_{\mathrm{Jel}}(s,d)\ge \textsf{C}_{\mathrm{UEG}}(s,d)$.

\subsection{The inequality $\textsf{C}_{\mathrm{Jel}}(s,d)\le \textsf{C}_{\mathrm{UEG}}(s,d)$}\label{easyineq}
It is easy to check that if $\gamma_N\mapsto\frac{\mu_N}{N}$, then there holds
\begin{equation}\label{eindbetter}
\quad\int E_{\mathrm{UEG},s}(K_R,\nu_{\vec x})\diff\gamma_N(\vec x)= \int E_{\mathrm{Jel},s}(K_R,\nu_{\vec x})\diff\gamma_N(\vec x).
\end{equation}
We can therefore rewrite, for the general case where $\mu_N$ is defined as in \eqref{scaledrhogen}, the expression in \eqref{reformulc2} with $ E_{\mathrm{Jel}}$ instead of $ E_{\mathrm{UEG}}$, that is

\begin{eqnarray}
\label{reformulc2a}
\textsf{C}_{\mathrm{UEG}}(s,d)&=&\lim_{N\to\infty} \frac{\inf\left\{\int E_{\mathrm{UEG},s}(K_R,\nu_{\vec x})\diff\gamma_N(\vec x):\ \gamma_N\mapsto\frac{\mu_N}{N}\right\}}{N}\nonumber\\
&=&\lim_{R^d=N\to\infty} \frac{\inf\left\{\int E_{\mathrm{Jel},s}(K_{R},\nu_{\vec x})\diff\gamma_N(\vec x):\ \gamma_N\mapsto \frac{\mu_N}{N}\right\}}{N},
\end{eqnarray}
 and it directly implies the following:
\begin{corollary}\label{firstineq}
Let either $d\ge 3$ and $d-2\le s<d$, or $d=2$ and $d-2<s<d$. If the limits defining $\textsf{C}_{\mathrm{Jel}}(s,d), \textsf{C}_{\mathrm{UEG}}(s,d)$ in \eqref{reformulc1}, \eqref{reformulc2} exist, then $\textsf{C}_{\mathrm{Jel}}(s,d)\le \textsf{C}_{\mathrm{UEG}}(s,d)$.
\end{corollary}
\begin{proof}
Suppose that $\nu_{\vec {x_0}}$ and ${\tilde\gamma}_N$ are respectively minimizers at fixed $N$ for the expression on the right-hand side in \eqref{reformulc1}, respectively on the right-hand side of \eqref{reformulc2a}. From \eqref{reformulc1}, one immediately has
\[ E_{\mathrm{Jel},s}(K_R,\nu_{\vec x_0})\le  E_{\mathrm{Jel},s}(K_R,\nu_{\vec x}),\forall\vec x,\]
from which it follows also that for all measures $\gamma_N(\vec x)$ such that $\gamma_N\mapsto\frac{\mu_N}{N}$
 \[ E_{\mathrm{Jel},s}(K_R,\nu_{\vec x_0})\le \int E_{\mathrm{Jel},s}(K_R,\nu_{\vec x})\diff\gamma_N(\vec x).\]
In particular, this will also hold for the ${\tilde{\gamma}}_N$ which achieves the minimum value in \eqref{reformulc2a}. Therefore we have
 \begin{equation}
  E_{\mathrm{Jel},s}(K_R,\nu_{\vec x_0})\le \int E_{\mathrm{Jel},s}(K_R,\nu_{\vec x})\diff{\tilde\gamma}_N(\vec x)=\inf\left\{\int E_{\mathrm{Jel},s}(K_R,\nu_{\vec x})\diff\gamma_N(\vec x): \gamma_N\mapsto\tfrac{\mu_N}{N}\right\}=E^{\mathrm{xc}}_{N,s}\left(\mu_N/N\right),
  \end{equation}
which, upon dividing by $N$ and taking the limit $N\to\infty$, implies the statement of the Corollary.
\end{proof}

\subsection{The framework for proving the inequality $\textsf{C}_{\mathrm{UEG}}(s,d)\le \textsf{C}_{\mathrm{Jel}}(s,d)$ for $d-2<s<d$}

We briefly outline here the proof strategy. We would like to construct from the minimizer of the Jellium energy \eqref{defjellium} a competitor to the optimal transport problem  ${\cal F}_{N,s}(\mu)$, with the marginal $\mu$ being here the uniform measure on a cube $K_R$ of volume $N$, and density $1_{K_R}(x)/N$. We would also like this competitor to be close enough for each $N$ to the actual minimizer of  ${\cal F}_{N,s}(\mu)$ so that, when scaling by $N$, we recover the next-order constant $\textsf{C}_{\mathrm{UEG}}(s,d)$ as the limit $N\to\infty$. We will not be able to do this directly: We will instead be able to construct in Section \ref{contrucompet} such a competitor for an optimal transport problem with the "wrong" marginal. The value resulting from this competitor (written in (\ref{tildeexc}) in terms of a suitable $E_{\mathrm{UEG},s}$ value) will be close enough to the minimizing optimal transport value  ${\cal F}_{N,s}(\mu)$ with the "correct marginal" $\mu$ so that their scaled-by-$N$ next-order constant values are equal in the large $N$ limit, as is shown in Section \ref{reinstate0}. In order to finalize the proof of the Main Theorem for $d-2<s<d$, the corresponding $E_{\mathrm{UEG},s}$ problem from (\ref{tildeexc}) for the competitor with the "wrong" marginal will be compared in Lemma \ref{comparisoneindejel} to the associated quantity for the $E_{\mathrm{Jel},s}$ problem. Their scaled-by-$N$ difference will be shown to be negligible for $d-2<s<d$ in the large $N$ limit (but crucially not negligible for $s=d-2$), which will allow to conclude the result for $d-2<s<d$.

\subsubsection{Construction of a competitor to an optimal transport problem with the "wrong" marginal}
\label{contrucompet}
In this section we construct, in (\ref{defgammanr1}), a suitable competitor to an optimal transport problem with the "wrong" marginal, denoted by $\mu_{N,R_1}$ and defined in (\ref{defmunr1}).

We start by fixing the scales $1\ll R_1\ll R$ with the integrality and divisibility constraints\footnote{Note that in this section and in part of the next one, we do not use the full condition $(R_1/2)^d\in\N$, and we could work with only $R_1^d\in\N$. However for consistency with the rest of the paper we keep the former, slightly more restrictive, condition throughout.}
\begin{equation}
\label{constrRR1}
(R_1/2)^d\in \mathbb N,\qquad R^d=N\in\mathbb N,\qquad R/R_1\in\mathbb N,
\end{equation}
and we let $\nu_{R_1}\in\mathsf{Config}$ be a $(R_1\mathbb Z)^d$-periodic configuration with $R_1^d$ points per fundamental domain, as appearing in \eqref{periodicfields}. Then $N=R^d$ is the number of points in $\nu_{R_1}|_{K_R}$. We define $\gamma_{N,R_1}\in\mathcal P_{sym}^N(\mathbb R^d)$ to be the symmetric probability measure which gives equal weight to all translations by $x\in K_{R_1}$ of $\mathsf{set}(\nu_{R_1})\cap K_R$, as follows:
We consider the probability measure $N^{-1}\mu_{N,R_1}\in\mathcal P(\mathbb R^d)$, where $\mu_{N,R_1}$ is the positive measure obtained by averaging the configuration $\nu_{R_1}|_{K_R}$ by translations, and is defined by:
\begin{equation}\label{defmunr1}
 \mu_{N,R_1}:=\frac{1}{|K_{R_1}|}\int_{K_{R_1}}\sum_{p'\in\mathsf{set}(\nu_{R_1})\cap K_R}\delta_{p'+x} \diff x.
\end{equation}
We then define a symmetric transport plan $\gamma_{N,R_1}\in\mathcal P^N_{sym}(\mathbb R^d)$, by:
\begin{equation}\label{defgammanr1}
\gamma_{N,R_1}:=\frac{1}{|K_{R_1}|}\int_{K_{R_1}} \delta_{\mathsf{sym}}((\mathsf{set}(\nu_{R_1})\cap K_R)+x)\ \diff x,
\end{equation}
where for $A\subset\mathbb R^d$ and $x\in\mathbb R^d$, the set $A+x$ is the translation of $A$ by $x$, and then to a multi-set $\{x_1,\ldots,x_N\}$ we associate, as already introduced and defined in (\ref{symmetrize001}), the symmetrized sum of Dirac measures concentrated on the corresponding point $N$-tuples with permuted coordinates:
\begin{equation}\label{symmetrize}
\delta_\mathsf{sym}(\{x_1,\ldots,x_N\}):=\frac{1}{N!}\sum_{\sigma\in\mathsf{Perm}_N}\delta_{(x_{\sigma(1)},\ldots,x_{\sigma(N)})},
\end{equation}
where we recall that $\mathsf{Perm}_N$ is the permutations group with $N$ elements. 

We claim that $\gamma_{N,R_1}$ from \eqref{defgammanr1} has marginal $N^{-1}\mu_{N,R_1}$ defined in \eqref{defmunr1}. To show this, we first note that if we write the multiset $\mathsf{set}(\nu_{R_1})\cap K_R$ as $\{p_1,\ldots,p_N\}$, then formulas \eqref{defgammanr1} and \eqref{symmetrize} give that by definition, for any $f\in C^0_c((\mathbb R^d)^N)$ there holds
\begin{equation}\label{integral}
\int_{(\mathbb R^d)^N} f \diff\gamma_{N,R_1} = \frac{1}{|K_{R_1}|}\int_{K_{R_1}}\left[\frac{1}{N!}\sum_{\sigma\in\mathsf{Perm}_N}f(p_{\sigma(1)}+x,\ldots,p_{\sigma(N)}+x)\right] \diff x.
\end{equation}
Note now that $\mathsf{Perm}_N=\bigcup_{j=1}^N\left(\mathsf{Perm}_N\cap\{\sigma:\ \sigma(1)=j\}\right)$ is a partition, and that for fixed $1\le j\le N$ the set $\{(\sigma(2),\ldots,\sigma(N)):\ \sigma\in\mathsf{Perm}_N, \sigma(1)=j\}$ contains exactly all permutations of $(2,\ldots,N)$, and thus has cardinality $(N-1)!$. This observation allows to obtain from \eqref{integral}, for $f(x_1,\ldots,x_N):=g(x_1)$ (i.e. $f:=g\circ\pi_1$) for any $g\in C^0_c(\mathbb R^d)$ that
\begin{eqnarray*}
 \int_{\mathbb R^d}g\ \diff(\pi_1)_\#\gamma_{N,R_1}&=&\int_{(\mathbb R^d)^N}f\ \diff\gamma_{N,R_1} = \frac{1}{|K_{R_1}|}\int_{K_{R_1}}\left[\frac{1}{N!}\sum_{\sigma\in\mathsf{Perm}_N}g(p_{\sigma(1)}+x)\right] \diff x\\
 &=&\frac{1}{|K_{R_1}|}\int_{K_{R_1}}\left[\frac{1}{N}\sum_{i=1}^N g(p_j+x)\right] \diff x = \frac{1}{N}\int_{\mathbb R^d}g\ \diff\mu_{N,R_1},
\end{eqnarray*}
which shows that indeed \eqref{defgammanr1} has as marginal $N^{-1}$ times the measure \eqref{defmunr1}.  

Note that the set $\mathsf{set}(\nu_{R_1})\cap K_R$ is, using the periodicity of $\nu_{R_1}$, just the union of $R/R_1$ distinct $(R_1\mathbb Z)^d$-translated copies of the charges $\mathsf{set}(\nu_{R_1})\cap\Omega_{R_1}$ where $\Omega_{R_1}$ is a fundamental region of the lattice $(R_1\mathbb Z)^d$.

If as above $\mu_{N,R_1}/N\in\mathcal P(\mathbb R^d)$ is the marginal of $\gamma_{N,R_1}$, we let $\rho_{N,R_1}(x)/N$ be its density. Note that the measure $\mu_{N,R_1}$ is not a probability measure, and has mass $N$, as in \eqref{symmetrize} we have an average of mass-$N$ measures, and in \eqref{defgammanr1} we again take an average of such quantities, obtaining a measure whose marginal is $\mu_{N,R_1}$. Then $\rho_{N,R_1}$ is a perturbation of $1_{K_R}$, with density between $0$ and $1$, which differs from $1_{K_R}$ only on the neighbourhood $K_{R+R_1}\setminus K_{R-R_1}$ of $\partial K_R$, due to the averaging.

We find that, using the measure $\gamma_{N,R_1}$ as a competitor to $E_N^{xc}(\mu_{N,R_1}/N)$ as well as the formula \eqref{defgammanr1} and the fact that the energy of a configuration does not change if we translate it, gives the following upper bound (which we emphasize that is valid for general configurations $\nu_{R_1}$ as above, and does not use a minimality property of $\nu_{R_1}$):
\begin{eqnarray}\label{tildeexc}
E_{N,s}^{xc}\left(\frac{\mu_{N,R_1}}{N}\right)&\le& \int_{(\mathbb R^d)^N}\sum_{1\le i\neq j\le N}\frac{1}{|x_i-x_j|^s}\diff\gamma_{N,R_1}(x_1,\ldots,x_N) - N^2\int_{\mathbb R^d}\int_{\mathbb R^d}\frac{\rho_{N,R_1}(x)}{N}\frac{\rho_{N,R_1}(y)}{N}\frac{\diff x\ \diff y}{|x-y|^s}\nonumber\\
&=&\frac{1}{|K_{R_1}|}\int_{K_{R_1}}\sum_{p\neq q\in ((\mathsf{set}(\nu_{R_1})\cap K_R)+x)}\frac{1}{|p-q|^s} \diff x - \int_{\mathbb R^d}\int_{\mathbb R^d}\frac{\rho_{N,R_1}(x)\rho_{N,R_1}(y)}{|x-y|^s}\diff x\ \diff y\nonumber\\
&=&\frac{1}{|K_{R_1}|}\int_{K_{R_1}}\sum_{p'\neq q'\in \mathsf{set}(\nu_{R_1})\cap K_R}\frac{1}{|p'-q'|^s} \diff x - \int_{\mathbb R^d}\int_{\mathbb R^d}\frac{\rho_{N,R_1}(x)\rho_{N,R_1}(y)}{|x-y|^s}\diff x\ \diff y\nonumber\\
&=&\sum_{p'\neq q'\in \mathsf{set}(\nu_{R_1})\cap K_R}\frac{1}{|p'-q'|^s} -  \int_{\mathbb R^d}\int_{\mathbb R^d}\frac{\rho_{N,R_1}(x)\rho_{N,R_1}(y)}{|x-y|^s}\diff x\ \diff y\nonumber\\
&=&E_{\mathrm{UEG},s}(\mu_{N,R_1},\nu_{R_1}|_{K_R}).
\end{eqnarray}

\subsubsection{Gap bounds with incorrect marginals for the constructed transport plan}
\label{gapbound0}
We will work once more in this section with $R,R_1$ satisfying \eqref{constrRR1}. We compare in this section the expression $E_{\mathrm{UEG},s}(\mu_{N,R_1},\nu_{R_1}|_{K_R})$ on the right-hand side in \eqref{tildeexc} with the one for $E_{\mathrm{UEG},s}(K_R,\nu_{R_1}|_{K_R})$ and for $E_{\mathrm{Jel},s}(K_R,\nu_{R_1}|_{K_R})$, again for general configurations $\nu_{R_1}$, not necessarily minimizing. As above, the configurations $\nu_{R_1}$ we use will satisfy $(R_1\mathbb Z)^d$-periodicity, and the zero barycenter condition from \eqref{perscreenbar}. These properties of $\nu_{R_1}$ imply that differences of the charge distributions from our three quantities will have zero quadrupole moments. The comparison, done in Lemma \ref{comparisoneindejel} below, shows that these differences, scaled by $N$, vanish as $N\to\infty$ for $d-2<s<d$.

Note first that
\begin{equation}\label{reexpressmunr1}
{d}\mu_{N,R_1}{(x)} = \nu_{R_1}^1*1_{K_R}(x) \diff x, \quad\text{ where }\quad \nu_{R_1}^1:=\frac{1}{|K_{R_1}|}\sum_{p\in \mathsf{set}(\nu_{R_1})\cap K_{R_1}}\delta_p.
\end{equation}
To prove \eqref{reexpressmunr1} note that up to applying a change of coordinates, we may denote by abuse of notation 
\begin{equation}\label{abusenotcubes}
K_R:=[0,R)^d, \quad K_{R_1}:=[0,R_1)^d.
\end{equation}
Then we see that the cubes $p+K_{R_1}$, with $p\in K_R\cap (R_1\mathbb Z)^d$, are disjoint and tile $K_R$ if $R/R_1\in\mathbb N$.

\par We decompose next the atomic and absolutely continuous charge distributions into contributions coming from each $p+K_{R_1}$ as $p$ varies in $K_R\cap (R_1\mathbb Z)^d$, namely
\begin{equation}\label{secondquant}
\text{define }\left\{\begin{array}{l}
\nu_p:= \sum_{q\in\textsf{set}(\nu_{R_1})\cap (K_{R_1}+p)} \delta_q,\\[3mm]
d\mu_p(x):= 1_{K_{R_1}+p}(x)\diff x,
\end{array}\right.
\text{and note that }
\left\{
\begin{array}{l}
\textsf{set}(\nu_{R_1})\cap K_R=\bigcup_{p\in K_R\cap (R_1\mathbb Z)^d} \textsf{set}(\nu_p), \\[2mm]
\nu_{R_1}|_{K_R}=\bigcup_{p\in K_R\cap (R_1\mathbb Z)^d} \nu_p,\\[2mm]
1_{K_R}(x)\diff x=\sum_{p\in K_R\cap (R_1\mathbb Z)^d}\mu_p.
\end{array}
\right.
\end{equation}
The equalities on the right-hand side in \eqref{secondquant} follow because the sets $p+K_{R_1}, p\in K_R\cap (R_1\mathbb Z)^d,$ form a partition of $K_R$, a fact which can be proved for $d=1$ directly, and for general $d$ follows by considering each coordinate of points in $K_R$ separately.

Then, via \eqref{defmunr1} and by the assumed $(R_1\mathbb Z)^d$-periodicity of $\nu_{R_1}$, we find 
\begin{eqnarray*}
\mu_{N,R_1}&=&\frac{1}{|K_{R_1}|}\int_{K_{R_1}}\sum_{p'\in\mathsf{set}(\nu_{R_1})\cap K_R}\delta_{p'+x} \diff x{\stackrel{\text{\eqref{secondquant}}}{=}}\frac{1}{|K_{R_1}|} \sum_{q\in(R_1\mathbb Z)^d\cap K_R}\int_{K_{R_1}}\sum_{p'\in\mathsf{set}(\nu_{R_1})\cap (K_{R_1}+q)}\delta_{p'+x} \diff x\\
&=&\frac{1}{|K_{R_1}|} \sum_{q\in(R_1\mathbb Z)^d\cap K_R}\int_{K_{R_1}}\sum_{p\in\mathsf{set}(\nu_{R_1})\cap K_{R_1}}\delta_{p+q+x} \diff x\\
&=&{\frac{1}{|K_{R_1}|}\sum_{q\in(R_1\mathbb Z)^d\cap K_R}\int_{K_{R_1}+q}\sum_{p\in\mathsf{set}(\nu_{R_1})\cap K_{R_1}}\delta_{p+x'} \diff x'}{\stackrel{\text{\eqref{secondquant}}}{=}}\frac{1}{|K_{R_1}|}\int_{K_R}\sum_{p\in\mathsf{set}(\nu_{R_1})\cap K_{R_1}}\delta_{p+x} \diff x \\
&=&\frac{1}{|K_{R_1}|}\left(\sum_{p\in\mathsf{set}(\nu_{R_1})\cap K_{R_1}}\delta_{p}\right)*\mu_{K_R}{\stackrel{\text{\eqref{reexpressmunr1}}}{=}}\nu_{R_1}^1*\mu_{K_R},
\end{eqnarray*}
where $d\mu_{K_R}(x):=1_{K_R}(x)\diff x$, which proves \eqref{reexpressmunr1}. For the third equality in the above, we used that by $(R_1\mathbb Z)^d$-periodicity, for $q\in(R_1\mathbb Z)^d$ there holds $\mathsf{set}(\nu_{R_1})\cap (K_{R_1}+q)=\mathsf{set}(\nu_{R_1})+q$. This identity holds, because the intersection of a periodic set with a translated periodicity cell produces the same as the intersection of the basic cell, translated. For the penultimate inequality in the above, we applied that for all $f\in C_c^0(\mathbb{R}^d)$ there needs to hold $\int f\ \diff(\delta_p*\mu_{K_R}):=\int_{K_R} f(x+p)\ \diff x=\int_{z\in\mathbb R^d} f(z)\int_{K_R} \delta_{p+x}(z)\ \diff x$, which says then that $\delta_p*\mu_{K_R}=\int_{K_R}\delta_{p+x}\diff x$. Summing this over $p\in\mathsf{set}(\nu_{R_1})\cap K_{R_1}$ and dividing by $|K_{R_1}|$, gives the equality.

Suppose from now on that $\nu_{R_1}$ has the barycenter in the origin (or equivalently, it is symmetric with respect to the origin), i.e.
\begin{equation}\label{baryor}
\int_{K_{R_1}}x \diff\nu_{R_1}(x)=0.
\end{equation}
Note that this condition is up to a translation the zero barycenter condition as appearing in \eqref{perscreenbar}.

Similarly to \eqref{fundsolrd+1}, for $\mu$ a locally finite (i.e. finite on compact sets) measure on $\mathbb R^d$ such that the integral below converges we define now for $x\in\mathbb R^d$
\begin{equation}\label{defw}
h^\mu(x):=\int_{\mathbb R^d}\mathsf c(x-x')\diff\mu(x').
\end{equation}
We then have the following:
\begin{lemma}\label{comparisoneindejel}
Set $0\le d-2<s<d$. Let $E_{\mathrm{Jel},s}, E_{\mathrm{UEG},s}$ be defined as in \eqref{defjellium} and \eqref{defeind}. With the notation \eqref{defjelliumind1} and if for $R_1^d, R^d\in\mathbb N$ such that $R/R_1\in\mathbb N$ the configuration $\nu_{R_1}\in\mathsf{Config}$ is $(R_1\mathbb Z)^d$-periodic, has $R_1^d$ points per fundamental domain and satisfies \eqref{baryor}, then there holds for $\diff \mu_{K_R}(x):=1_{K_R}(x)\diff x$
\begin{equation}\label{comparisoneq1}
\lim_{\substack{N=R^d\to\infty\\ R/R_1\in\mathbb N}}\frac{ E_{\mathrm{UEG},s}(K_R, \nu_{R_1}|_{K_R}) -  E_{\mathrm{Jel},s}(K_R, \nu_{R_1}|_{K_R})}{N} = \frac{{2}}{|K_{R_1}|}\int_{\mathbb R^d} h^{\nu_{R_1}|_{K_{R_1}} - \mu_{K_{R_1}}}(x) \diff x,
\end{equation}
and 
\begin{equation}\label{comparisoneq2}
\lim_{\substack{N=R^d\to\infty\\ R/R_1\in\mathbb N}}\frac{ E_{\mathrm{UEG},s}(\mu_{N,R_1}, \nu_{R_1}|_{K_R}) -  E_{\mathrm{UEG},s}(K_R, \nu_{R_1}|_{K_R})}{N} = \frac{1}{|K_{R_1}|}\int_{\mathbb R^d} {h^{(\delta_0-\nu_{R_1}^1*\nu_{R_1}^1)*\mu_{K_{R_1}}}(x)} \diff x.
\end{equation}
\end{lemma}
\begin{proof}
Note that due to the translation-invariance of our energies, by abuse of notation we may again denote for the duration of this proof
\[
K_R:=[0,R)^d, \quad K_{R_1}:=[0,R_1)^d,
\]
like in \eqref{abusenotcubes}, and then we may use the subdivision \eqref{secondquant}. We prove in detail \eqref{comparisoneq1} and only give a sketch of the proof for \eqref{comparisoneq2}, as it is proved exactly the same way. Using the periodic decomposition from formulas \eqref{secondquant} and the bi-linearity of $\langle\cdot,\cdot\rangle_s$, as well as \eqref{whatever} and \eqref{noselfint}, and recalling \eqref{defjelliumind1}
\begin{eqnarray*}
\lefteqn{ E_{\mathrm{UEG},s}(K_R, \nu_{R_1}|_{K_R}) -  E_{\mathrm{Jel},s}(K_R, \nu_{R_1}|_{K_R})}\nonumber\\
&=&\left\langle \nu_{R_1}|_{K_R}, \nu_{R_1}|_{K_R}\right\rangle_s^* - \left\langle \mu_{K_R}, \mu_{K_R}\right\rangle_s^* -\left\langle \nu_{R_1}|_{K_R} - \mu_{K_R},\nu_{R_1}|_{K_R} - \mu_{K_R}\right\rangle_s^*\nonumber\\
&=&{2}\left\langle \nu_{R_1}|_{K_R}- \mu_{K_R},\mu_{K_R}  \right \rangle_s^*={2}\left\langle \nu_{R_1}|_{K_R}- \mu_{K_R},\mu_{K_R}\right \rangle_s \text{ (because $\mu_{K_R}$ is absolutely continuous)}
\end{eqnarray*}
Next, we claim that
\begin{equation}
\label{provcomp1}
 E_{\mathrm{UEG},s}(K_R, \nu_{R_1}|_{K_R}) -  E_{\mathrm{Jel},s}(K_R, \nu_{R_1}|_{K_R})=2h^{\nu_{R_1}|_{K_{R_1}} - \mu_{K_{R_1}}}*\left(\sum_{p\in K_R\cap(R_1\mathbb Z)^d}1_{K_R-p}\right){(0)}.
\end{equation}

To prove \eqref{provcomp1} we first recall the notation \eqref{tau_p} and apply once more the subdivision \eqref{secondquant}, where $\lambda$ can be a general finite measure such that all the terms below are finite:
%
\begin{eqnarray}\label{translationstuff}
\langle \mu_{K_R}, \lambda*\mu_{K_R}\rangle_s&=&\sum_{p\in K_R\cap(R_1\mathbb Z)^d}\langle \mu_{K_R}, \lambda * \mu_{K_{R_1}+p}\rangle_s =\sum_{p\in K_R\cap(R_1\mathbb Z)^d}\langle \mu_{K_R},\lambda * ((\tau_p)_\#\mu_{K_{R_1}})\rangle_s \nonumber\\
&=& \sum_{p\in K_R\cap(R_1\mathbb Z)^d}\langle \mu_{K_R}, (\tau_p)_\#\left(\lambda * \mu_{K_{R_1}}\right)\rangle_s = \sum_{p\in K_R\cap(R_1\mathbb Z)^d}\langle (\tau_{-p})_\#\mu_{K_R}, \lambda * \mu_{K_{R_1}}\rangle_s\nonumber\\
&=&\sum_{p\in K_R\cap(R_1\mathbb Z)^d}\langle \mu_{K_R-p}, \lambda * \mu_{K_{R_1}}\rangle_s= h^{\lambda*\mu_{K_{R_1}}}*\left(\sum_{p\in K_R\cap(R_1\mathbb Z)^d}1_{K_R-p}\right)(0),
\end{eqnarray}
where we used the bilinearity of $\langle\cdot,\cdot\rangle_s$, the fact that $\mu_{A+p}=(\tau_p)_\#\mu_A$ if $\diff\mu_A(x):=1_A(x)\diff x$, the fact that convolution of measures satisfies $(\tau_p)_\#(\mu*\nu)=\mu*((\tau_p)_\#\nu)$, the fact that $\langle\cdot,\cdot\rangle_s$ is invariant under the action of translations, and the definition \eqref{defw} of $h^{\lambda*\mu_{K_{R_1}}}$.

Moreover, from \eqref{secondquant}, by applying the definition of $\nu_{R_1}|_{K_R}, \nu_{R_1}|_{K_{R_1}}$, and by using the $(R_1\mathbb Z)^d$-periodicity of $\nu_{R_1}$, we obtain
\begin{eqnarray*}
\nu_{R_1}|_{K_R}&=&\sum_{p\in\textsf{set}(\nu_{R_1})\cap K_R}\nu_p=\sum_{p'\in K_R\cap (R_1\mathbb Z)^d}\,\, \sum_{p\in\textsf{set}(\nu_{R_1})\cap (K_{R_1}+p')} \delta_p\\
&=&\sum_{p'\in K_R\cap (R_1\mathbb Z)^d}\,\, \sum_{q\in\textsf{set}(\nu_{R_1})\cap  K_{R_1}} \delta_{q+p'}=\sum_{p'\in K_R\cap (R_1\mathbb Z)^d}(\tau_p)_\#(\nu_{R_1}|_{K_{R_1}}).
\end{eqnarray*}
For the third equality in the above, we used once more that by $(R_1\mathbb Z)^d$-periodicity, for $p'\in(R_1\mathbb Z)^d$ there holds $\mathsf{set}(\nu_{R_1})\cap (K_{R_1}+p')=\mathsf{set}(\nu_{R_1}|_{K_{R_1}})+p'$. Therefore, exactly as in the above calculation \eqref{translationstuff}, with $\lambda=\delta_0$ and using instead of the decomposition $\mu_{K_R}=\sum_{p\in(R_1\mathbb Z)^d}\mu_{K_{R_1}+p}$ the decomposition $\nu_{R_1}|_{K_R}=\sum_{p\in K_R\cap (R_1\mathbb Z)^d}(\tau_p)_\#(\nu_{R_1}|_{K_{R_1}})$, we find
\begin{eqnarray}
\label{transplus}
\left\langle \mu_{K_R}, \nu_{R_1}|_{K_R} \right \rangle_s&=&\left\langle \mu_{K_R}, \sum_{p\in K_R\cap (R_1\mathbb Z)^d}(\tau_p)_\#(\nu_{R_1}|_{K_{R_1}})\right\rangle_s\nonumber\\
&=&h^{\nu_{R_1}|_{K_{R_1}}}*\left(\sum_{p\in K_R\cap(R_1\mathbb Z)^d}1_{K_R-p}\right){(0)}.
\end{eqnarray}
Together with \eqref{translationstuff}, the above proves \eqref{provcomp1}. 

We now observe that for each $x\in\mathbb R^d$ there holds
\begin{equation}
\label{bounddomconv}
0\le \frac{\sum_{p\in K_R\cap(R_1\mathbb Z)^d}1_{K_R-p}(x)}{(R/R_1)^d}\le 1, \qquad \lim_{\substack{N=R^d\to\infty\\ R/R_1\in\mathbb N}}\frac{\sum_{p\in K_R\cap(R_1\mathbb Z)^d}1_{K_R-p}(x)}{(R/R_1)^d}=1.
\end{equation}
Thus by the Dominated Convergence Theorem we find from \eqref{provcomp1} that
\begin{equation}
\label{domconv}
\lim_{\substack{N=R^d\to\infty\\ R/R_1\in\mathbb N}} \frac{1}{N}h^{\nu_{R_1}|_{K_{R_1}} - \mu_{K_{R_1}}}*\left(\sum_{p\in K_R\cap(R_1\mathbb Z)^d}1_{K_R-p}\right)(0)= \frac{1}{|K_{R_1}|}\int_{\mathbb R^d}h^{\nu_{R_1}|_{K_{R_1}} - \mu_{K_{R_1}}}(x)\diff x, 
\end{equation}
which proves \eqref{comparisoneq1}, as desired.

To be able to apply the Dominated Convergence Theorem in \eqref{domconv}, it suffices to show that 
\begin{equation}\label{domconv_toprove}
\int_{\R^d}\left|h^{\nu_{R_1}|_{K_{R_1}} - \mu_{K_{R_1}}}(x)\right|\diff x<\infty.
\end{equation}
We first note the Taylor (or multipole) expansion
\begin{equation}\label{taylor_powerlaw}
\frac1{|x-y|^s} = \frac1{|x|^s} - s \frac{\langle x,y\rangle}{|x|^{s+2}} + \frac{{\cal R}(x,y)}{|x|^{s+4}}, \quad |{\cal R}(x,y)|\le \left(s(s+2)+1\right)|y|^2|x|^2.
\end{equation}
Now recall that $\nu_{R_1}|K_{R_1}$ and $\mu_{K_{R_1}}$ are both positive measures supported on $K_{R_1}$, and we may also write 
\begin{equation}\label{express_hnu}
\left|h^{\nu_{R_1}|_{K_{R_1}} - \mu_{K_{R_1}}}(x)\right|=\left|\sum_{p\in\mathsf{set}(\nu_{R_1})\cap K_{R_1}} \frac1{|x-p|^s}- \int_{K_{R_1}}\frac1{|x-y|^s}\diff y\right|,
\end{equation}
which we treat differently in the cases $|x|\le C R_1$ and $|x|>C R_1$, where $C\ge \sqrt d$ is a constant such that the diameter of $K_{R_1}$ is smaller than $C R_1$. In the first case, continuing \eqref{express_hnu}, and using the triangle inequality, the fact that $\nu_{R_1}|K_{R_1}$ and $1_{K_{R_1}}(x)\diff x$ both have mass $|K_{R_1}|$, and the fact that $\max\{|x-p|:\ |x|\le CR_1, p\in K_{R_1}\}\le 2CR_1$, we write:
\begin{eqnarray}\label{domconv_est1}
\lefteqn{\int_{|x|\le C R_1}\left|\sum_{p\in\mathsf{set}(\nu_{R_1})\cap K_{R_1}} \frac1{|x-p|^s}- \int_{K_{R_1}}\frac1{|x-y|^s}\diff y\right|\diff x}\nonumber\\
&\le&\sum_{p\in\mathsf{set}(\nu_{R_1})\cap K_{R_1}} \int_{|x|\le C R_1}\frac1{|x-p|^s}\diff x + \int_{K_{R_1}}\int_{|x|\le C R_1}\frac1{|x-y|^s}\diff x\ \diff y\nonumber\\
&\le& 2 |K_{R_1}| \int_{|x|\le 2C R_1}\frac1{|x|^s}\diff x \le 2 |K_{R_1}| (d-s)^{-1}|\mathbb S^{d-1}| (2C R_1)^{d-s}.
\end{eqnarray}
For the remaining region $\{|x|> C R_1\}$ we use \eqref{taylor_powerlaw} and the fact that $K_{R_1}$ and $\nu_{R_1}|K_{R_1}$ are both having the same mass and are symmetric with respect to the origin. More precisely,  the first term 
$$\sum_{p\in\mathsf{set}(\nu_{R_1})\cap K_{R_1}} \frac1{|x|^s}- \int_{K_{R_1}}\frac1{|x|^s}\diff y$$
in the expansion  \eqref{taylor_powerlaw} applied to \eqref{express_hnu} cancels since $\#(\nu_{R_1}\cap K_{R_1})=|K_{R_1}|$, whereas the second term 
$$-s\sum_{p\in\mathsf{set}(\nu_{R_1})\cap K_{R_1}} \frac{\langle x,p\rangle}{|x|^{s+2}}+ s\int_{K_{R_1}}\frac{\langle x,y\rangle}{|x|^{s+2}}\diff y$$
cancels by the balancing condition \eqref{baryor} assumed on $\nu_{R_1}$ and by the symmetry with respect to the origin of $K_{R_1}$. Thus the first two terms of the Taylor expansion \eqref{taylor_powerlaw}, when applied to the right-hand side of \eqref{express_hnu} vanish. Using now that $\max\{|y|:\ y\in K_{R_1}\}\le C R_1$ and the bound on ${\cal R}(x,y)$ from \eqref{taylor_powerlaw}, we can write:
\begin{eqnarray}\label{domconv_est2}
\lefteqn{\int_{|x|> C R_1}\left|\sum_{p\in\mathsf{set}(\nu_{R_1})\cap K_{R_1}} \frac1{|x-p|^s}- \int_{K_{R_1}}\frac1{|x-y|^s}\diff y\right|\diff x}\nonumber\\
&=&\int_{|x|>C R_1}\left|\sum_{p\in\mathsf{set}(\nu_{R_1})\cap K_{R_1}}\frac{{\cal R}(x,p)}{|x|^{s+4}} -\int_{K_{R_1}}\frac{{\cal R}(x,y)}{|x|^{s+4}}\diff y\right|\diff x\nonumber\\
&\le&\int_{|x|>C R_1}\sum_{p\in\mathsf{set}(\nu_{R_1})\cap K_{R_1}}\frac{|{\cal R}(x,p)|}{|x|^{s+4}} \diff x+\int_{|x|>C R_1}\int_{K_{R_1}}\frac{|{\cal R}(x,y)|}{|x|^{s+4}}\diff y\diff x\nonumber\\
&\le&2(s^2+2s+1)(CR_1)^2|K_{R_1}|\int_{|x|>C R_1}\frac1{|x|^{s+2}}\diff x\nonumber\\
&\le&2\frac{s^2+2s+1}{s+2-d}(CR_1)^2|K_{R_1}|\mathbb S^{d-1}|(CR_1)^{s+2-d}.
\end{eqnarray}
Summing up the bounds \eqref{domconv_est1} and \eqref{domconv_est2} we find the desired bound \eqref{domconv_toprove}.

\medskip

Moving now to the proof of \eqref{comparisoneq2}, we first observe that in the comparison below the interactions of atomic measures cancel and we find, in the notation \eqref{notationduality}, and recalling \eqref{reexpressejeueg} and \eqref{defjelliumind1} 
\begin{eqnarray}\label{compindind}
\lefteqn{E_{\mathrm{UEG},s}(\mu_{N,R_1},\nu_{R_1}|_{K_R}) -  E_{\mathrm{UEG},s}(K_R,\nu_{R_1}|_{K_R})}\nonumber\\
&=&\langle \nu_{R_1}|_{K_R},\nu_{R_1}|_{K_R} \rangle_s- \langle \mu_{N,R_1}, \mu_{N,R_1}\rangle_s -\langle \nu_{R_1}|_{K_R},\nu_{R_1}|_{K_R} \rangle_s+\langle \mu_{K_R} \diff x,\mu_{K_R} \diff x\rangle_s\nonumber\\
 &=& -\langle \mu_{N,R_1}, \mu_{N,R_1}\rangle_s +\langle \mu_{K_R} \diff x,\mu_{K_R} \diff x\rangle_s= -\left\langle  \nu_{R_1}^1*\mu_{K_R},\nu_{R_1}^1*\mu_{K_R}\right\rangle_s + \langle \mu_{K_R} ,\mu_{K_R} \rangle_s\nonumber\\
&=&\left\langle (\delta_0 - \nu_{R_1}^1*\nu_{R_1}^1)*\mu_{K_R}, \mu_{K_R}\right\rangle_s=h^{(\delta_0 -\nu_{R_1}^1* \nu_{R_1}^1)*\mu_{K_{R_1}}} *\left(\sum_{p\in(R_1\mathbb Z)^d\cap K_R}1_{K_R-p}\right)(0),
\end{eqnarray}
where the last equality in the above follows by using for $\lambda=\delta_0 -\nu_{R_1}^1*\nu_{R_1}^1$ the same argument as the one used to obtain \eqref{translationstuff}.

By means of the above, of \eqref{bounddomconv} and of the dominated convergence theorem, we get 
$$\lim_{\substack{N=R^d\to\infty\\ R/R_1\in\mathbb N}}\frac{1}{N}h^{(\delta_0 -\nu_{R_1}^1* \nu_{R_1}^1)*\mu_{K_{R_1}}} *\left(\sum_{p\in(R_1\mathbb Z)^d\cap K_R}1_{K_R-p}\right)(0)=\frac{1}{|K_{R_1}|}\int h^{(\delta_0 -\nu_{R_1}^1*\nu_{R_1}^1)*\mu_{K_{R_1}}}(x)\diff x ,$$
which proves \eqref{comparisoneq2}.

\medskip

Here the Dominated Convergence Theorem is proved similarly to the case \eqref{bounddomconv}. We now aim to prove that 
\begin{equation}\label{domconv_toprove2}
\int \left|h^{(\delta_0 -\nu_{R_1}^1*\nu_{R_1}^1)*\mu_{K_{R_1}}}(x)\right|\diff x<\infty.
\end{equation}
In this case we again subdivide the domain into the regions where $|x|\le C R_1$ and $|x|>C R_1$, with a choice of $C$ such that the first region contains the support of $\nu_{R_1}^1*\nu_{R_1}^1*1_{K_{R_1}}$. Then we note again that both measures $\mu_{K_{R_1}}$ and $\nu_{R_1}^1*\nu_{R_1}^1*\mu_{K_{R_1}}$ are positive measures of total mass $|K_{R_1}|$, thus using the fact that 
\begin{equation}\label{express_hnu2}
\left|h^{(\delta_0 -\nu_{R_1}^1*\nu_{R_1}^1)*\mu_{K_{R_1}}}(x)\right|=\left|\int_{K_{R_1}}\frac1{|x-y|^s}\diff y - \int \frac1{|x-y|^s}(\nu_{R_1}^1*\nu_{R_1}^1*\mu_{K_{R_1}})(\diff y)\right|,
\end{equation}
we proceed exactly like in \eqref{domconv_est1} and bound
\begin{equation}\label{domconv2_est1}
\int_{|x|\le C R_1}\left|h^{(\delta_0 -\nu_{R_1}^1*\nu_{R_1}^1)*\mu_{K_{R_1}}}(x)\right| \diff x \le 2|K_{R_1}|(d-s)^{-1}|\mathbb S^{d-1}|(2CR_1)^{d-s}.
\end{equation}
Then we note that due to the fact that $\nu_{R_1}|K_{R_1}$ (and thus $\nu_{R_1}^1$) satisfies the balancing condition \eqref{baryor} and $K_{R_1}$ is symmetric with respect to the origin, we have also the zero-barycenter condition
\begin{equation}\label{balancing_nur11}
\int y (\nu_{R_1}^1*\nu_{R_1}^1*\mu_{K_{R_1}})(\diff y)=0,
\end{equation}
which then allows to cancel the first two terms in \eqref{taylor_powerlaw} when inserted in \eqref{express_hnu2} too, and thus gives the following bound analogous to \eqref{domconv_est1}, by the same reasoning, and now using the fact that $\sup\{|y|:\ y\in\mathsf{supp}(\nu_{R_1}^1*\nu_{R_1}^1*1_{K_{R_1}})\} \le CR_1$ for our new choice of $C$:
\begin{eqnarray}\label{domconv2_est2}
\lefteqn{\int_{|x|> C R_1}\left|\int_{K_{R_1}}\frac1{|x-y|^s}\diff y - \int\frac1{|x-y|^s}(\nu_{R_1}^1*\nu_{R_1}^1*\mu_{K_{R_1}})(\diff y)\right|\diff x}\nonumber\\
&=&\int_{|x|>C R_1}\left|\int_{K_{R_1}}\frac{{\cal R}(x,y)}{|x|^{s+4}}\diff y -\int\frac{{\cal R}(x,y)}{|x|^{s+4}} (\nu_{R_1}^1*\nu_{R_1}^1*\mu_{K_{R_1}})(\diff y)\right|\diff x\nonumber\\
&\le&\int_{|x|>C R_1}\int_{K_{R_1}}\frac{|{\cal R}(x,y)|}{|x|^{s+4}}\diff y\diff x + \int_{|x|>C R_1}\int\frac{|{\cal R}(x,y)|}{|x|^{s+4}} (\nu_{R_1}^1*\nu_{R_1}^1*\mu_{K_{R_1}})(\diff y)\diff x\nonumber\\
&\le&2(s^2+2s+1)(CR_1)^2|K_{R_1}|\int_{|x|>C R_1}\frac1{|x|^{s+2}}\diff x\nonumber\\
&\le&2\frac{s^2+2s+1}{s+2-d}(CR_1)^2|K_{R_1}|\mathbb S^{d-1}|(CR_1)^{s+2-d}.
\end{eqnarray}
Now \eqref{domconv2_est1} and \eqref{domconv2_est2} again give our desired bound \eqref{domconv_toprove2}.
\end{proof}
\begin{remark}
Before we proceed to the next statement, we observe here that (\ref{domconv_est2}) and (\ref{domconv2_est2}) above only hold for $d-2<s<d$, so this is a necessary condition for the results of the Lemma \ref{comparisoneindejel} to hold.
\end{remark}
Next, we show
\begin{lemma}
Under the same assumptions and with the same notations as in Lemma \ref{comparisoneindejel}, we have for $0\le d-2<s<d$
\begin{equation}
\label{mapjecon}
\int h^{\nu_{R_1}|_{K_{R_1}} - \mu_{K_{R_1}}}(x)\diff x =0~~~\mbox{and}~~~\int h^{(\delta_0 -\nu_{R_1}^1*\nu_{R_1}^1)*\mu_{K_{R_1}}}(x)\diff x=0.
\end{equation}
\end{lemma}
\begin{proof}
For the usual range of exponents $0\le d-2<s<d$, we have the following expression of the first integral from (\ref{mapjecon}) in terms of the Fourier transform (denoted here by $\mathcal F$) of $h^{\nu_{R_1} - 1_{K_{R_1}}}$: 
\begin{eqnarray}\label{wr1fourier}
\int h^{\nu_{R_1}|_{K_{R_1}} - \mu_{K_{R_1}}}(x)\diff x &=&\lim_{|\xi|\to 0}\mathcal F(h^{\nu_{R_1} - \mu_{K_{R_1}}})(\xi)=\lim_{|\xi|\to 0}\mathcal F(|x|^{-s}*(\nu_{R_1} - \mu_{K_{R_1}}))(\xi)\nonumber\\
&=&\lim_{|\xi|\to 0}\mathcal F(|x|^{-s})(\xi)\mathcal F(\nu_{R_1} - \mu_{K_{R_1}})(\xi)\nonumber\\
&=&\lim_{|\xi|\to 0} \frac{c_{s,d}}{|\xi|^{d-s}}\left(\sum_{p\in\mathsf{set}(\nu_{R_1})\cap K_{R_1}}e^{-i \xi\cdot p} - \int_{K_{R_1}}e^{-i \xi\cdot q}\diff q\right)=0,
\end{eqnarray}
where we use the following Taylor series expression:
\begin{equation}\label{taylorexpr}
e^{-i \xi\cdot p} = 1 - i \xi\cdot p - \frac{1}{2}|\xi\cdot p|^2 + O(|\xi\cdot p|^3),
\end{equation}
and we use again the fact that the contribution of the first above term cancels due to the fact that $\#(\mathsf{set}(\nu_{R_1})\cap K_{R_1})=|K_{R_1}|$, the second term cancels by the balancing condition \eqref{baryor} assumed on $\nu_{R_1}$ and by the symmetry with respect to the origin of $K_{R_1}$, whereas the remaining terms divided by $|\xi|^{d-s}$ with $d-s<2$ vanish in the limit $\xi\to 0$. 
\par By the same method we also find that
\begin{equation}\label{wr2fourier}
\int h^{(\delta_0 -\nu_{R_1}^1*\nu_{R_1}^1)*\mu_{K_{R_1}}}(x)\diff x = \lim_{|\xi|\to 0} \frac{c_{s,d}}{|\xi|^{d-s}}\left(\frac{1}{|K_{R_1}|^2}\sum_{p,q\in\mathsf{set}(\nu_{R_1})\cap K_{R_1}}\left(1-e^{-i \xi\cdot (p+q)}\right)\right)\int_{K_{R_1}}e^{-i \xi\cdot q}\diff q=0.
\end{equation}
\end{proof}

\begin{rmk}[Difficulty for the case $s=d-2$ and link to the Abrikosov conjecture]\label{rmk:diffabri}
For $s=d-2$ the factor $|\xi|^{s-d}$ in \eqref{wr1fourier} becomes equal to $|\xi|^{-2}$, which explodes as $\xi\to 0$ precisely in such a way that the square (higher-order) term in \eqref{taylorexpr} intervenes in the limit, and does not allow to say that the limit in \eqref{wr2fourier} is zero.

\medskip 

Configurations $\nu_{R_1}$ belonging to a fixed lattice yield an explicitly computable value of the contribution \eqref{comparisoneq1}, see \cite[App. B]{lewinlieb}. The fact that $E_{\mathrm{Jel},s}$-minimizing configurations are asymptotically lattice-like is not known as explained in Remark \ref{Cristconj} (in $d=2$ this is the so-called Abrikosov conjecture, which can be stated for $0<s<2$ as a generalization of the case of log-kernels corresponding to $s=0$ in that case). Furthermore, it is not known whether boundary effects occur for $E_{\mathrm{UEG},s}$ to make the minimizers for the optimal transport problem non-lattice-like. These open questions preclude us also from deducing that due to the quadratic term in \eqref{taylorexpr} the limit in \eqref{wr2fourier} is nonzero.
\end{rmk}

\subsubsection{Returning to the correct marginal} 
\label{reinstate0}

In this section we compare the asymptotics of $E_{N,s}^{\mathrm{xc}}\left(\frac{\mu_{N, R_1}}{N}\right)$ to $\textsf{C}_{\mathrm{UEG}}(s,d)$ by just directly using the asymptotics for $E_{N,s}^{\mathrm{xc}}$ proved in general, and the fact that $\textsf{C}_{\mathrm{UEG}}(s,d)$ is independent of the marginal.

Note that we cannot apply directly to $E_{N,s}^{\mathrm{xc}}\left(\frac{\mu_{N, R_1}}{N}\right)$ the limiting main result from \cite{cotpet} due to the dependence of the measure $\frac{\mu_{N, R_1}}{N}$ on $R$ and $R_1$ (which in turn both depend on $N$ and also on each other, and satisfy a number of additional constraints as explained in \eqref{choiceparam2} below). The strategy we will adopt instead will be to compare the aymptotics of $E_{N,s}^{\mathrm{xc}}\left(\frac{\mu_{N, R_1}}{N}\right)$ with those of $E_{(R+R_1)^d,s}^{\mathrm{xc}}\left(\frac{1_{K_{R+R_1}}}{(R+R_1)^d}\right)$ and $E_{(R-R_1)^d,s}^{\mathrm{xc}}\left(\frac{1_{K_{R-R_1}}}{(R-R_1)^d}\right)$, for appropriate values of $R,R_1$, and for which last two quantities we can easily eliminate the dependence on $N$, due to both $\frac{1_{K_{R+R_1}}}{(R+R_1)^d}$ and $\frac{1_{K_{R-R_1}}}{(R-R_1)^d}$ being densities corresponding to uniform measures. We note that by abuse of notation above, as well as in the following, we denote the measure $\diff \mu_{K_{R-R_1}}(x) := \frac{1_{K_{R-R_1}}(x)}{(R-R_1)^d}\diff x$ by the same notation as its density $\frac{1_{K_{R-R_1}}}{(R-R_1)^d}$.

The marginal $N^{-1}\mu_{N,R_1}$ of $\gamma_{N, R_1}$ has density $N^{-1}\rho_{N,R_1}$, which can be split as follows:
\begin{equation}\label{defrnr1}
N^{-1}\rho_{N,R_1}=\alpha_N \frac{1_{K_{R-R_1}}}{|K_{R-R_1}|}+(1-\alpha_N)g_N,
\end{equation}
where {$g_N:\mathbb R^d\to [0,1/(N(1-\alpha_N))]$ is a probability density supported on $K_{R+R_1}\setminus K_{R-R_1}$ , and $\alpha_N={\frac{|K_{R-R_1}|}{|K_R|}=(1-R_1/R)^d}$, with $1-CN^{-1/d}<\alpha_N<1$ and $\alpha_N\to 1$ as $R=N^{1/d}\rightarrow\infty$. 

Indeed, recall that by \eqref{reexpressmunr1} and since $\gamma_{N,R_1}$ has marginal $\mu_{N,R_1}/N$ and density $\rho_{N,R_1}/N$, we have 
\begin{equation}\label{reexprsum2}
\rho_{N,R_1}(x)\diff x=\diff \mu_{N,R_1}(x)=\diff\left(\nu_{R_1}^1*\mu_{K_R}\right)(x)=\left[\frac1{|K_{R_1}|}\sum_{p\in\mathsf{set}(\nu_{R_1})\cap K_{R_1}}1_{K_R+p}(x)\right]\diff x.
\end{equation}
Because $p\in K_{R_1}$, we find that the above summands can be nonzero only for $x\in K_{R_1}+K_R=K_{R+R_1}$. We also find that for any $p\in K_{R_1}$ there holds $K_{R-R_1}+p\subset K_{R-R_1}+K_{R_1}=K_R$, therefore for $x\in K_{R-R_1}$ the density from the right-hand side of \eqref{reexprsum2} receives $\#(\nu_{R_1}\cap K_{R_1})=|K_{R_1}|$ contributions, and is thus constantly equal to 1. This shows that the formula \eqref{defrnr1} holds true. Since for general points $x\in K_{R+R_1}$, it receives less than $|K_{R_1}|$, the density of $\rho_{N,R_1}$ is smaller than $1$, implying the bound on $g_N$.}

\medskip

We would like to show that that the $g_N$ term is negligible in the estimates.  In our computations the size and periodicity parameters $R_1,R,$ will again be satisfying
\begin{equation}\label{choiceparam2}
R^d=N, \quad (R_1/2)^d\in\mathbb N, \quad R/R_1\in\mathbb N,
\end{equation}
which assumptions will matter for the computations below. In this regime, we will show the following
\begin{lemma}\label{boundbelowmuR1}
Let $\max\{0,d-2\}\le s<d$. There holds for any sequences of $R_1$ and $R$ satisfying the regime \eqref{choiceparam2}
\begin{equation}\label{wished-for-proof}
\lim_{ R_1\to\infty\atop (R_1/2)^d\in\mathbb{N}}\liminf_{\substack{R=N^{1/d},N\rightarrow\infty\\R/R_1\in\N}}\frac{1}{N}E_{N,s}^{\mathrm{xc}}\left(\frac{\mu_{N,R_1}}{N}\right)=\textsf{C}_{\mathrm{UEG}}(s,d).
\end{equation}
In other words, if $R=R^{(N)}=N^{1/d}$ and $R_1$ independent of $N$ satisfy the convergence regime \eqref{choiceparam2}, then the limit of $\tfrac1N E_{N_s}^{\mathrm{xc}}(\tfrac{\mu_{N,R_1}}{N})$ equals $\textsf{C}_{\mathrm{UEG}}(s,d)$.
\end{lemma}
\begin{proof}

\par Note that $(R+R_1)^d, (R-R_1)^d\in\mathbb{N}$, in view of \eqref{choiceparam2}, which fact will be used below.

\par \textbf{Step 1:} We start by proving the '$\ge$' direction of (\ref{wished-for-proof}). More precisely, we claim that
\begin{equation}\label{wished-for-ge}
\lim_{ R_1\to\infty\atop (R_1/2)^d\in\mathbb{N}}\liminf_{\substack{R=N^{1/d},N\rightarrow\infty\\R/R_1\in\N}}\frac{1}{N}E_{N,s}^{\mathrm{xc}}\left(\frac{\mu_{N,R_1}}{N}\right)\ge \lim_{ R_1\to\infty\atop (R_1/2)^d\in\mathbb{N}}\liminf_{\substack{R=N^{1/d}, N\rightarrow\infty\\R/R_1\in\N}}\frac{1}{N}E_{(R+R_1)^d,s}^{\mathrm{xc}}\left(\frac{1_{K_{R+R_1}}}{(R+R_1)^d}\right)=\textsf{C}_{\mathrm{UEG}}(s,d).
\end{equation}
To prove \eqref{wished-for-ge} we first note that, from \eqref{defrnr1}, we have
\begin{equation*}
\rho_{N,R_1}=1_{K_{R-R_1}} +N(1-\alpha_N)g_N\le 1_{K_{R-R_1}}+1_{K_{R+R_1}\setminus K_{R-R_1}}=1_{K_{R+R_1}}.
\end{equation*}
Therefore, we can write $1_{K_{R+R_1}}=\rho_{N,R_1}+\rho'_{N,R+R_1}$, with the obvious definition for $\rho'_{N,R+R_1}\ge 0$, and where
\[
\int\rho'_{N,R+R_1}(x) \diff x=(R+R_1)^d-N.
\]
Consequently,
\begin{equation}\label{inspectrho}
\frac{1_{K_{R+R_1}}}{(R+R_1)^d}=\frac{N}{(R+R_1)^d}\frac{\rho_{N,R_1}}{N}+\frac{(R+R_1)^d-N}{(R+R_1)^d}\frac{\rho'_{N,R+R_1}}{(R+R_1)^d-N}\,, 
\end{equation}
where both $\frac{\rho_{N,R_1}}{N}$ and $\frac{\rho'_{N,R+R_1}}{(R+R_1)^d-N}$ are probability measures. 
Before we proceed, we recall \cite[Prop. 2.3]{cotpet}, which we will use below.
\begin{proposition}\label{subadd3}\cite[Prop. 2.3]{cotpet}
Let $\mathsf{c}:\mathbb{R}^d\times\mathbb{R}^d\rightarrow \mathbb R\cup\{+\infty\}$. Consider $k$ probability measures $\mu_1, \ldots, \mu_k,$ with densities respectively equal to $\rho_1, \ldots,\rho_k$, such that the quantities below are well-defined and finite (for $\mathsf{c}(x,y)=|x-y|^{-s}$, let $\rho_i\in L^{1+\frac{s}{d}}(\mathbb{R}^d), i=1,\ldots,k$). Fix $M_1,\ldots, M_k\in\mathbb N_+$, and let $\mu$ be the probability measure with density $\frac{\sum_{i=1}^kM_i \rho_i}{\sum_{i=1}^kM_i}$. Then the following subadditivity relation holds:
\begin{equation}\label{subadd1gen}
E_{\sum_{i=1}^kM_i,\mathsf{c}}^\mathrm{xc}\left(\mu\right):=E_{\sum_{i=1}^kM_i,\mathsf{c}}^\mathrm{xc}\left(\frac{\sum_{i=1}^kM_i\mu_i}{\sum_{i=1}^kM_i}\right)\le \sum_{i=1}^kE_{M_i,\mathsf{c}}^\mathrm{xc}(\mu_i).
\end{equation}
(We apply in the above the convention that ${\mathcal F}_{1,\mathsf{c}}(\mu)=0$ and thus $E_{1,\mathsf{c}}^\mathrm{xc}(\mu)=-\int_{{\mathbb{R}^d}\times {\mathbb{R}^d}}\mathsf{c}(x,y)\diff\mu(x)\diff\mu(y)$.)
\end{proposition}

We apply now Proposition \ref{subadd3} to the decomposition (\ref{inspectrho}) of $\frac{1_{K_{R+R_1}}}{(R+R_1)^d}$, and we get
\begin{eqnarray}\label{inspectrho1}
E_{(R+R_1)^d,s}^{\mathrm{xc}}\left(\frac{1_{K_{R+R_1}}}{(R+R_1)^d}\right)&\le& E_{N,s}^{\mathrm{xc}}\left(\frac{\rho_{N,R_1}}{N}\right)+E_{(R+R_1)^d-N,s}^{\mathrm{xc}}\left(\frac{\rho'_{N,R+R_1}}{(R+R_1)^d-N}\right)\nonumber\\
&\le&E_{N,s}^{\mathrm{xc}}\left(\frac{\rho_{N,R_1}}{N}\right),
\end{eqnarray}
where for the second inequality we used $E_{(R+R_1)^d-N,s}^{\mathrm{xc}}\left(\frac{\rho'_{N,R+R_1}}{(R+R_1)^d-N}\right)\le 0$ (see also Remark 4.7 from \cite{cotpet}). 

Next, by means of Lemma 2.4 (b) from \cite{cotpet} we have
$$E_{(R+R_1)^d,s}^{\mathrm{xc}}\left(\frac{1_{K_{R+R_1}}}{(R+R_1)^d}\right)=(R+R_1)^{-s}E_{(R+R_1)^d,s}^{\mathrm{xc}}\left(\frac{1_{K_{1}}}{|K_{1}|}\right),$$
which implies
\begin{eqnarray*}
\lim_{ R_1\to\infty\atop (R_1/2)^d\in\mathbb{N}}\liminf_{\substack{R=N^{1/d}, N\rightarrow\infty\\R/R_1\in\N}}\frac{1}{N}E_{(R+R_1)^d,s}^{\mathrm{xc}}\left(\frac{1_{K_{R+R_1}}}{(R+R_1)^d}\right)&=&\lim_{ R_1\to\infty\atop (R_1/2)^d\in\mathbb{N}}\liminf_{\substack{R=N^{1/d}, N\rightarrow\infty\\R/R_1\in\N}}\frac{(R+R_1)^{-s}}{N}E_{(R+R_1)^d,s}^{\mathrm{xc}}\left(\frac{1_{K_{1}}}{|K_{1}|}\right)\\
&=&\lim_{ R_1\to\infty\atop (R_1/2)^d\in\mathbb{N}}\liminf_{\substack{R=N^{1/d}, N\rightarrow\infty\\R/R_1\in\N}}\frac{1}{(R+R_1)^{d+s}}E_{(R+R_1)^d,s}^{\mathrm{xc}}\left(\frac{1_{K_{1}}}{|K_{1}|}\right)\\
&=&\textsf{C}_{\mathrm{UEG}}(s,d),
\end{eqnarray*}
where for the last equality we used Theorem 1.1 from \cite{cotpet}. Together with (\ref{inspectrho1}), the above proves Step 1.

\par \textbf{Step 2:} Finally, the bound
\begin{equation}\label{wished-for-le}
\lim_{ R_1\to\infty\atop (R_1/2)^d\in\mathbb{N}}\liminf_{\substack{R=N^{1/d},N\rightarrow\infty\\R/R_1\in\N}}\frac{1}{N}E_{N,s}^{\mathrm{xc}}\left(\frac{\mu_{N,R_1}}{N}\right)\le \lim_{ R_1\to\infty\atop (R_1/2)^d\in\mathbb{N}}\liminf_{\substack{R=N^{1/d}, N\rightarrow\infty\\R/R_1\in\N}}\frac{1}{N}E_{(R-R_1)^d,s}^{\mathrm{xc}}\left(\frac{1_{K_{R-R_1}}}{|K_{R-R_1}|}\right)=\textsf{C}_{\mathrm{UEG}}(s,d)
\end{equation}
follows by the same argument as in Step 1 above, via an application of Proposition \ref{subadd3} to \eqref{defrnr1}. More precisely, we obtain
$$E_{N,s}^{\mathrm{xc}}\left(\frac{\mu_{N,R_1}}{N}\right)\le E_{(R-R_1)^d,s}^{\mathrm{xc}}\left(\frac{1_{K_{R-R_1}}}{(R-R_1)^d}\right)+E_{N-(R-R_1)^d,s}^{\mathrm{xc}}\left(g_N\right).$$
 The result in (\ref{wished-for-le}) follows once more by use of Theorem 1.1 from \cite{cotpet}, via the same arguments as in Step 1, and will be omitted.

\end{proof}
\subsubsection{Conclusion of the proof of the Main Theorem for $d-2<s<d$}
\label{conclstrict}
\textbf{Proof of the Main Theorem for $d-2<s<d$}

Combining the results from the previous subsections, by further taking the $R_1\to\infty$ limit, we find that for ${0\le }d-2<s<d$ and for $\nu_{R_1}$ chosen to be configurations corresponding to minimizers of the problem \eqref{minperscr} like in point (b) of Lemma~\ref{reformulc1c2}, there holds
\begin{eqnarray}
\textsf{C}_{\mathrm{UEG}}(s,d)&\stackrel{\text{Cor. \ref{firstineq}}}{\ge}& \textsf{C}_{\mathrm{Jel}}(s,d)\nonumber\\
&\stackrel{\text{\eqref{valuec1}}}{=}&\lim_{R_1\to\infty\atop (R_1/2)^d\in\mathbb N}\liminf_{\substack{R=N^{1/d}\to\infty\\R/R_1\in\mathbb N}}\frac{ E_{\mathrm{Jel},s}(K_R, \nu_{R_1}|_{K_R})}{N}\label{usevaluec1}\nonumber\\
&\stackrel{\text{Lem. \ref{comparisoneindejel}}}{\ge}&\lim_{R_1\to\infty\atop (R_1/2)^d\in\mathbb N}\liminf_{\substack{R=N^{1/d}\to\infty\\R/R_1\in\mathbb N}}\frac{ E_{\mathrm{UEG},s}(\mu_{N,R_1}, \nu_{R_1}|_{K_R})}{N} \nonumber\\
&&- \lim_{R_1\to\infty}\frac{1}{|K_{R_1}|}\left(2\int_{\mathbb R^d}h^{\nu_{R_1}|_{K_{R_1}}-\mu_{K_{R_1}}}(x)\diff x + \int_{\mathbb R^d} h^{(\delta_0 - \nu_{R_1}^1*\nu_{R_1}^1)*\mu_{K_{R_1}}}(x)\diff x\right)\nonumber\\
&\stackrel{\text{\eqref{mapjecon}}}{=}& \lim_{R_1\to\infty\atop (R_1/2)^d\in\mathbb N}\ \liminf_{\substack{R=N^{1/d}, N\to\infty\\R/R_1\in\mathbb N}}\frac{ E_{\mathrm{UEG},s}(\mu_{N,R_1}, \nu_{R_1}|_{K_R})}{N} \ .\label{passageneededlater}\\
&\stackrel{\text{\eqref{tildeexc}}}{\ge}&\lim_{ R_1\to\infty\atop (R_1/2)^d\in\mathbb{N}}\liminf_{\substack{R=N^{1/d},N\rightarrow\infty\\R/R_1\in\N}}\frac{1}{N}E_{N,s}^{\mathrm{xc}}\left(\frac{\mu_{N,R_1}}{N}\right)\stackrel{\text{\eqref{wished-for-proof}}}{\ge}\textsf{C}_{\mathrm{UEG}}(s,d)\ .\nonumber
\end{eqnarray}
For the second inequality in (\ref{passageneededlater}), we applied the identity
\begin{eqnarray*}
E_{\mathrm{Jel},s}(K_R, \nu_{R_1}|_{K_R})&=&  E_{\mathrm{Jel},s}(K_R, \nu_{R_1}|_{K_R})-E_{\mathrm{UEG},s}(K_R, \nu_{R_1}|_{K_R})+E_{\mathrm{UEG},s}(K_R, \nu_{R_1}|_{K_R})\\
&&- E_{\mathrm{UEG},s}(\mu_{N,R_1}, \nu_{R_1}|_{K_R})+E_{\mathrm{UEG},s}(\mu_{N,R_1}, \nu_{R_1}|_{K_R}),
\end{eqnarray*}
we made use of (\ref{comparisoneq1}) and (\ref{comparisoneq2}),  and we also utilized that, as the second term in the below does not depend on $R$, we can write
\begin{multline*}
\liminf_{\substack{R=N^{1/d}\to\infty\\R/R_1\in\mathbb N}}\left(\frac{ E_{\mathrm{UEG},s}(\mu_{N,R_1}, \nu_{R_1}|_{K_R})}{N} -\frac{1}{|K_{R_1}|}\left(2\int_{\mathbb R^d}h^{\nu_{R_1}|_{K_{R_1}}-\mu_{K_{R_1}}}(x)\diff x + \int_{\mathbb R^d} h^{(\delta_0 - \nu_{R_1}^1*\nu_{R_1}^1)*\mu_{K_{R_1}}}(x)\diff x\right)\right)\\
=\liminf_{\substack{R=N^{1/d}\to\infty\\R/R_1\in\mathbb N}}\frac{ E_{\mathrm{UEG},s}(\mu_{N,R_1}, \nu_{R_1}|_{K_R})}{N} -\frac{1}{|K_{R_1}|}\left(2\int_{\mathbb R^d}h^{\nu_{R_1}|_{K_{R_1}}-\mu_{K_{R_1}}}(x)\diff x + \int_{\mathbb R^d} h^{(\delta_0 - \nu_{R_1}^1*\nu_{R_1}^1)*\mu_{K_{R_1}}}(x)\diff x\right).
\end{multline*}
\medskip

The chain of inequalities in \eqref{passageneededlater} thus allows to prove $\textsf{C}_{\mathrm{Jel}}(s,d) \ge \textsf{C}_{\mathrm{UEG}}(s,d)$ for $0\le d-2<s<d$. Coupled with Corollary \ref{firstineq}, this concludes the proof of the Main Theorem in this case. 

\qed
\begin{rmk}[``Closeness'' of minimizers for the Jellium and Uniform Electron Gas]\label{eqmin}
We note that due to the inequalities above, and since by \eqref{eindbetter} each point in the support of $\gamma_{N,R_1}$ given in \eqref{defgammanr1} corresponds by definition to a minimizing configuration for the Jellium problem defining $\mathsf C_{\mathrm{Jel}}(s,d)$, we have that any minimizer to the $E_{\mathrm{Jel},s}$-problem induces a sequence of almost-competitors, which after reinstating the correct marginal are also almost-minimizers to the $E_{\mathrm{UEG},s}$-problem.

Viceversa, any competitor in the minimization defining $\mathsf C_{\mathrm{UEG}}(s,d)$ as in \eqref{reformulc2a} automatically gives rise, as a consequence of the inequalities \eqref{passageneededlater}, to an almost-minimizer for the minimization defining $\mathsf C_{\mathrm{Jel}}(s,d)$. Therefore as a consequence of the equality $\textsf{C}_{\mathrm{Jel}}(s,d) = \textsf{C}_{\mathrm{UEG}}(s,d)$ in the case $0\le d-2<s<d$ we have that for any sequence of optimizing transport plans $\gamma_N$ we have that for $\gamma_N$-a.e. point $(x_1,\ldots,x_N)$, the periodization of the configuration $\nu:=\sum_{i=1}^N\delta_{x_i}$ has $\mathcal W$-energy $o_{N\to\infty}(1)$-close to that of a minimizer from the problem \eqref{defc1} defining $\mathsf C_{\mathrm{Jel}}(s,d)$ (or, in an equivalent formulation valid in view of the definitions \eqref{weta}, \eqref{w}, there exists a vector field $E\in \mathsf{Comp}_\nu$ with $\mathcal W$-energy $o_{N\to\infty}(1)$-close to the minimum in \eqref{defc1_long}.
\end{rmk}

\section{Continuity of the map $s\rightarrow \mathsf C_{\mathrm{UEG}}(s,d), 0<s<d$}
\label{contsec}
 
 The main result of this section is the proof of Proposition \ref{continuityc2}, which states the continuity of the constant $C_{\mathrm{UEG}}(s,d)$ in $s\in(0,d)$. The proof is based on the Moore-Osgood Theorem of interchanging the double limits between $N$ and $s$ for $\lim_{s\to s_0}\lim_{N\to\infty}{E}^{\mathrm{xc}}_{N,s}(\mu)/N^{1+s/d}$. However, in order to apply the Moore-Osgood Theorem we need continuity in $s\in I$ of ${E}^{\mathrm{xc}}_{N,s}(\mu)/N^{1+s/d}$ at fixed $N$ (proved  in Lemma \ref{convot} below), and uniform convergence in $N\to\infty$ of ${E}^{\mathrm{xc}}_{N,s}(\mu)/N^{1+s/d}$ with respect to the parameter $s\in I$, where $I=[s_0,s_1]\subset (0,d)$ is a closed interval. This last property is shown in Corollary 5.1 from \cite{cotpet}, but only for the ${E}^{\mathrm{xc}}_{\mathrm{GC},N,s}(\mu)/N^{1+s/d}$ 
 grand-canonical exchange corellation energy (for a definition, see (\ref{ExcGC}) below), and not for the ${E}^{\mathrm{xc}}_{N,s}(\mu)/N^{1+s/d}$ energy, for which the needed tools are currently lacking. Continuity in $s\in I$ of ${E}^{\mathrm{xc}}_{\mathrm{GC},N,s}(\mu)/N^{1+s/d}$ at fixed $N$ is proved  in Lemma \ref{convotgc} below. Since it is shown in Theorem 1.1 from \cite{cotpet} that for a large class of $\mu\in \mathcal{P}(\mathbb{R}^d)$
 $$\lim_{N\to\infty}\frac{{E}^{\mathrm{xc}}_{N,s}(\mu)}{N^{1+s/d}}=\lim_{N\to\infty}\frac{{E}^{\mathrm{xc}}_{\mathrm{GC},N,s}(\mu)}{N^{1+s/d}}=\textsf{C}_{\mathrm{UEG}}(s,d) \int_{\mathbb R^d} \rho^{1+\frac{s}{d}}(x)\diff x\ ,$$
 this allows us to work with ${E}^{\mathrm{xc}}_{\mathrm{GC},N,s}(\mu)/N^{1+s/d}$ in the proof below of Proposition \ref{continuityc2} rather than with ${E}^{\mathrm{xc}}_{N,s}(\mu)/N^{1+s/d}$, for which last quantity we are unable to prove the needed properties which would enable us to interchange the $N,s,$ limits therein.

It remains therefore to prove Lemma \ref{convotgc}. In order to highlight the principles at work in the proof, and since the proofs are largely the same, we show them for more general measures $\rho$ although they will be used only for the uniform measure with density $\rho(x)=1_{[0,1]^d}(x)$. 
 
 \subsection{Continuity of the map $s\rightarrow {E}^{\mathrm{xc}}_{N,s}(\mu), 0<s<d$}
Even though we will only need Lemma \ref{convotgc} in our proof of Proposition \ref{continuityc2}, we will also state and prove below Lemma \ref{convot} of continuity in $s\in I$ of ${E}^{\mathrm{xc}}_{N,s}(\mu)/N^{1+s/d}$ at fixed $N$, since it is interesting in its own right. We will make use in the proof of the separation of points for the minimizer, as given in Proposition \ref{modulusint}.
 
We mention at first the following abstract lemma, which will be useful in the proof of Lemma \ref{convot} below (for a proof, see Lemma 6.1 from \cite{ss2d}, for the case where the measures $P_\epsilon$ are assumed to be probability measures, whose proof applies also in the case of our Lemma \ref{convot}).
\begin{lemma}\label{lem:polish_sp}
Assume that $X$ is a Polish metric space, $\{{\mathbb P}_j\}_{j\in\mathbb N}$ form a tight set of Borel positive measures on $X$ and are such that ${\mathbb P}_j\to {\mathbb P}$ weak-$*$ as $j\to \infty$, and assume that $\{f_j\}_{j\in\mathbb N}$ and $f$ are positive and measurable functions on $X$ such that $\liminf_{j\to \infty} f_j(x_j)\ge f(x)$ whenever $x_j\to x$. Then,
\[
\liminf_{j\to\infty}\int f_j \diff {\mathbb P}_j\ge \int f \diff {\mathbb P}.
\]
\end{lemma}
Next we show
\begin{lemma}\label{convot}
Fix $N\in\mathbb{N},N\ge 2$, and let $I\subset(0,d)$ be a closed interval. Let $\mu\in {\mathcal P}({\mathbb{R}}^d)$ with density $\rho$, and such that for all exponents $s\in I$ we have $\rho\in L^{1+\frac{s}{d}}({\mathbb{R}}^d)$. Then we have 
\begin{itemize}
\item [(a)]
The function $s\mapsto{\cal F}_{N,s}(\mu)$ is continuous for all $s\in\mathrm{int}(I)$ (where $\mathrm{int}(I)$ denotes the interior of $I$), i.e. for all $s_0\in\mathrm{int}(I)$ there holds
\begin{equation}\label{convcont}
\lim_{\substack{s\rightarrow s_0}} {\cal F}_{N,s}(\mu)= {\cal F}_{N,s_0}(\mu).
\end{equation}
Moreover, if we let $\gamma_{N,s}\in\mathcal P^N_{sym}(\mathbb{R}^d), \gamma_{N,s}\mapsto \mu$, be a minimizing solution for  ${\cal F}_{N,s}(\mu)$ for $s, s_0\in\mathrm{int}(I)$, then as $s\rightarrow s_0$ we have that (up to a subsequence) $\gamma_{N,s}$ converges weakly to a minimizer of ${\cal F}_{N,s_0}(\mu)$.
\item [(b)] 
The function $s\mapsto E^{\mathrm{xc}}_{N,s}(\mu)$ is continuous for $s\in\mathrm{int}(I)$, i.e. for all $s_0\in\mathrm{int}(I)$ there holds
\begin{equation}\label{convcontexc}
\lim_{\substack{s\rightarrow s_0}} {E}^{\mathrm{xc}}_{N,s}(\mu)= {E}^{\mathrm{xc}}_{N,s_0}(\mu).
\end{equation}
\end{itemize}
\end{lemma}
\begin{proof}[Proof of Lemma~\ref{convot}]
We note first that, due to our assumption that for $s\in\mathrm{int}(I)$ there holds $\rho\in L^{1+\frac{s}{d}}({\mathbb{R}}^d)$, we have
\[
\limsup_{\substack{s\rightarrow s_0}}\int_{\mathbb{R}^d}\int_{\mathbb{R}^d}\frac{1}{|x-y|^s}\rho(x)\rho(y) \diff x \diff y<\infty,
\]
from which we immediately obtain that
\[
0<\limsup_{\substack{s\rightarrow s_0}} {\cal F}_{N,s}(\mu)<\infty~~~\mbox{and}~~~-\infty<\liminf_{\substack{s\rightarrow s_0}}E^{\mathrm{xc}}_{N,s}(\mu)\le \limsup_{\substack{s\rightarrow s_0}}E^{\mathrm{xc}}_{N,s}(\mu)\le 0.
\]

\medskip

\textit{Proof of (a):}

\textbf{Step 1:} We show here the inequality ``$\le$'' in \eqref{convcont}, i.e.
\[
\limsup_{s\rightarrow s_0} {\cal F}_{N,s}(\mu)\le  {\cal F}_{N,s_0}(\mu).
\]
Using as a competitor for ${\cal F}_{N,s}(\mu)$ the optimizer $\gamma_{N,s_0}$ of ${\cal F}_{N,s_0}(\mu)$, we have
\begin{equation}\label{comparbrut}
{\cal F}_{N,s}(\mu)\le\int_{(\mathbb{R}^d)^N}\sum_{\substack{i,j=1\\ i\neq j}}^N\frac{1}{|x_i-x_j|^s} \diff \gamma_{N,s_0}(x_1,\ldots,x_N)\le \int_{(\mathbb{R}^d)^N}\sup_{s\in\mathrm{int}(I)}\sum_{\substack{i,j=1\\ i\neq j}}^N\frac{1}{|x_i-x_j|^s} \diff \gamma_{N,s_0}(x_1,\ldots,x_N)<\infty,
\end{equation}
where the second inequality in the above is finite in view of (\ref{sptawaydiag}) and \eqref{rnos} from Proposition \ref{modulusint} below.

Taking now the limit $s\rightarrow s_0$ in \eqref{comparbrut} and applying the Reverse Fatou Lemma (which holds in view of \eqref{rnos} below and \eqref{comparbrut}),
gives
\begin{eqnarray*}
\limsup_{s\rightarrow s_0}{\cal F}_{N,s}(\mu)&\le& \limsup_{s\rightarrow s_0}\int_{(\mathbb R^d)^N}\,\sum_{i,j=1,i\neq j}^N\frac{1}{|x_i-x_j|^s} \diff \gamma_{N,s_0}(x_1,\ldots,x_N)\\
&=& \int_{(\mathbb R^d)^N}\,\limsup_{s\rightarrow s_0}\sum_{i,j=1,i\neq j}^N\frac{1}{|x_i-x_j|^s} \diff \gamma_{N,s_0}(x_1,\ldots,x_N)\\
&=&\int_{(\mathbb R^d)^N}\,\sum_{i,j=1,i\neq j}^N\frac{1}{|x_i-x_j|^{s_0}} \diff \gamma_{N,s_0}(x_1,\ldots,x_N)={\cal F}_{N,s_0}(\mu).
\end{eqnarray*}

\textbf{Step 2:} We show here the inequality ``$\ge$'' in \eqref{convcont}, i.e.
\[
\liminf_{s\rightarrow s_0} {\cal F}_{N,s}(\mu)\ge {\cal F}_{N,s_0}(\mu).
\]
We first argue tightness in $\calP_{sym}^N(\R^d)$ of the set of all $\gamma\in\calP_{sym}^N(\R^{d})$ such that $\gamma\mapsto\mu$ for some fixed $\mu\in\calP(\R^{d})$. This is proved similarly to \cite[Lem. 4.3]{Vill09}: we note that $\mu$ is tight in $\calP(\R^d)$ because $\R^d$ is a Polish space. Thus, the optimal measures $\gamma_{N,s}\in\mathcal P^N_{sym}(\mathbb{R}^d), \gamma_{N,s}\mapsto \mu$, all lie in a tight set, so we can extract a further subsequence (which, for simplicity, we will denote once more by $(\gamma_{N,s})_s$) which converges weakly as $s\rightarrow s_0$ to some measure ${\tilde{\gamma}}_{N,s_0}\in\mathcal P^N_{sym}(\mathbb{R}^d), {\tilde{\gamma}}_{N,s_0}\mapsto \mu$. Furthermore, we have
\begin{equation*}
{\cal F}_{N,s_0}(\mu)\le \int_{(\mathbb R^d)^N}\,\sum_{i,j=1,i\neq j}^N\frac{1}{|x_i-x_j|^{s_0}} \diff {\tilde{\gamma}}_{N,s_0}(x_1,\ldots,x_N).
\end{equation*}
It remains to show that
\begin{equation}
\label{compets0}
\int_{(\mathbb R^d)^N}\,\sum_{\substack{i,j=1\\i\neq j}}^N\frac{1}{|x_i-x_j|^{s_0}} \diff {\tilde{\gamma}}_{N,s_0}(x_1,\ldots,x_N)\le \liminf_{s\to s_0} \int_{(\mathbb R^d)^N}\,\sum_{\substack{i,j=1\\i\neq j}}^N\frac{1}{|x_i-x_j|^s} \diff {\gamma}_{N,s}(x_1,\ldots,x_N)=\liminf_{s\rightarrow s_0} {\cal F}_{N,s}(\mu).
\end{equation}
The above is an immediate consequence of Lemma \ref{lem:polish_sp}, applied to $\gamma_{N,s_j}$ and $\sum_{i,i'=1, i\neq i'}^N{|x_i-x_{i'}|^{-s_j}}$, along a subsequence $s_j\to s_0$ such that $\liminf_{s\to s_0}$ is achieved by the $(s_j)$, and on the space $X=(\mathbb R^d)^N$.

\par Furthermore, Steps 1 and 2 above conclude the proof of \eqref{convcont}.
\par Lastly, to prove that ${\tilde{\gamma}}_{N,s_0}$ is an optimizer, it will suffice to repeat the argument from Step 1, with ${\tilde{\gamma}}_{N,s_0}$ replacing ${\gamma}_{N,s_0}$ therein.

\medskip

\textit{Proof of (b):} For any $0<\epsilon<\min\{s_0,d-s_0\}$, and any $s\in I=[s_0-\epsilon,s_0+\epsilon]$ we have
$$\frac{1}{|x-y|^s}\le \frac{1}{|x-y|^{s_0-\epsilon}}+\frac{1}{|x-y|^{s_0+\epsilon}},$$
and
$$\int_{\mathbb{R}^{2d}}\frac{1}{|x-y|^s}\diff\mu(x)\diff\mu(y)\le \int_{\mathbb{R}^{2d}}\frac{1}{|x-y|^{s_0-\epsilon}}\diff\mu(x)\diff\mu(y)+\int_{\mathbb{R}^{2d}}\frac{1}{|x-y|^{s_0+\epsilon}}\diff\mu(x)\diff\mu(y)<\infty,$$
where the second inequality follows in view of (\ref{finmu}) and of the hypothesis $\rho\in L^{1+\frac{s}{d}}({\mathbb{R}}^d),s\in I$. By the Dominated Convergence Theorem, we now immediately obtain that
$$
\lim_{s\to s_0}\int_{\mathbb{R}^{2d}}\frac{1}{|x-y|^s}\diff\mu(x)\diff\mu(y)=\int_{\mathbb{R}^{2d}}\frac{1}{|x-y|^{s_0}}\diff\mu(x)\diff\mu(y).
$$
\end{proof}

\medskip

\subsection{Continuity of the map $s\rightarrow {E}^{\mathrm{xc}}_{\mathrm{GC}, N,s}(\mu), 0<s<d$}
Before we start the proofs, we need to introduce the grand-canonical setting. For all $N\in\mathbb{R}_{>0}, N\ge 2$ and for $\mathsf{c}:\mathbb{R}^d\times\mathbb{R}^d\rightarrow\mathbb{R}\cup\{+\infty\}$, let us define the \textit{grand-canonical optimal transport}
\begin{subequations}\label{OT_relax}
\begin{equation}\label{OTGC}
{\mathcal F}_{\mathrm{GC},N,\mathsf{c}}\left(\mu\right):=\inf \left\{\sum_{n=2}^\infty\alpha_n {\mathcal F}_{n,\mathsf{c}}(\mu_n)\left|\begin{array}{c}\sum_{n=0}^\infty\alpha_n=1,\ \sum_{n=1}^\infty n \alpha_n\mu_n = N\mu,\\[3mm]\mu_n\in {\cal P}(\mathbb{R}^d),\ \alpha_n\ge 0,\quad\mbox{for all}\quad n\in\N\end{array}\right.\right\}\ ,
\end{equation}
and the \textit{grand-canonical exchange correlation energy}
\begin{equation}
\label{ExcGC}
E_{\mathrm{GC},N,\mathsf{c}}^\mathrm{xc}\left(\mu\right):={\mathcal F}_{\mathrm{GC},N,\mathsf{c}}\left(\mu\right)-N^2 \int_{{\mathbb{R}^d}\times {\mathbb{R}^d}}\mathsf{c}(x,y)\diff\mu(x)\diff\mu(y).
\end{equation}
\end{subequations}
The classical definition of \eqref{OT_relax} is usually given only for $N\ge 2$, though one can define the quantities  in (\ref{OTGC}) and (\ref{ExcGC}) for all $N>0$ by using the convention that ${\mathcal F}_{N,\mathsf{c}}=0$ for $N\in\{0,1\}$. For more properties of ${\mathcal F}_{\mathrm{GC},N,\mathsf{c}}\left(\mu\right)$, please see Section 4.2.1 from \cite{cotpet}.

We next state and prove for ${E}^{\mathrm{xc}}_{\mathrm{GC}, N,s}(\mu)$ the equivalent result to the one for ${E}^{\mathrm{xc}}_{N,s}(\mu)$ from Lemma \ref{convot}. However, the case of ${E}^{\mathrm{xc}}_{\mathrm{GC}, N,s}(\mu)$ is more tricky to handle as the term ${\mathcal F}_{\mathrm{GC}, N,s}(\mu)$ is an infinite sum, in which we need to take limits of terms over which we do not have a priori a good control. It is not sufficient now to just apply the separation of points from Proposition \ref{modulusint}, as we did in the proof of Lemma \ref{convot}, because the bounds therein are not uniform in $N$, while possibly infinitely many values of $N$ do appear in the sum from ${\mathcal F}_{\mathrm{GC}, N,s}(\mu)$. One way to overcome this issue is to try and reduce the infinite sum to a finite one, with summation up to, say $n_0$, where $n_0$ is the same for all $s\in I$, for some closed interval $I\subset (0,d)$. To such a finite sum we can apply Proposition \ref{modulusint}. To achieve this reduction, we will restrict ourselves to studying the simpler setting where the marginals are compact, which assumption will aid us in the proof.

\begin{lemma}\label{convotgc}
Fix $N\in\mathbb{N},N\ge 2$, and let $I\subset(0,d)$ be a closed interval. Let $\mu\in {\mathcal P}({\mathbb{R}}^d)$, with compactly-supported density $\rho$, and such that for all exponents $s\in I$ we have $\rho\in L^{1+\frac{s}{d}}({\mathbb{R}}^d)$. Then
\begin{itemize}
\item [(a)] 
The function $s\mapsto{\cal F}_{\mathrm{GC}, N,s}(\mu)$ is continuous for
$s\in \mathrm{int}(I)$, i.e. for all $s_0\in \mathrm{int}(I)$ there holds
\begin{equation}\label{convcontGC}
\lim_{\substack{s\rightarrow s_0}} {\cal F}_{\mathrm{GC}, N,s}(\mu)= {\cal F}_{\mathrm{GC}, N,s_0}(\mu).
\end{equation}
Furthermore, the sequence of minimizers $({\vec\lambda}_s,{\vec\mu}_{s}, {\vec\gamma}_{s})_s$ for $({\cal F}_{\mathrm{GC}, N,s}(\mu))_s$ can be shown to converge to a minimizer $({\vec\lambda}_{s_0},{\vec\mu}_{s_0}, {\vec\gamma}_{s_0})$ for  ${\cal F}_{\mathrm{GC}, N,s_0}(\mu)$ (see \cite[Lem. E.1, E.2]{cotpet} for a related statement and proof).
 \item [(b)]  
 The function $s\mapsto E^{\mathrm{xc}}_{\mathrm{GC}, N,s}(\mu)$ is continuous for $s\in \mathrm{int}(I)$, i.e. for all $s_0\in\mathrm{int}(I)$ there holds
\begin{equation}
\label{convcontexcGC}
\lim_{\substack{s\rightarrow s_0}} {E}^{\mathrm{xc}}_{\mathrm{GC}, N,s}(\mu)= {E}^{\mathrm{xc}}_{\mathrm{GC}, N,s_0}(\mu).
\end{equation}

\end{itemize}
\end{lemma}

\begin{proof} 
As $\mathsf{c}>0$ and $\rho$ has support contained in a ball $B_R:=\{x\in\mathbb{R}^d:|x|<R\}$, we have $\mathsf{c}(x,y)>\min_{x',y'\in B(0,R)}\mathsf{c}(x',y')=m_R>0$, valid for all $x,y\in \mathrm{spt}(\mu)$, and then there holds for any competitor $\gamma_{n,s}$ to ${\mathcal F}_{\mathrm{GC},N,s}(\mu)$ 
\begin{equation}\label{fgcbound}
\sum_{n=2}^\infty\lambda_{n, s}\bigg(\int_{{\mathbb{R}}^{Nd}}\,\sum_{i,j=1,i\neq j}^n c(x_i,x_j) \diff \gamma_{n,s}(x_1,\ldots,x_n)\bigg)\ge m_R \sum_{n=2}^\infty n(n-1)\lambda_{n,s}\ge m'_R \sum_{n=2}^\infty n^2\lambda_{n,s},
\end{equation}
for some $m'_R>0$. In view of \eqref{fgcbound}, we also have that
$$ m'_R \sum_{n=2}^\infty n^2\lambda_{n,s}\le {\mathcal F}_{\mathrm{GC},N,s}(\mu)\le \sup_{s\in I} {\mathcal F}_{N,s}(\mu)<N^2 \sup_{s\in I}\int_{\mathbb{R}^{2d}}\frac{1}{|x-y|^s}\mathrm{d}\mu(x)\mathrm{d}\mu(y)<\infty,$$
thus at fixed $N$ for $(\vec\lambda_s,\vec \gamma_s)$ optimizing ${\mathcal F}_{\mathrm{GC},N,s}(\mu)$, we obtain 
\begin{equation}\label{unifboundgammainfty}
\sup_{s\in I}\sum_{n=2}^\infty n^2\lambda_{n,s}<C(\mu)<\infty.
\end{equation}
Furthermore, \eqref{unifboundgammainfty} implies for all $s\in I$ and $m\ge 2$
\begin{equation}
\label{unifboundgc}
\sum_{n=m}^\infty n\lambda_{n,s}<\frac{C(\mu)}{m-1},
\end{equation}
since
$$\sum_{n=m}^\infty n\lambda_{n,s}\le (m-1)\sum_{n=m}^\infty n\lambda_{n,s}\le \sum_{n=2}^\infty n^2\lambda_{n,s}<C(\mu).$$

\par \textit{Proof of (a):} \textbf{Step 1:} We will show first the inequality ``$\le$'', i.e.
\begin{equation}
\label{easydirgc}
\limsup_{s\rightarrow s_0} {\cal F}_{\mathrm{GC},N,s}(\mu)\le  {\cal F}_{\mathrm{GC}, N,s_0}(\mu).
\end{equation}
Fix $s\in \mathrm{int}(I)$. In order to prove (\ref{easydirgc}), we will construct from the minimizer $({\vec\lambda}_{s_0},{\vec\mu}_{s_0}, {\vec\gamma}_{s_0})$ for ${\cal F}_{\mathrm{GC}, N,s_0}(\mu)$ a competitor for a slightly modified version of ${\cal F}_{\mathrm{GC},N,s}(\mu)$, which competitor we will then compare with ${\cal F}_{\mathrm{GC},N,s}(\mu)$. 

Namely, let $\delta>0$. Take $n_0\ge N$, depending on $s_0$, such that $\sum_{n=n_0}^\infty n\lambda_{n,s_0}<\delta$. We re-write
\begin{eqnarray*}
N\mu&=&\sum_{n=1}^\infty n\lambda_{n,s_0}\mu_{n,s_0}=\bigg(\alpha_{1,s_0}\mu_{1,s_0}+\sum_{n= n_0}^\infty n\lambda_{n,s_0}\mu_{n,s_0}\bigg)+\sum_{n=2}^{n_0-1}n\lambda_{n,s_0}\mu_{n,s_0}\\
&=&\bigg(\lambda_{1,s_0}+\sum_{n=n_0}^\infty n\lambda_{n,s_0}\bigg)\mu'_{1,s_0}+\sum_{n=2}^{n_0-1}n\lambda_{n,s_0}\mu_{n,s_0},
\end{eqnarray*}
where $\mu'_{1,s_0}$ is defined as
$$\mu'_{1,s_0}:=\frac{1}{\lambda_{1,s_0}+\sum_{n=n_0}^\infty n\lambda_{n,s_0}}\left(\alpha_{1,s_0}\mu_{1,s_0}+\sum_{n= n_0}^\infty n\lambda_{n,s_0}\mu_{n,s_0}\right).$$
Let $A_{s_0}=A_{s_0}(\delta):=\sum_{n=0}^{n_0-1}\lambda_{n,s_0}+\sum_{n=n_0}^\infty n\lambda_{n,s_0}, 1\le A_{s_0}\le 1+\delta$. Let $({\vec \lambda'}_{s_0}, {\vec\mu'}_{s_0},  {\vec\gamma'}_{s_0})$ defined by
\begin{eqnarray}
\label{competweibast}
 {\vec\lambda'}_{s_0}&=&\bigg(A_{s_0}^{-1} \lambda_{0,s_0},A_{s_0}^{-1}\big(\lambda_{1,s_0}+\sum_{n=n_0}^\infty n\lambda_{n,s_0}\big),A_{s_0}^{-1} \lambda_{2,s_0},\ldots, A_{s_0}^{-1} \lambda_{n_0-1,s_0}\bigg),~ \sum_{n= 0}^{n_0-1}\lambda'_{n,s_0}=1,\nonumber\\[1mm]
{\vec\mu'}_{s_0}&=&(\mu'_{1,s_0},\mu_{2,s_0},\ldots, \mu_{n_0-1,s_0}),~{\vec\gamma}_{s_0}=(\gamma_{2,s_0},\ldots, \gamma_{n_0-1,s_0}), ~\sum_{n=1}^{n_0-1} n\lambda'_{n,s_0}\mu_{n,s_0}=NA^{-1}_{s_0}\mu.
 \end{eqnarray}
Then $({\vec \lambda'}_{s_0}, {\vec\mu'}_{s_0},  {\vec\gamma'}_{s_0})$ is a competitor for ${\cal F}_{\mathrm{GC}, NA^{-1}_{s_0},s}(\mu)$. By applying  \cite[Rem. 4.7 (2), (4)]{cotpet}, we have
$$E_{\mathrm{GC},N, s}^\mathrm{xc}(\mu)\le E_{\mathrm{GC},N A^{-1}_{s_0},s}^\mathrm{xc}(\mu).$$
The above is equivalent to
\begin{eqnarray}
\label{fgcarsm}
{\cal F}_{\mathrm{GC}, N,s}(\mu)&\le&{\cal F}_{\mathrm{GC}, NA^{-1}_{s_0},s}(\mu)+N^2\left(1- A^{-2}_{s_0}\right)\int_{\mathbb{R}^d\times\mathbb{R}^d}\frac{1}{|x-y|^s}\diff\mu(x)\diff\mu(y)\nonumber\\
&\le& {\cal F}_{\mathrm{GC}, NA^{-1}_{s_0},s}(\mu)+C(N) \delta,
\end{eqnarray}
for some $C(N)>0$, and where for the second inequality in the above we utilized that $1\le A_{s_0}\le 1+\delta$. Recalling that by (\ref{competweibast}) the $({\vec \lambda'}_{s_0}, {\vec\mu'}_{s_0},  {\vec\gamma'}_{s_0})$ is a competitor for ${\cal F}_{\mathrm{GC}, NA^{-1}_{s_0},s}(\mu)$, (\ref{fgcarsm}) becomes
\begin{equation}
\label{compss00}
{\cal F}_{\mathrm{GC}, N,s}(\mu)\le \sum_{n=2}^{n_0-1}\alpha_{n,s_0}  \int_{(\mathbb R^d)^n}\,\sum_{i,j=1,i\neq j}^n\frac{1}{|x_i-x_j|^s} \diff \gamma_{n,s_0}(x_1,\ldots,x_n)+C(N)\delta.
\end{equation}
If $\mu$ has a concentration modulus (in the sense of Appendix B below), then $\mu_{n,s_0}, n=1,\ldots,n_0-1$ have a common concentration modulus, so we can use Proposition \ref{modulusint} for each $ \gamma_{n,s_0}, n=1,\ldots n_0$, and apply the Reverse Fatou Lemma in (\ref{compss00}), to conclude
\begin{eqnarray}
\label{compss0}
\limsup_{s\to s_0}{\cal F}_{\mathrm{GC}, N,s}(\mu)&\le& \sum_{n=2}^{n_0-1}\alpha_{n,s_0}  \int_{(\mathbb R^d)^n}\,\sum_{i,j=1,i\neq j}^n\frac{1}{|x_i-x_j|^{s_0}} \diff \gamma_{n,s_0}(x_1,\ldots,x_n)+C(N)\delta\nonumber\\
&\le&{\cal F}_{\mathrm{GC}, N,s_0}(\mu)+C(N)\delta.
\end{eqnarray}
Taking $\delta\to 0$ in (\ref{compss0}) proves the claim in (\ref{easydirgc}).

\par \textbf{Step 2:} We now prove 
\begin{equation}
\label{harddirgc}
\liminf_{s\rightarrow s_0} {\cal F}_{\mathrm{GC},N,s}(\mu)\ge  {\cal F}_{\mathrm{GC}, N,s_0}(\mu).
\end{equation}
The proof follows partly along the same lines as the proof of (\ref{easydirgc}), that is, for each $s\in I$ we will construct from the minimizer $(\vec\lambda_s,\vec\gamma_s, \vec\mu_s)$ for ${\cal F}_{\mathrm{GC}, N,s}(\mu)$ a competitor for a slightly modified version of ${\cal F}_{\mathrm{GC},N,s_0}(\mu)$, which competitor we will then compare with ${\cal F}_{\mathrm{GC},N,s_0}(\mu)$. We will do the construction uniformly in $s\in I$, to allow us to then take the limit $s\to s_0$.

More precisely, let $\delta>0$ and fix $s\in I$. Take a minimizer $(\vec\lambda_s,\vec\gamma_s, \vec\mu_s)$ for $F_{\mathrm{GC},N,\mathsf{c}}^\mathrm{OT}(\mu)$. Recalling (\ref{unifboundgc}), set $n'_0=n'_0(\delta)\ge N$ independent of $s$ such that $\sup_{s\in I}\sum_{n=n'_0}^\infty n\lambda_{n,s}<\delta$. Similarly to the construction in Step 1, we define
$$\mu'_{1,s}:=\frac{1}{\lambda_{1,s}+\sum_{n=n'_0}^\infty n\lambda_{n,s}}\left(\alpha_{1,s}\mu_{1,s}+\sum_{n= n'_0}^\infty n\lambda_{n,s}\mu_{n,s}\right).$$
Let $A_{s}=A_{s}(\delta):=\sum_{n=0}^{n'_0-1}\lambda_{n,s}+\sum_{n=n'_0}^\infty n\lambda_{n,s}, 1\le A_{s}\le 1+\delta$. Let $({\vec \lambda'}_{s}, {\vec\mu'}_{s},  {\vec\gamma'}_{s})$ defined by
\begin{eqnarray}
\label{competweibastst2}
 {\vec\lambda'}_{s}&=&\bigg(A_{s}^{-1} \lambda_{0,s},A_{s}^{-1}\big(\lambda_{1,s}+\sum_{n=n'_0}^\infty n\lambda_{n,s}\big),A_{s}^{-1} \lambda_{2,s},\ldots, A_{s}^{-1} \lambda_{n'_0-1,s}\bigg),~ \sum_{n=0}^{n'_0-1}\lambda'_{n,s}=1,\nonumber\\[1mm]
{\vec\mu'}_{s}&=&(\mu'_{1,s},\mu_{2,s},\ldots, \mu_{n'_0-1,s}),~{\vec\gamma}_{s}=(\gamma_{2,s},\ldots, \gamma_{n'_0-1,s}), ~\sum_{n=1}^{n'_0-1} n\lambda'_{n,s}\mu_{n,s}=NA^{-1}_{s}\mu.
 \end{eqnarray}
Then $({\vec \lambda'}_{s}, {\vec\mu'}_{s},  {\vec\gamma'}_{s})$ is a competitor for ${\cal F}_{\mathrm{GC}, NA^{-1}_{s},s}(\mu)$, and we obtain now instead of (\ref{compss00}) and by the same argument employed there, it holds for some $C'(N)>0$ that
\begin{equation}
\label{compss00'}
{\cal F}_{\mathrm{GC}, N,s_0}(\mu)\le \sum_{n=2}^{n'_0-1}\alpha_{n,s}  \int_{(\mathbb R^d)^n}\,\sum_{i,j=1,i\neq j}^n\frac{1}{|x_i-x_j|^{s_0}} \diff \gamma_{n,s}(x_1,\ldots,x_n)+C'(N)\delta.
\end{equation}

Next, the set $D_{r_{N,\omega_\rho,s}}$ from (\ref{sptawaydiag}) is a compact set for each $s\in I$ and due to the explicit expression \eqref{rnos}, for any interval $I:=[s_0-\epsilon,s_0+\epsilon]\subset(0,d)$ there holds $\cup_{s\in I}D_{r_{N,\omega_\rho,s}}\subset D_{r_{N,\omega_\rho,s_0-\epsilon}}$. The latter is a compact set on which the costs $\sum_{i,j=1,i\neq j}^N |x_i-x_j|^{-s}$ form a monotonically increasing sequence of continuous functions, which converge point-wise to a continuous function as $s\to s_0$. Thus by Dini's theorem the convergence is uniform. Therefore, there exists $\beta_\delta>0$ such that for all $s\in I$ such that $|s-s_0|<\beta_\delta$, we have via (\ref{compss00'})
\begin{eqnarray*}
{\cal F}_{\mathrm{GC},N,s_0}(\mu)&\le&\sum_{n=2}^{n'_0-1}\alpha_{n,s}  \int_{(\mathbb R^d)^n}\,\sum_{i,j=1,i\neq j}^n\frac{1}{|x_i-x_j|^{s}} \diff \gamma_{n,s}(x_1,\ldots,x_n)+C'(N)\delta+\delta n'_0(n'_0-1)\\
&\le&{\cal F}_{\mathrm{GC},N,s}(\mu)+C'(N)\delta+\delta n'_0(n'_0-1),
\end{eqnarray*}
from which
\[
{\cal F}_{\mathrm{GC}, N,s_0}(\mu)\le\liminf_{s\rightarrow s_0} {\cal F}_{\mathrm{GC},N,s}(\mu)+C'(N)\delta+\delta n'_0(n'_0-1).
\]
Taking $\delta\rightarrow 0$ in the above finishes the proof of (\ref{harddirgc}).

\textit{Proof of (b):} The proof follows by the same argument as in (b) from Lemma \ref{convot}, and will be omitted.

\end{proof}

\subsection{Proof of Proposition \ref{continuityc2}.}

\medskip

For simplicity of arguments, and since by the main result of \cite{cotpet}, we have that $\textsf{C}_{\mathrm{UEG}}(s,d)$ is the same for all probability measures with density $\rho\in L^{1+\frac{s}{d}}(\mathbb{R}^d)$, we restrict for the proof to considering only $d\mu_{\mathrm{UEG}}(x):=1_{[0,1]^d}(x)\diff x$.

From equation \eqref{nextorderot} above, Theorem 3.1 in \cite{LewLiebSeir17} and equation (4.73) in \cite{cotpet}, we have
\begin{equation}\label{uniflim1}
\lim_{N\rightarrow\infty}N^{-1-s/d}E^{\mathrm{xc}}_{\mathrm{GC}, N,s}(\mu_{\mathrm{UEG}})=\textsf{C}_{\mathrm{UEG}}(s,d).
\end{equation}
Let $I\subset (0,d)$ be a closed interval.  From Lemma \ref{convotgc}, for all $s_0\in I$ we have
\begin{equation}
\label{limdens}
\lim_{s\rightarrow s_0}N^{-1-s/d}E^{\mathrm{xc}}_{\mathrm{GC}, N,s}(\mu_{\mathrm{UEG}})=N^{-1-s_0/d}E^{\mathrm{xc}}_{\mathrm{GC}, N,s_0}(\mu_{\mathrm{UEG}}).
\end{equation}
From Corollary 5.1 from \cite{cotpet} we have that $E_{\mathrm{GC}, N,s}^{\mathrm{xc}}(\mu)/(N)^{1+s/d}$ converges as $N\to\infty$ uniformly with respect to the parameter $s\in I$. We apply next the Moore-Osgood theorem on interchanging limits. Then
\begin{eqnarray*}
\textsf{C}_{\mathrm{UEG}}(s,d_0)&\stackrel{\text{\eqref{uniflim1}}}=&\lim_{N\rightarrow \infty}{N}^{-1-s_0/d}E^{\mathrm{xc}}_{\mathrm{GC}, N,s_0}(\mu_{\mathrm{UEG}})\stackrel{\text{\eqref{limdens}}}=\lim_{N\rightarrow \infty}\lim_{s\rightarrow s_0}{N}^{-1-s/d}E^{\mathrm{xc}}_{\mathrm{GC}, N,s}(\mu_{\mathrm{UEG}})\\
&\stackrel{\text{Moore-Osgood}}=&\lim_{s\rightarrow s_0}\lim_{N\rightarrow \infty} {N}^{-1-s/d}E^{\mathrm{xc}}_{\mathrm{GC}, N,s}(\mu_{\mathrm{UEG}})\stackrel{\text{\eqref{uniflim1}}}=\lim_{s\rightarrow s_0}\textsf{C}_{\mathrm{UEG}}(s,d).
\end{eqnarray*}
This completes the proof of the Proposition.
\qed

\section{Proof of the Main Theorem for $s=d-2$}
\label{contsecjel}
The main result of this section is the proof of the Main Theorem for the crucial Coulomb case $s=d-2$. In preparation for the proof, we will need some helpful results, which will be introduced and proved in Sections \ref{subadjel} and \ref{contjelcc} and in Appendix \ref{usefulprop}. To this purpose, let for $\mathsf{c}:\mathrm{R}^d\times\mathrm{R}^d\to\mathrm{R}\cup\{0\}$
\begin{equation}
\label{wignermjer}
\Xi_{N,\mathsf{c}}(K_R):= \min\left\{E_{\mathrm{Jel},\mathsf{c}}(K_R,\nu_{\vec x}):\ \vec x=(x_1,\ldots, x_N)\in(\mathbb R^d)^N\right\},
\end{equation}
where we recall that $E_{\mathrm{Jel},\mathsf{c}}(K_R,\nu_{\vec x})$ has been defined in (\ref{defjellium}).

To show the equality of $C_{\mathrm{UEG}}(d-2,d)$ and $C_{\mathrm{Jel}}(d-2,d)$, we proceed as follows. The asymptotic (\ref{nextordergas}) implies that instead of looking at $C_{\mathrm{Jel}}(d-2,d)$, we can consider the $\Xi_{N_0,d-2}(K_R)/N_0$, which for very large $N_0$ becomes arbitrarily close to $C_{\mathrm{Jel}}(d-2,d)$. At each fixed $N$, Lemma \ref{convotjel} provides continuity of $\Xi_{N,s}(K_R)$ in $s\in [d-2,d)$, so rather than considering $\Xi_{N_0,d-2}(K_R)/N_0$, we can consider $\Xi_{N_0,s}(K_R)/N_0$, for all $s$ close enough to $d-2$. The almost subbaditivity argument presented in Lemma \ref{subaddjell} allows to compare $\Xi_{N_0,s}(K_R)/N_0$ with $\Xi_{N,s}(K_R)/N, N\ge N_0$, with bounds uniform in $s$ close enough to $d-2$. We then use the equality of $C_{\mathrm{UEG}}(s,d)$ and $C_{\mathrm{Jel}}(s,d)$ for $d-2<s<d$, proved in Section \ref{framework}, and the continuity of $s\mapsto C_{\mathrm{UEG}}(s,d)$ for $s\in (0,d)$, proved in Section \ref{contsec}. In this way, in order to bring our proof to a close, we do not need to obtain for our $s=d-2$ main proof a rate of convergence result for $\Xi_{N,s}(K_R)/N$, uniformly in $s$, as we had to obtain in Theorem 1.4 from \cite{cotpet} for the optimal transport problem, which result was needed in the proof of continuity in $s$ of $C_{\mathrm{UEG}}(s,d)$.

The methods introduced below in Lemma \ref{subaddjell} and Section \ref{ApB.2} will be treated separately in \cite{CotGnPet} in great generality of measure and of Riesz (and other) costs in a different work involving the optimal transport and Jellium problems. Our almost subadditivity statement from Lemma \ref{subaddjell} helps to provide an upper bound method alternative to the screening method of Sandier and Serfaty \cite{SandSerbook}, \cite{ss1d}, \cite{ss2d}, to obtain optimal upper bounds (and also a rate of convergence as in Theorem 1.4 from \cite{cotpet}) for the next-order term in the Jellium and the Coulomb and Riesz gases problems.

\subsection{Almost subadditivity property for the minimum Jellium energy}
\label{subadjel}
Before we proceed with the proof of the subadditivity statement, we need to introduce first the Fefferman-Gregg decomposition, which, together with Lemma \ref{unifbound} and Lemma \ref{LNequiv} in the Appendix, is the key tool in the proof.

\subsubsection{Fefferman-Gregg decomposition}
We recall to begin with the Swiss cheese lemma, first introduced in \cite{Lieblebowitz} and also used in \cite{Feff85, Hughes89, Gregg89}, for $d=3$, and in \cite{cotpet} for general $d\ge 2$. The main idea therein was to decompose regions in space into sets of disjoint balls with geometrically increasing radii.
\begin{lemma}[``Swiss cheese'' lemma in general dimension]\label{cheeselemma}
For each $d\ge 2$ define $C_d:=\frac{2^{d+1}}{|B_1|}$, where $B_1$ is the unit ball centred at zero $\{x\in \mathbb R^d: |x|< 1\}$. Consider a sequence of real numbers $0<r_1<\cdots<r_M$ such that $r_{k+1}>(1+4\sqrt d |B_1|) r_k$ for all $k$. Assume that $Q$ is a cube of side length $l> 8\sqrt d |B_1|(M+C_d) r_M$. Then there exists a family $\mathcal B$ of disjoint balls of the form $B_r(x)\subset Q,$ which are centred at $x$, for some $x\in\mathbb{R}^d$, such that 
\begin{itemize}
\item The balls in $\mathcal B$ have radii $r\in \{r_1,\ldots,r_M\}$
\item for each $i=1,\ldots, M$, if $Y_i$ is the set of $x\in \mathbb R^d$ that are centers of some ball $B_r(x)\in \mathcal B$ such that $r=r_i$, then there holds
\begin{equation}\label{cheesebound}
\frac{1}{M+C_d+1}<c_i:={\frac{\left|\bigcup_{x\in Y_i}B_{r_i}(x)\right|}{|Q|}}<\frac{1}{M+C_d}.
\end{equation}
\end{itemize}
\end{lemma}
We consider here the cube $Q=[-l/2,l/2]^d$ and a lattice $\mathcal L\subset\mathbb{R}^d$ with fundamental cell $Q$. We then use the ball packing of $Q$ given by the above lemma, and then we define a covering of the whole of $\mathbb R^d$ by extending by $\mathcal L$-periodicity. Following \cite{cotpet}, we then fix a small parameter $\kappa$, say $\kappa\in ]0,1/2]$. For $t\in[1-\kappa,1+\kappa]$, with the interval arbitrarily chosen, consider a positive function $\rho_\kappa\in C_0^\infty([1-\kappa,1+\kappa])$ such that $\int\rho_\kappa(t) dt=1$. We denote now by
\begin{equation}\label{lbrough3}
{\begin{aligned}
\Omega_l &:= \left\{(t,y):\ t\in [1-\kappa,1+\kappa[\ ,\ y\in\left[-\frac{lt}{2}, \frac{lt}{2}\right]^d\right\},\\
d\,\mathbb{P}_l(t,y)&:=\frac{\rho_\kappa(t)}{(lt)^d}1_{\Omega_l}(t,y)\ \diff t\ \diff y,
\end{aligned}}
\end{equation}
and for each $\omega:=(t,y)\in\Omega_l$ we define a $t\mathcal L$-periodic packing $F_\omega$ by
\begin{equation}\label{lbrough4}
F_\omega^l :={tF_{\mathcal B}+y=\{tB+y:\ B\in F_{\mathcal B}\},}~~~\mbox{where}~~~F_{\mathcal B}:=\{B + p: B\in\mathcal B, p\in\mathcal L\}.
\end{equation}
\begin{proposition}[Fefferman-Gregg decomposition from Proposition 1.6 in \cite{cotpet}]\label{prop3ws}
Let $M\in \mathbb N_+$, $0<\epsilon<d/2$ and $\epsilon\le s\le d-\epsilon$. Then there exists a constant $C(\epsilon,d)>0$, such that the cost $|x_1-x_2|^{-s}$ can be decomposed as follows:
\begin{equation}\label{decomptouse}
\frac{1}{|x_1-x_2|^s}=\frac{M}{M+C}\left\{\int_{\Omega_l}\left(\sum_{A\in F^l_\omega}\frac{1_A(x_1)1_A(x_2)}{|x_1-x_2|^s}\right)\diff\mathbb{P}_l(\omega)+w(x_1-x_2)\right\}\ ,
\end{equation}
where $w$ is positive definite. Furthermore, as shown in Lemma \ref{LNequiv} from the Appendix, there exists $c_{\mathrm{LN}}(\epsilon,d)<0$ such that for all $N$, it holds
$$\frac{\Xi_{N,w}(K_R)}{N}\ge\frac{1}{M} c_{\mathrm{LN}}(\epsilon,d).$$
\end{proposition}

\subsubsection{Choice of parameters}\label{sec:choicepar}
Lemma \ref{cheeselemma} applies in particular if $R_M>(1+4\sqrt{d}|B_1|)^M R_1,$ and $l> 8\sqrt d |B_1|(M+C_d) R_M$. Thus, for Lemma \ref{cheeselemma} to apply, the extra constraint linking $M,l,R_1$, which can be formulated in two equivalent ways:
\begin{equation}\label{boundlmr}
l>C(M+C)R_M>C(M+C)C^M\ R_1 \quad\Leftrightarrow\quad \log\frac{l}{R_1}>\log C + \log(M+C) + M\log C\ ,
\end{equation}
 where $C:=\max\{1+4\sqrt{d}|B^d_1|,8\sqrt d |B_1|, C_d\}$ depends only on $d$.
If $R_1,l,M$ are such that
\begin{equation}\label{boundlmr2}
M< \frac{\log(l/R_1)}{3\log C} \quad\Leftrightarrow\quad R_1< C^{-3M}\ l\ ,
\end{equation}
then there exists $M_d>0$ depending only on $d$ such that \eqref{boundlmr} holds for all $M\ge M_d$. Indeed, note that if
\begin{equation}\label{boundM}
M\ge\max\left\{1,\tfrac{\log(M+C)}{\log C}\right\},
\end{equation}
then $3M\log C$ is larger than the right hand side of the second equation in \eqref{boundlmr}, and as a consequence \eqref{boundlmr2} implies \eqref{boundlmr} for such $M$. It suffices then to take $M_d$ to be the smallest value of $M\ge1$ such that \eqref{boundM} holds. It is easy to verify that for any $M\ge M_d$ \eqref{boundM} also holds, that the value $M_d$ depends only on $d$ because $C$ above depends only on $d$. 

A suitable choice of the parameters will be specified in (\ref{boundlmr3}) below.

\subsubsection{Almost subadditivity formula for the minimum Jellium energy}
We will show here one of the key tools in the proof of the Main Theorem for $s=d-2$, and the analogue subadditivity result for the minimum Jellium energy to the one for the optimal transport problem (\ref{otmin}) as stated in Proposition \ref{subadd3} above (see also Proposition 2.3 in \cite{cotpet}). Even though we present the result below only for the case of uniform marginals, we can extend the argument to the case of measures $\mu$ with density $\rho$, which is bounded below and above by positive constants, via a generalization of Lemma \ref{unifbound}.

\begin{lemma}
\label{subaddjell} (Almost subadditivity formula for the minimum Jellium energy)

Let $d\ge 2$.  Let $K_{R_{N_1}}\subset \mathbb{R}^d$ be a cube of volume $R_{N_1}^d=N_1\in\mathbb{N},$ and let $K_{R_{N_1+N_2}}\subset\mathbb{R}^d$ be a cube of volume $R_{N_1+N_2}^d=N_1+N_2\in\mathbb{N}$, with $N_1,N_2\ge 2$. Set $0<\epsilon\le\min(2,d/2)$. Then for all $0<d-2\le s\le d-\epsilon$, we have
$$\Xi_{N_1+N_2,s}(K_{R_{N_1+N_2}})\le  \Xi_{N_1,s}(K_{R_{N_1}})+ \Xi_{N_2,s}(K_{R_{N_1+N_2}}\setminus K_{R_{N_1}})+C_{\mathrm{add}}(\epsilon,d)\frac{N_1+N_2}{\log(\min(N_1,N_2))},$$
where $C_{\mathrm{add}}(\epsilon,d)>0$.
\end{lemma}
The strategy is similar to the proof of subadditivity from the optimal transport problem (\ref{otmin}) as stated in Lemma 2.3 in \cite{cotpet}. We would like to use the minimizers of  $\Xi_{N_1,s}(K_{R_{N_1}})$ and $\Xi_{N_2,s}(K_{R_{N_1+N_2}}\setminus K_{R_{N_1}})$ to construct a competitor for $\Xi_{N_1+N_2,s}(K_{R_{N_1+N_2}})$. This was easier to achieve for the optimal transport problem, in which mixed terms for the competitor constructed from the two minimizers could be cancelled by the corresponding mean field terms, by making use of the fixed marginal assumption of the measures. The cancellation of the mixed terms in the competitor for the Jellium minimizer is more delicate and requires us exploiting additional properties of the process, such as the point separation of the minimizers proved in (\ref{drnos}) following the strategy of \cite{PS}, and the short-range interactions of the dominant term of the Jellium energy when applying the Fefferman-Gregg decomposition.

\medskip

The proof would simplify if we could prove the separation of points from the configurations minimizing $\Xi_{N_1,s}(K_{R_{N_1}})$ from the boundary of $K_{R_{N_1}}$. While the separation of points between themselves can be proved using potential theory methods as in \cite{PS}, using the fact that the potential generated by a charge is blowing up near the charge, the boundary itself generates no strong repelling force, and we were not able to find a simple proof of the separation of the points from the boundary.
\begin{proof}
\par \textbf{Step 1:} To start with, we assume a Fefferman-Gregg decomposition with parameters as given in (\ref{boundlmr}) and (\ref{boundlmr2}), with precise values in our case being given in (\ref{boundlmr3}) below. Due to Lemma \ref{unifbound} configurations $\vec x=(x_1,\ldots,x_N)$ minimizing $\Xi_{N,s}(K_{R_{N}})$ satisfy $x_j\in K_{R_N}, 1\le j\le N$. Due to Lemma \ref{unifbound} there further holds $\min_{i\neq j}|x_i-x_j|\ge r_\mathrm{sep}(\epsilon)$. By abuse of notation, we will denote by
\[
K_{R_N}^\epsilon:=\{x\in K_{R_N}:\ \mathrm{dist}(x,\partial K_{R_N})\ge r_\mathrm{sep}(\epsilon)\}, \quad \partial K_{R_N}^\epsilon:=K_{R_N}\setminus K_{R_N}^\epsilon.
\]
In view of (\ref{decomptouse}), for all $N\in\mathbb{N},N\ge 2,$ we write for $(x_1,\ldots,x_N)\in K_{R_N}$ 
\begin{eqnarray*}
\lefteqn{E_{\mathrm{Jel},s}(K_{R_N},\nu_{\vec x})}\nonumber\\
&=&\sum_{1\le i,j\le N\atop i\neq j}\!\!\!\bigg\{\int_{\Omega_l}\!\!\!\bigg(\!\!\!\sum_{A\in F^l_\omega\atop A\subset K_{R_N}^\epsilon}\!\!\!\frac{1_A(x_i)1_A(x_j)}{|x_i-x_j|^s}\bigg)\diff\mathbb{P}_l(\omega)+\int_{\Omega_l}\!\!\!\bigg(\!\!\!\sum_{A\in F^l_\omega\atop A\cap\partial K_{R_N}^\epsilon\neq\emptyset}\!\!\!\frac{1_A(x_i)1_A(x_j)}{|x_i-x_j|^s}\bigg)\diff\mathbb{P}_l(\omega)+w(x_i-x_j)\bigg\}\\
&&-2\!\!\sum_{1\le i\le N}\int_{K_{R_N}}\!\!\!\bigg\{\int_{\Omega_l}\!\!\!\bigg(\sum_{A\in F^l_\omega\atop A\subset K_{R_N}^\epsilon}\!\!\frac{1_A(x_i)1_A(y)}{|x_i-y|^s}\bigg)\diff\mathbb{P}_l(\omega)\!+\!\int_{\Omega_l}\!\!\!\bigg(\!\!\!\sum_{A\in F^l_\omega\atop A\cap \partial K_{R_N}^\epsilon\neq\emptyset}\!\!\!\!\frac{1_A(x_i)1_A(y)}{|x_i-y|^s}\bigg)\diff\mathbb{P}_l(\omega)\!+\!w(x_i-y)\bigg\}\diff y\\
&&+\int_{K_{R_N}}\!\int_{K_{R_N}}\!\!\!\bigg\{\int_{\Omega_l}\!\!\!\bigg(\!\!\!\sum_{A\in F^l_\omega\atop A\subset K_{R_N}^\epsilon}\frac{1_A(x)1_A(y)}{|x-y|^s}\bigg)\diff\mathbb{P}_l(\omega)+\int_{\Omega_l}\!\!\!\bigg(\!\!\!\!\sum_{A\in F^l_\omega\atop A\cap \partial K_{R_N}^\epsilon\neq\emptyset}\!\!\!\!\frac{1_A(x)1_A(y)}{|x-y|^s}\bigg)\diff\mathbb{P}_l(\omega)+w(x-y)\bigg\}\diff x\diff y\\
&&-\frac{C}{M}E_{\mathrm{Jel},s}(K_{R_N},\nu_{\vec x})\nonumber\\
&:=&\Theta_{N,s}(K_{R_{N}}^\epsilon,\nu_{\vec{x}})+\Theta_{N,s}(\partial K_{R_{N}}^\epsilon,\nu_{\vec{x}})+E_{\mathrm{Jel},w}(K_{R_N},\nu_{\vec x})-\frac{C}{M}E_{\mathrm{Jel},s}(K_{R_N},\nu_{\vec x}),\nonumber\\
\end{eqnarray*}
with the obvious definitions for $\Theta_{N,s}(K_{R_{N}}^\epsilon,\nu_{\vec{x}})$ and $\Theta_{N,s}(\partial K_{R_{N}}^\epsilon,\nu_{\vec{x}})$, where for the last inequality we applied (\ref{loterm1LN}) and $c_{\mathrm{LN}}(\epsilon,d)>0$. Next, due to \eqref{decomptouse} and Lemma \ref{unifbound}, we have
\begin{multline}
\label{applyfefgreg}
\Xi_{N,s}(K_{R_{N}})=\inf\bigg\{\Theta_{N,s}(K_{R_{N}}^\epsilon,\nu_{\vec{x}})+\Theta_{N,s}(\partial K_{R_{N}}^\epsilon,\nu_{\vec{x}})+E_{\mathrm{Jel},w}(K_{R_N},\nu_{\vec x}):(x_1,\ldots,x_N)\in (K_{R_N})^N,\\\min_{i\neq j}|x_i-x_j|\ge r_{\mathrm{sep}}(\epsilon)\bigg\}.
\end{multline}
Let
$$
\Theta^k_{N,s}(K^\epsilon_{R_{N}}):=\inf\bigg\{\Theta_{N,s}(K_{R_{N}}^\epsilon,\nu_{\vec{x}}):(x_1,\ldots,x_N)\in (K_{R_N})^N,\min_{i\neq j}|x_i-x_j|\ge \frac{r_\mathrm{sep}(\epsilon)}{k}\bigg\}.
$$
Note that a minimizer for \eqref{applyfefgreg} exists because the set $(K_{R_N})^N\cap\{\vec x: \min_{i\neq j}|x_i-x_j|\ge r_{\mathrm{sep}}(\epsilon)/k\}$ is compact and the map $\vec x\mapsto \Theta_{N,s}(K_{R_N}^\epsilon,\nu_{\vec x})$ is continuous on this set.

By means of (\ref{loterm1LN}) and (\ref{loterm1LNupbd}) from Lemma \ref{LNequiv}, we obtain from the above for $c_{\mathrm{LN}}(\epsilon,d)<0$
\begin{equation}\label{removwup0}
\left|I(k,s,N,\epsilon)-\Xi_{N,s}(K_{R_{N}})\right|\le -2c_{\mathrm{LN}}(\epsilon,d)\frac{N}{M},
\end{equation}
where
\begin{multline}
I(k,s,N,\epsilon):=\inf\bigg\{\Theta_{N,s}(K_{R_{N}}^\epsilon,\nu_{\vec{x}})+\Theta_{N,s}(\partial K_{R_{N}}^\epsilon,\nu_{\vec{x}}):(x_1,\ldots,x_N)\in (K_{R_N})^N,\\\min_{i\neq j}|x_i-x_j|\ge \frac{r_\mathrm{sep}(\epsilon)}{k}\bigg\} -2c_{\mathrm{LN}}(\epsilon,d)\frac{N}{M}.
\end{multline}

\par\textbf{Step 2:} We show here that for each $k\ge 1$, there exists $c^k_{\mathrm{LN}}(\epsilon,d)<0$, such that there holds 
\begin{equation}
\label{removwuplo}
\left|\Theta_{N,s}^k(K_{R_N}^\epsilon) - \Xi_{N,s}(K_{R_{N}})\right|\le-c^k_{\mathrm{LN}}(\epsilon,d)\left(\frac{N}{M}+ N^{\frac{d-1}{d}}R_M^{d+1}\right).
\end{equation}

\vspace{2mm}

To prove this, due to \eqref{removwup0} it suffice to find upper and lower bounds for the energy contribution $\Theta^k_{N,s}(\partial K_{R_{N}}^\epsilon,\nu_{\vec{x}})$, coming from the intersection of the balls $A\in F_\omega^l$ with $\partial K_{R_N}^\epsilon$. We start with the upper bound, for which we estimate the positive terms in $\Theta^k_{N,s}(\partial K_{R_{N}}^\epsilon,\nu_{\vec{x}})$. Firstly, since $\min_{i\neq j}|x_i-x_j|\ge r_\mathrm{sep}(\epsilon)/k$, the number of points in any ball $A=B(x,r)\in F_\omega^l$ is less than $N^k_{\mathrm{max}}:=c^k_{\mathrm{max}}(d,\epsilon)r^d$, where $c^k_{\mathrm{max}}(d,\epsilon)>0$, provided $r>r_{\mathrm{sep}}(\epsilon)/k$. The latter condition will be ensured by our choice of $R_1\ge 1$ in \eqref{boundlmr3}, for $k>\frac{1}{r_{\mathrm{sep}}(\epsilon)}$. Indeed, in the above case the $\frac{r_{\mathrm{sep}}(\epsilon)}{2k}$-balls centered at the $x_i$ are disjoint, and are contained in the ball $B(x,r+r_{\mathrm{sep}}(\epsilon)/k)$, and thus by comparing the total volumes we obtain $(r_{\mathrm{sep}}(\epsilon)/k)^d\#(\{x_1,\ldots,x_N\}\cap A)\le (r+r_{\mathrm{sep}}(\epsilon)/k)^d\le (2 r)^d$, from which the desired estimate follows.

\medskip

For each two points $x,y\in A=B(x,r)$ which are at distance at least $r_{\mathrm{sep}}/k$ apart there holds $|x-y|^{-s}\le (k/r_{\mathrm{sep}}(\epsilon))^s$ and the number of such pairs is less than $\left(N_{\mathrm{max}}^k\right)^2\le \left(c_{\mathrm{max}}^k(d,\epsilon)\right)^2 r^{2d}$, thus the total interaction between points in $A$ can be bounded above by $\tilde c(d,\epsilon,k) r^{2d}$ for a some $\tilde c(d,\epsilon,k)>0$. Note that all the balls of radius $\le R_M$ that cross $\partial K_{R_N}^\epsilon$ do not contain points at distance larger than $2R_M+r_{\mathrm{sep}}(\epsilon)\le 3R_M$ from $\partial K_{R_N}$, and thus stay within a set of volume $\le C_dN^{\frac{d-1}{d}}R_M$. By a volume comparison as before we find
\begin{eqnarray*}
\sum_{1\le i,j\le N\atop i\neq j}\int_{\Omega_l}\bigg(\sum_{A\in F^l_\omega\atop A\cap\partial K_{R_N}^\epsilon\neq\emptyset}\frac{1_A(x_i)1_A(x_j)}{|x_i-x_j|^s}\bigg)\diff\mathbb{P}_l(\omega)&\le& \int_{\Omega_l}\sum_{A\in F^l_\omega,\atop {A\cap\partial K_{R_{N}}^\epsilon\neq \emptyset}}\sum_{\substack{1\le i\neq j\le N\\ x_i,x_j\in A}}\frac{1}{|x_i-x_j|^s}\diff\mathbb{P}_l(\omega)\\
&\le&\max_{\omega\in \Omega_l}\sum_{\substack{A\in F_\omega^l, A=B(x,r)\\A\subset K_{R_N+3R_M}\setminus K_{R_N}}}\left(c_{\mathrm{max}}^k(d,\epsilon)\right)^2 r^{2d}\\
&\le&\left(c_{\mathrm{max}}^k(d,\epsilon)\right)^2 C_d N^{\frac{d-1}{d}}R_M\max_{\omega\in\Omega_l}\max_{A\in F_\omega^l, A=B(x,r)}r^d\\
&\le& c_{\mathrm{int}}(d,\epsilon,k)N^{\frac{d-1}{d}}R_M^{d+1},
\end{eqnarray*}
for some constant $c_{\mathrm{int}}(d,\epsilon,k)>0$. A similar bound holds for
$$\int_{\Omega_l}\sum_{A\in F^l_\omega,\atop {A\cap\partial K_{R_N}^\epsilon\neq \emptyset}}\int_{K_{R_N}}\int_{K_{R_N}}\frac{1_A(x)1_A(y)}{|x-y|^s}\diff x \diff y\diff\mathbb{P}_l(\omega).$$
Together with Step 1, and by discarding the negative term in the definition of $\Theta_{N,s}^k(\partial K_{R_N}^\epsilon, \nu_{\vec x})$, this proves the upper bound in (\ref{removwuplo}). 

Moving now to the lower bound, it is sufficient to bound the negative term in $\Theta_{N,s}(\partial K_{R_{N}}^\epsilon,\nu_{\vec{x}})$. More precisely, since for each $A\in F^l_\omega$ and for each $1\le i \le N$ such that $x_i\in A$
\begin{eqnarray*}
\int_{K_R}\frac{1_A(x_i)1_A(y)}{\verti{x_i-y}^s} \diff y &=& \int_{K_R\setminus B_{r_{\mathrm{sep}}(\epsilon)}(x_i)}\frac{1_A(x_i)1_A(y)}{\verti{x_i-y}^s} \diff y + \int_{K_R\cap B_{r_{\mathrm{sep}}}(\epsilon)(x_i)}\frac{1_A(x_i)1_A(y)}{\verti{x_i-y}^s} \diff y\\
&\le&\frac{|A|}{\left(r_{\mathrm{sep}}(\epsilon)\right)^s}+\int_{B_{r_{\mathrm{sep}}(\epsilon)}(0)}\frac{1}{|y|^s}\diff y=\frac{|A|+\frac{d}{d-s}|B_{r_{\mathrm{sep}}(\epsilon)}(0)|}{(r_{\mathrm{sep}}(\epsilon))^s},
\end{eqnarray*}
we get for some $c'_{\mathrm{int}}(d,\epsilon,k)>0$
\begin{eqnarray*}
\sum_{1\le i\le  N}\int_{K_{R_N}}\int_{\Omega_l}\bigg(\sum_{A\in F^l_\omega\atop A\cap \partial K_{R_N}^\epsilon\neq\emptyset}\!\!\!\!\frac{1_A(x_i)1_A(y)}{|x_i-y|^s}\bigg)\diff\mathbb{P}_l(\omega)\diff y\le c'_{\mathrm{int}}(d,\epsilon,k)N^{\frac{d-1}{d}}R_M^{d+1}.
\end{eqnarray*}

\par \textbf{Step 3:} We will show here an almost-subadditivity result for $\Theta^k_{N,s}(K_{R_{N}}^\epsilon)$. More precisely, we will show for large enough $k$ 
\begin{equation}
\label{selfsob}
\Theta^k_{N_1+N_2,s}(K_{R_{N_1+N_2}}^\epsilon)\le  \Theta^1_{N_1,s}(K_{R_{N_1}}^\epsilon)+ \Theta^1_{N_2,s}((K_{R_{N_1+N_2}}\setminus K_{R_{N_1}})^\epsilon)+C_{\mathrm{err}}(\epsilon,d)\frac{N_1+N_2}{\log(\min(N_1,N_2))},
\end{equation}
where $C_{\mathrm{err}}(\epsilon,d)>0$.

Let $(x_1,\ldots,x_{N_1})\in (K_{R_{N_1}})^{N_1}$ be a minimizer for $\Theta^1_{N_1,s}(K^\epsilon_{R_{N_1}})$, and let $(x'_1,\ldots,x'_{N_2})\in (K_{R_{N_1+N_2}})^{N_2}$ be a minimizer for $\Theta^1_{N_2,s}((K_{R_{N_1+N_2}}\setminus K_{R_{N_1}})^\epsilon)$. Up to reordering the coordinates, we may assume that 
\[\min_{1\le i\le N_1-N^\epsilon_1}\mathrm{dist}(x_i,\partial K_{R_{N_1}})\ge r_\mathrm{sep}(\epsilon),\quad\mbox{and}\quad\min_{1\le j\le N_2-N^\epsilon_2}\mathrm{dist}\left(x'_j,\partial \left(K_{R_{N_1+N_2}}\setminus K_{R_{N_1}}\right)\right)\ge r_\mathrm{sep}(\epsilon)
\]
and that
$\min_{1\le i\le  N_1-N_1^\epsilon\atop 1\le j\le N_2-N_2^\epsilon}|x_i-x'_j|\ge r_\mathrm{sep}(\epsilon)$. By a volume comparison reasoning like in Step 2, we can also obtain the above with the further bounds $N_1^\epsilon\le c_1(d,\epsilon)|\partial K_{R_{N_1}}^\epsilon|$ and $N_2^\epsilon\le c_1(d,\epsilon)\left|\partial \left(K_{R_{N_1+N_2}}\setminus K_{R_{N_1}}\right)^\epsilon\right|$.

To create a competitor for $\Theta^k_{N_1+N_2,s}(K_{R_{N_1+N_2}}^\epsilon)$ using  $(x_1,\ldots,x_{N_1-N_1^\epsilon})$ and $(x'_1,\ldots, x'_{N_2-N_2^\epsilon})$, we note now that for $k\ge k(d)$ where $k(d)$ is a packing constant depending only on the dimension, we can pack $r_{\mathrm{sep}}(\epsilon)/k$-separated points $(x_{N_1-N_1^\epsilon+1},\ldots, x_{N_1})\in (\partial K_{R_{N_1}}^\epsilon)^{N_1^\epsilon}$ and  $(x'_{N_2-N_2^\epsilon+1},\ldots, x'_{N_2})\in (\partial \left(K_{R_{N_1+N_2}}\setminus K_{R_{N_1}}\right)^\epsilon)^{N_2^\epsilon}$ such that 
\[
\min_{\substack{1\le i,j\le N_1\\ i\neq j}}|x_i-x_j|\ge \frac{r_{\mathrm{sep}}(\epsilon)}{k},\quad \min_{\substack{1\le i,j\le N_2\\ i\neq j}}|x'_i-x'_j|\ge \frac{r_{\mathrm{sep}}(\epsilon)}{k},\quad \min_{\substack{1\le i\le N_1\\ 1\le j\le N_2}}|x_i-x_j'|\ge \frac{r_{\mathrm{sep}}(\epsilon)}{k},
\]
and 
\[\min_{N_1-N_1^\epsilon\le i\le N_1}\mathrm{dist}\left(x_i,\partial K_{R_{N_1}}\right)\ge \frac{r_\mathrm{sep}(\epsilon)}{k},\quad \min_{N_2-N_2^\epsilon\le i\le N_2}\mathrm{dist}\left(x'_i,\partial\left(K_{R_{N_1+N_2}}\setminus K_{R_{N_1}}\right)\right)\ge \frac{r_\mathrm{sep}(\epsilon)}{k}.
\]

Then $(x_1,\ldots,x_{N_1},x'_1,\ldots, x'_{N_2})$ is a competitor for $\Theta_{N_1+N_2,s}^k(K_{R_{N_1+N_2}}^\epsilon)$, and we obtain 
\begin{eqnarray*}
\lefteqn{\Theta^k_{N_1+N_2,s}(K^\epsilon_{R_{N_1+N_2}})}\\
&\le&\Theta_{N,s}(K^\epsilon_{R_{N_1+N_2}},\nu_{({\vec{x}}_{N_1}, {\vec{x}}_{N_2})})=\Theta^1_{N_1,s}(K^\epsilon_{R_{N_1}})+\Theta^1_{N_1,s}((K_{R_{N_1+N_2}}\setminus K_{R_{N_1}})^\epsilon)\\
&&+2\int_{\Omega_l}\sum_{\substack{A\in F^l_\omega\\ A\cap\partial K^\epsilon_{R_{N_1}}\neq \emptyset~~\mathrm{or}\\ A\cap \partial\left(K_{R_{N_1+N_2}}\setminus K_{R_{N_1}}\right)^\epsilon\neq \emptyset}}\bigg(\sum_{1\le i\le N_1\atop 1\le j\le N_2}\frac{1_A(x_i)1_A(x'_j)}{|x_i-x'_j|^s}-2\sum_{1\le i\le N_1}\int_{K_{R_{N_1+N_2}}\setminus K_{R_{N_1}}}\frac{1_A(x_i)1_A(y)}{|x_i-y|^s}\diff y\\
&&-2\sum_{1\le i\le N_2}\int_{K_{R_{N_1}}}\frac{1_A(x'_j)1_A(y)}{|x'_j-y|^s}\diff y+\int_{K_{R_{N_1}}}\int_{K_{R_{N_1+N_2}}\setminus K_{R_{N_1}}}\frac{1_A(x)1_A(y)}{|x-y|^s}\diff x \diff y\bigg)\diff\mathbb{P}_l(\omega),
\end{eqnarray*}
where the contribution from the mixed terms can be estimated similarly to the bounds for the boundary intersected sets in Step 2, being bounded by $c(d,\epsilon,k)(N_1+N_2)^{\frac{d-1}{d}}R_M^{d+1}$.

We make now a suitable choice of parameters satisfying (\ref{boundlmr}) and (\ref{boundlmr2}), and such that $( N_1 +N_2)^{\frac{d-1}{d}}R_M^{d+1}$ is of order not bigger than $(N_1+N_2)/M$. Such a good choice turns out to be as follows. We take for $N_1+N_2\ge C^{18d(M_d+1)}$ 
\begin{equation}\label{boundlmr3}
R_1:=(\min(N_1,N_2))^{\frac1{3d(d+1)}}\ge 1,\quad l:=(\min(N_1,N_2))^{\frac1{2d(d+1)}},\quad M:=\left[\frac{\log \min(N_1,N_2)}{18\ d(d+1)\log C}\right]-1\ .
\end{equation}
By choice of $R_M$, (\ref{selfsob}) follows.

\par \textbf{Step 4:} Conclusion of the proof

Fix $k$ large enough such that (\ref{selfsob}) holds. Then the conclusion of the lemma follows  now immediately by applying (\ref{removwuplo}) to $\Theta^k_{N_1+N_2,s}(K_{R_{N_1+N_2}}^\epsilon)$, and (\ref{removwuplo}) to $\Theta^1_{N_1,s}(K_{R_{N_1}}^\epsilon)$ and $\Theta^1_{N_2,s}((K_{R_{N_1+N_2}}\setminus K_{R_{N_1}})^\epsilon)$. Then
\begin{eqnarray*}
\Xi_{N_1+N_2,s}(K_{R_{N_1+N_2}})&\stackrel{\text{\eqref{removwuplo}}}\le&\Theta^k_{N_1+N_2,s}(K_{R_{N_1+N_2}}^\epsilon)+c'_{\mathrm{err}}(\epsilon,d)\frac{N_1+N_2}{M}\\
&\stackrel{\text{\eqref{selfsob}}}\le&\Theta^1_{N_1,s}(K_{R_{N_1}}^\epsilon)+ \Theta^1_{N_2,s}((K_{R_{N_1+N_2}}\setminus K_{R_{N_1}})^\epsilon)+C''_{\mathrm{err}}(\epsilon,d)\frac{N_1+N_2}{M}\\
&\stackrel{\text{\eqref{removwuplo}}}\le&\Xi_{N_1,s}(K_{R_{N_2}})+\Xi_{N_2,s}(K_{R_{N_1+N_2}}\setminus K_{R_{N_1}})+C'''_{\mathrm{err}}(\epsilon,d)\frac{N_1+N_2}{M},
\end{eqnarray*}
where $C'_{\mathrm{err}}(\epsilon,d),C''_{\mathrm{err}}(\epsilon,d),C'''_{\mathrm{err}}(\epsilon,d)>0$. The statement of the lemma follows. 
\end{proof}

\subsection{Continuity of the map $s\rightarrow {\Xi}_{N,s}(K_R),d-2\le s<d$}
\label{contjelcc}

We prove here the following lemma, which is of independent interest and which can be in fact stated and proved for any measure $\mu$ with density $\rho$, bounded below and above by positive constants. Similarly to the proof  of Lemma \ref{convot}, we will make use of the separation of points for the minimizer as stated in (\ref{drnos}) below.
\begin{lemma}\label{convotjel}
Let $d\ge 2$. Let $K_R\subset \mathbb{R}^d$ be a cube of volume $R^d=N$, with $N\ge 2$. Set $0<\epsilon\le\min(2,d/2)$. Then we have that the function $s\mapsto \Xi_{N,s}(K_R)$ is continuous for $0<d-2\le s\le d-\epsilon$, i.e. for all $0<d-2\le s_0\le s-\epsilon$ there holds 
\begin{equation}\label{convcontjel}
\lim_{\substack{s\rightarrow s_0}} \Xi_{N,s}(K_R)= \Xi_{N,s_0}(K_R).
\end{equation}

Any sequence ${\vec x}^{N,s}=(x^{N,s}_1,\ldots, x^{N,s}_N)$ of optimizers of ${\Xi}_{N,s}(K_R)$ converges to an optimizer of ${\Xi}_{N,s_0}(K_R)$.
\end{lemma}
\begin{proof}
We note first from Lemma \ref{LNequiv} that 
\[
-\infty<\liminf_{\substack{s\rightarrow s_0}}\Xi_{N,s}(\mu)\le \limsup_{\substack{s\rightarrow s_0}}\Xi_{N,s}(\mu)\le 0.
\]
In what follows, we denote for each $d-2\le s\le d-\epsilon$ by ${\vec x}^{N,s}=(x^{N,s}_1,\ldots, x^{N,s}_N)$ the optimizer of ${\Xi}_{N,s}(K_R)$.

Next, since by Lemma \ref{unifbound} it holds that $(\vec x^{N,s})_s\in (K_R)^N$, we have that every subsequence of $(\vec x^{N,s})_s, s\rightarrow s_0,$ contains a subsubsequence converging to a limit $\vec {\hat x}^{N,s_0}$. In Step 1, we will prove that $\lim_{s\rightarrow s_0}\Xi_{N,s}(K_R)=E_{\mathrm{Jel},s}(K_R,\nu_{\vec {\hat x}^{N,s_0}})$. In Step 2 it will be shown that for any such subsubsequence, $\vec {\hat x}^{N,s_0}$ is a minimizer of $\Xi_{N,s_0}(K_R)$. This will immediately imply (\ref{convcontjel}). 

\textbf{Step 1:} We show here that
$$\lim_{s\rightarrow s_0}\Xi_{N,s}(K_R)=E_{\mathrm{Jel},s}(K_R,\nu_{\vec {\hat x}^{N,s_0}}).$$

In order to simplify notation, we will still use in our arguments below the same notation $(\vec x^{N,s})_s$ for any converging subsubsequence. To start with, we have
\begin{equation}\label{compbrut1}
\Xi_{N,s}(K_R)=E_{\mathrm{Jel},s}(K_R,\nu_{\vec x^{N,s}})=\sum_{i,j=1,i\neq j}^N\frac{1}{|x_i^{N,s}-x_j^{N,s}|^s} - 2\sum_{i=1}^N\int_{K_R}\frac{1}{|x_i^{N,s}-y|^s} \diff y + \int_{K_R}\int_{K_R} \frac{1}{|x-y|^s}\diff x\ \diff y.
\end{equation}
By means of Lemma \ref{unifbound}, we immediately obtain
\begin{equation}
\label{unifbound11}
\lim_{s\rightarrow s_0}\sum_{i,j=1,i\neq j}^N\frac{1}{|x_i^{N,s}-x_j^{N,s}|^s} =\sum_{i,j=1,i\neq j}^N\frac{1}{|{\hat x}_i^{N,s_0}-{\hat x}_j^{N,s_0}|^{s_0}}.
\end{equation}
For the limit of the second term in (\ref{compbrut1}), we next claim that
\begin{equation}\label{unifbound123a}
\lim_{s\to s_0}\sum_{i=1}^N\int_{K_R}\frac{1}{|x_i^{N,s}-y|^s }\diff y=\sum_{i=1}^N\int_{K_R}\frac{1}{|{\hat x}_i^{N,s_0}-y|^{s_0}} \diff y.
\end{equation}
To prove this, it suffices to prove the convergence separately for the $N$ summands, namely  that for all $i=1,\ldots, N$, there holds
\begin{equation}\label{unifbound123}
\lim_{s\to s_0}\int_{K_R}\frac{1}{|x_i^{N,s}-y|^s }\diff y=\int_{K_R}\frac{1}{|{\hat x}_i^{N,s_0}-y|^{s_0}} \diff y.
\end{equation}
To prove \eqref{unifbound123}, we may write, for any $\delta\in]0,1[$,
\begin{equation*}
I(s):=\int_{K_R}\frac{1}{\verti{x_i^{N,s}-y}^s} \diff y = \int_{K_R\setminus B_\delta(\hat x_i^{N,s_0})}\frac{1}{\verti{x_i^{N,s}-y}^s} \diff y + \int_{K_R\cap B_\delta(\hat x_i^{N,s_0})}\frac{1}{\verti{x_i^{N,s}-y}^s} \diff y :=I_\delta(s) + II_\delta(s).
\end{equation*}
Since $x_i^{N,s}\to \hat x_i^{N,s_0}$ as $s\to s_0$, for any fixed $\delta\in]0,1[$ there exists $\delta^\prime>0$ such that if $\verti{s-s_0}<\delta^\prime$ then there holds $x_i^{N,s}\in B_{\delta/2}(\hat x_i^{N,s_0})$ for all $i=1,\ldots,N$. Thus, if $|y-\hat x_i^{N,s_0}|\ge \delta$, there holds for $\verti{s-s_0}<\delta^\prime$ 
\begin{equation}\label{bound1}
\verti{x_i^{N,s}-y}^{-s} \le\left(\frac{2}{\delta}\right)^{\epsilon-d} \quad\mbox{ for all }\quad y\in K_R\setminus B_{\delta}(x_i^{N,s_0}),\ s\in[d-2,d-\epsilon],
\end{equation}
and thus by dominated convergence
\begin{equation}\label{bound11}
\lim_{s\to s_0} I_\delta(s) = I_\delta(s_0).
\end{equation}
Next, note that
\begin{eqnarray}\label{bound2}
\max_{s\in[d-2,d-\epsilon]}II_\delta(s)&\le&\max_{s\in[d-2,d-\epsilon]}\int_{B_\delta(\hat x_i^{N,s_0})}\frac{1}{\verti{x_i^{N,s}-y}^s} \diff y \le \max_{s\in[d-2,d-\epsilon]}\int_{\big|y+x_i^{N,s}-\hat x_i^{N,s_0}\big|\le\delta}\frac{1}{\verti{y}^s} \diff y\nonumber\\
&\le& \max_{s\in[d-2,d-\epsilon]}\int_{B_{3\delta/2}(0)}\frac{1}{\verti{y}^s} \diff y
=\max_{s\in[d-2,d-\epsilon]}\verti{\mathbb S^{d-1}} \int_0^{3\delta/2}(r^\prime)^{d-1-s}\diff r^\prime \nonumber\\
&=& \verti{\mathbb S^{d-1}}\max_{s\in[d-2,d-\epsilon]}\frac{(3\delta/2)^{d-s}}{d-s}\le \verti{\mathbb S^{d-1}}\frac{(3\delta/2)^{d-2}}{\epsilon},
\end{eqnarray}
where for the second inequality we applied $|x_i^{N,s}-\hat x_i^{N,s_0}|\le\delta/2$.

By \eqref{bound11} and \eqref{bound2} we find
\begin{eqnarray*}
I_\delta(s_0)&=&\lim_{s\to s_0}I_\delta(s)\le \lim_{s\to s_0}I(s)=\lim_{s\to s_0}\left(I_\delta(s) + II_\delta(s)\right)\\
&\le& \lim_{s\to s_0}I_\delta(s) + \verti{\mathbb S^{d-1}}\frac{(3\delta/2)^{d-2}}{\epsilon}=I_\delta(s_0) + \verti{\mathbb S^{d-1}}\frac{(3\delta/2)^{d-2}}{\epsilon}.
\end{eqnarray*}
Then by taking $\delta\to 0$ we find that $\lim_{s\to s_0}I(s)=\lim_{\delta\to 0}I_\delta(s_0)$. We then note that $\lim_{\delta\to 0}I_\delta(s_0) = I(s_0)$, by the monotone convergence theorem, and thus we conclude the proof of \eqref{unifbound123}.

\medskip

Lastly, we note that
\begin{eqnarray}\label{unif1c}
\lim_{s\rightarrow s_0} \int\int_{K_R\times K_R}\frac{1}{|x-y|^s}\diff x\ \diff y
&=&\int\int_{K_R\times K_R}\frac1{\verti{x-y}^{s_0}}\diff x\ \diff y,
\end{eqnarray}
by the same arguments as in the proof of Lemma \ref{convot} (b).

\medskip

Gathering together (\ref{unifbound11}), (\ref{unifbound123}) and (\ref{unif1c}), we obtain in (\ref{compbrut1})
\begin{equation}
\label{unifconc}
\lim_{s\rightarrow s_0}\Xi_{N,s}(K_R)=\sum_{i,j=1,i\neq j}^N\frac{1}{|{\hat x}_i^{N,s_0}-{\hat x}_j^{N,s_0}|^{s_0}} - 2\sum_{i=1}^N\int_{K_R}\frac{1}{|{\hat x}_i^{N,s_0}-y|^{s_0}}\diff y + \int_{K_R}\int_{K_R} \frac{1}{|x-y|^{s_0}}\diff x\ \diff y.
\end{equation}

\textbf{Step 2:} We show here that ${\hat x}^{N,s_0}$  is a minimizer for $\Xi_{N,s_0}(K_R)$, which in turn will prove that
\[
\limsup_{s\rightarrow s_0} \Xi_{N,s}(K_R)=  \Xi_{N,s_0}(K_R).
\]
Using as a competitor for $\Xi_{N,s}(K_R)$ the optimizer ${\vec x}^{N,s_0}$ of $\Xi_{N,s_0}(K_R)$, we find
\begin{equation}\label{compbrut}
\Xi_{N,s}(K_R)\le \sum_{i,j=1,i\neq j}^N\frac{1}{|x_i^{N,s_0}-x_j^{N,s_0}|^s} - 2\sum_{i=1}^N\int_{K_R}\frac{1}{|x_i^{N,s_0}-y|^s} \diff y + \int_{K_R}\int_{K_R} \frac{1}{|x-y|^s}\diff x\ \diff y.
\end{equation}
Taking now limits in (\ref{compbrut}) and applying the same reasonings as in Step 1 to ensure the convergence, we find
\begin{eqnarray*}
\lim_{s\rightarrow s_0}\Xi_{N,s}(K_R)&\le& \lim_{s\rightarrow s_0} \bigg(\sum_{i,j=1,i\neq j}^N\frac{1}{|x_i^{N,s_0}-x_j^{N,s_0}|^s} - 2\sum_{i=1}^N\int_{K_R}\frac{1}{|x_i^{N,s_0}-y|^s} \diff y  + \int_{K_R}\int_{K_R} \frac{1}{|x-y|^s}\diff x\  \diff y\bigg)\\
&=&\sum_{i,j=1,i\neq j}^N\frac{1}{|x_i^{N,s_0}-x_j^{N,s_0}|^{s_0}} - 2\sum_{i=1}^N\int_{K_R}\frac{1}{|x_i^{N,s_0}-y|^{s_0}} \diff y  + \int_{K_R}\int_{K_R} \frac{1}{|x-y|^{s_0}}\diff x\  \diff y \\
&=& \Xi_{N,s_0}(K_R).
\end{eqnarray*}
Coupling the above with (\ref{unifconc}) and Step 1, we get that 
\[
\lim_{s\rightarrow s_0}\Xi_{N,s}(K_R)=E_{\mathrm{Jel},s}(K_R,\nu_{\vec {\hat x}^{N,s_0}}) \le \Xi_{N,s_0}(K_R),
\]
which implies that ${\hat x}^{N,s_0}$  is a minimizer for $\Xi_{N,s_0}(\mu)$ and also gives equality in the above, finishing the argument for Step 2.
\end{proof}

\subsection{Conclusion of the proof of the Main Theorem for $s=d-2$}
\textbf{Proof of the Main Theorem for $s=d-2$}

In view  of Corollary \ref{firstineq}, it is sufficient to show that $C_{\mathrm{UEG}}(d-2,d)\le C_{\mathrm{Jel}}(d-2,d)$.

We are going to use a modification of Lemma \ref{subaddjell}, for a subsequence of $N$'s which are so far apart that after renormalization by dividing by $N$ we have a rapidly convergent sequence. As our error terms are uniform in $s$ for $d-2\le s\le d-2+\epsilon$ for $0<\epsilon<2$, but of order $1/\log N$, a subsequence which works well is $\{2^{2^N}\}_{N\in\mathbb N}$. In a first step of the proof we establish the comparison between successive terms in this subsequence, and then we use this in a second step in order to conclude the uniform upper bound of the limits defining $C(s,d)$ for $s\in[d-2,d-2+\epsilon]$.

\par \textbf{Step 1:} We claim that for some $C_{\mathrm{add}}(\epsilon,d)>0$
\begin{equation}
\label{reducsubsec0}
\Xi_{2^{2^N},s}(K_{R_{2^{2^N}}})\le  2^{2^{N-1}}\Xi_{2^{2^{N-1}},s}(K_{R_{2^{2^{N-1}}}})+C_{\mathrm{add}}(\epsilon,d)\frac{2^{2^N}}{2^N}.
\end{equation}
To prove the above, we use the translation-invariance of the energy, namely the fact that if $K_{R_{N}}+a, a\in\mathbb{R}^d,$ is the cube obtained from translating $K_{R_N}$ by $a$, then for all $N\ge 2$ we have
\begin{equation}\label{translateKR}
\Xi_{N,s}(K_{R_{N}})=\Xi_{N,s}(K_{R_{N}}+a).
\end{equation}
We now repeat the proof of Lemma \ref{subaddjell}, with (\ref{selfsob}) applied now to $2^{2^{N-1}}$ copies $K_{R_{2^{2^{N-1}}}}+a_i, i=1,\ldots, 2^{2^{N-1}},$ of $K_{R_{2^{2^{N-1}}}}$ covering the whole $K_{R_{2^{2^N}}}$. More precisely, we construct a competitor of $\Xi_{2^{2^{N}},s}(K_{R_{2^{2^{N}}}})$ from the minimizers for $\Xi_{2^{2^{N-1}},s}(K_{R_{2^{2^{N-1}}}}+a_i)$. Thus for each common face of two cubes in the covering, we have a formula like \eqref{selfsob} with $N_1=N_2=2^{2^{N-1}}$. Therefore, in the analogue of Step 3 from Lemma \ref{subaddjell} we have an almost subadditvity statement with error of order $2^{2^{N-1}}/2^N$ for each of the boundary intersections for the $2^{2^{N-1}}$ cubes, which gives in total 
\[
\Theta^k_{2^{2^N},s}(K_{R_{2^{2^N}}}^\epsilon)\le  \sum_{i=1}^{2^{2^{N-1}}}\Theta^1_{2^{2^{N-1}},s}((K_{R_{2^{2^{N-1}}}}+a_i)^\epsilon)+C_{\mathrm{err}}(\epsilon,d)\frac{2^{2^{N-1}}\times 2^{2^{N-1}}}{2^N}.
\]
Using the above and (\ref{translateKR}), the proof to obtain (\ref{reducsubsec0}) follows now similarly to the proof of Lemma \ref{subaddjell}, and will be omitted.

\par \textbf{Step 2:} Here we prove the theorem for $s=d-2$.

From (\ref{reducsubsec0}) we obtain
\begin{equation}
\label{reducsubsec}
\frac{\Xi_{2^{2^N},s}(K_{R_{2^{2^N}}})}{2^{2^N}}\le  \frac{\Xi_{2^{2^{N-1}},s}(K_{R_{2^{2^{N-1}}}})}{2^{2^{N-1}}}+\frac{C_{\mathrm{add}}(\epsilon,d)}{2^N}.
\end{equation}
Applying (\ref{reducsubsec}) repeatedly, we have for all $N\ge N_0$
\begin{eqnarray}
\label{iddm1}
\frac{\Xi_{2^{2^N},s}(K_{R_{2^{2^N}}})}{2^{2^N}}&\le& \frac{\Xi_{2^{2^{N_0}},s}(K_{R_{2^{2^{N_0}}}})}{2^{2^{N_0}}}+C_{\mathrm{add}}(\epsilon,d)\left(\frac{1}{2^{N_0}}+\ldots+\frac{1}{2^{N}}\right)\nonumber\\
&=&\frac{\Xi_{2^{2^{N_0}},s}(K_{R_{2^{2^{N_0}}}})}{2^{2^{N_0}}}+\frac{C'_{\mathrm{add}}(\epsilon,d)}{2^{N_0}},
\end{eqnarray}
where $C'_{\mathrm{add}}(\epsilon,d)>0$. Recall now that for all $d-2\le s<d$, we have from (\ref{reformulc1}) that
$$\lim_{N\to\infty}\frac{\Xi_{2^{2^N},s}(K_{R_{2^{2^N}}})}{2^{2^N}}= C_{\mathrm{Jel}}(s,d).$$
Fix $\delta>0$, and take $N_0=N_0(\delta,d)$ such that
\begin{equation}
\label{iddm2}
\frac{1}{2^{N_0}}<\delta~~~\mbox{and}~~~\frac{\Xi_{2^{2^N},d-2}(K_{R_{2^{2^N}}})}{2^{2^N}}\le C_{\mathrm{Jel}}(d-2,d)+\delta,~~\forall N\ge N_0.
\end{equation}
From Lemma \ref{convotjel}, there exists $\beta=\beta_{N_0}\in [0,2)$ such that for all $d-2\le s<\beta+d-2$, it holds that
\begin{equation}
\label{iddm3}
\frac{\Xi_{2^{2^{N_0}},s}(K_{R_{2^{2^{N_0}}}})}{2^{2^{N_0}}}<\frac{\Xi_{2^{2^{N_0}},d-2}(K_{R_{2^{2^{N_0}}}})}{2^{2^{N_0}}}+\delta.
\end{equation}
From (\ref{iddm1}), (\ref{iddm2}) and (\ref{iddm3}), we get for all $N\ge N_0$ and $d-2\le s<\beta+d-2$
$$\frac{\Xi_{2^{2^{N}},s}(K_{R_{2^{2^{N}}}})}{2^{2^{N}}}\le C_{\mathrm{Jel}}(d-2,d)+\delta \left(C'_{\mathrm{add}}(\epsilon,d)+2\right).$$
Taking $N\to\infty$ in the above, and applying $ C_{\mathrm{Jel}}(s,d)= C_{\mathrm{UEG}}(s,d)$ for $d-2<s<d$, we obtain
$$C_{\mathrm{UEG}}(s,d)\le C_{\mathrm{Jel}}(d-2,d)+\delta \left(C'_{\mathrm{add}}(\epsilon,d)+2\right),~~~\mbox{for}~~d-2<s<d.$$
Making now use of Proposition \ref{continuityc2} gives 
$$\lim_{s\searrow d-2}C_{\mathrm{UEG}}(s,d)=C_{\mathrm{UEG}}(d-2,d)\le C_{\mathrm{Jel}}(d-2,d)+\delta \left(C'_{\mathrm{add}}(\epsilon,d)+2\right).$$ 
Taking $\delta\to 0$ in the above allows to conclude.
\qed

\appendix

\section{Concentration estimates on transport plans}

We define here for a nonnegative function $f\in L^1(\mathbb R^d)$ the \emph{concentration modulus} as the function $\omega_f:\mathbb R_+\to\mathbb R_+$ such that $\omega_f(t)\to 0$ as $t\to 0$, defined as follows:
\[
\omega_f(t):=\inf\left\{r>0:\ \exists x\in\mathbb R^d,\ \int_{B_r(x)} f(y) \diff y >t\right\}\ .
\]
If $f_k\rightharpoonup f$ weakly, if there exists a compact $K\subset \mathbb R^d$ such that the supports $\mathrm{spt}(f_k)$ and $\mathrm{spt}(f)$ are all contained in $K$ and $f_k,f\in L^1(\mathbb R^d)$ then the $f_k$ have a common concentration modulus in the above sense. Indeed, assuming the claim were false, up to extracting a subsequence we find balls $B_{1/n}(x_n)$ with $x_n\to x\in K$ such that $\int_{B_{1/n}(x_n)}f_n(y)\diff y>\epsilon_0>0$, and this generates a contraddiction using the fact that $f$ is locally integrable near $x$.

\par We note the following result which is a reformulation of \cite[Thm. 2.4]{ButChampdePas16}. Here we assume that 
\begin{equation}\label{assumptioncbcdp}
\mathsf{c}(x,y)=g(x-y)=l(|x-y|)\quad\text{where}\quad l:(0,\infty)\to[0,\infty)\quad\text{is}\quad\left\{\begin{array}{l}\mbox{continuous}\ ,\\\mbox{strictly decreasing}\ ,\\\mbox{such that }\lim_{t\to 0^+}l(t)=+\infty\ .\end{array}\right.
\end{equation}
As explained in Remark 2.1 from \cite{ButChampdePas16}, the strictly decreasing assumption can be weakened to $l$ being bounded at $+\infty$. These hypotheses are satisfied by $g(x)=|x|^{-s}$ for $s>0$, but not by $g(x) = -\log|x|$.
\begin{proposition}[Concentration modulus estimates on transport plans]\label{modulusint}
Assume $\mu\in \mathcal P(\mathbb R^d)$ has density $\rho$, and let $\omega_\rho$ be a concentration modulus of $\rho$. Assume that the cost $\mathsf{c}$ satisfies \eqref{assumptioncbcdp}. Then for each $N\in\mathbb N$ there exists $r_{N, \omega_\rho,\mathsf{c}}>0$ depending on $N, \omega_\rho$, such that for any optimal plan $\gamma_N$ realizing ${\mathcal F}_{N, \mathsf{c}}(\mu)$ there holds for all $(x_1,\ldots,x_N)\in\op{spt}\gamma_N $
\begin{equation}\label{sptawaydiag}
(x_,\ldots, x_N)\in D_{r_N, \omega_\rho, \mathsf{c}}:=\{(x_1,\ldots, x_N):\quad \forall 1\le i\neq j\le N\ ,\quad |x_i-x_j|>r_{N,\omega_\rho,\mathsf{c}}\}\ .
\end{equation}
Moreover we have the explicit bound
\begin{equation}\label{boundawaydiag}
r_{N,\omega_\rho,\mathsf{c}}\ge g^{-1}\left(\frac{N^2(N-1)}2\ \omega_\rho\left(\frac{1}{N^2(N-1)}\right)\right)\ .
\end{equation}
For the particular case where $g(x)=|x|^{-s}$, we may take
\begin{equation}\label{rnos}
r_{N,\omega_\rho,s}=\left(\frac{N^2(N-1)}{2}\omega_\rho\left(N^{-2}(N-1)^{-1}\right)\right)^{-1/s}>0.
\end{equation}
\end{proposition}

\section{Useful properties of the minimum Jellium energy, uniformly in $d-2\le s<d$}
\label{usefulprop}

We will derive in this section some properties for the Jellium energy, applied in crucial steps in the proofs of Lemma \ref{subaddjell} and Lemma \ref{convotjel}.

\subsection{Point separation of minimizers for the Jellium energy} 
We will show here that we can constrain the minimizer $(x_1,\ldots,x_N)$ to be in $(K_R)^N$ without altering either the minimizing value of $\Xi_{N,s}$ or its minimizers, and that we have a uniform-in-$s$ point separation result. The method follows closely the analogous result from \cite[Thm. 5]{PS}. Note however that the point separation was not explicitly quantified in \cite{PS}, which we need in order to show that this separation is independent on $s$ on any interval $s\in[d-2,d-\epsilon]$, where $0<\epsilon<2$. Moreover the minimization problem treated in \cite{PS} was of the type \eqref{gasmin} and not \eqref{wignermjer}, thus  we indicate modifications needed to adapt \cite[Thm. 5]{PS} to the different minimization treated here, for completeness. For clarity of exposition, we state the result from Lemma \ref{unifbound} for the case of a measure $d\mu(x)=1_{K_R}(x)\diff x$ as needed in this paper, however the same proof holds for any measure $\mu$ of mass $N$ with density $\rho$ which is bounded below and above by positive constants, when the minimization can be shown to be of form 
\begin{equation}\label{wignermin} \Xi_{N,s}(\mu):=\min\left\{E_{\mathrm{Jel}}(\mu,\nu_{\vec x}):\ \vec x= (x_1,\ldots, x_N)\in(\mathbb R^d)^N, x_1,\ldots, x_N\in \mbox{supp}(\mu)\right\}.
\end{equation}

\medskip

The following result was proved (in the case of points on a sphere) first in \cite{bhs}, and later transferred to the framework of \cite{PS} which framework we adapt in our proof below.
\begin{lemma}\label{unifbound}
Let $K_R\subset \R^d$ be a cube of volume $R^d=N\ge 2$; for each $d-2\le s<d$, let ${\vec x}^{N,s}=(x^{N,s}_1,\ldots, x^{N,s}_N),$ be the optimizer of ${\Xi}_{N,s}(K_R)$. Then the following hold:

\begin{itemize}
\item [(a)] We may add the condition that $x_i$ should be contained in $K_R$, so that 
\begin{equation}\label{reformulc1aa}
\Xi_{N,s}(K_R)=\min\left\{E_{\mathrm{Jel}}(K_R,\nu_{\vec x}):\ \vec x= (x_1,\ldots, x_N)\in(K_R)^N\right\}.
\end{equation}
Additionally, (\ref{reformulc1aa}) implies via (\ref{reformulc1}) that 
\begin{equation}\label{reformulc1ab}
\textsf{C}_{\mathrm{Jel}}(d,s)=\lim_{R^d=N\to\infty}\frac{\Xi_{N,s}(K_R)}{N}.
\end{equation}
\item [(b)] Then for every $\epsilon \in ]0,2[$ there exists a number $r_\mathrm{sep}(\epsilon)>0$ such that the following holds. If $d\ge 3$, $K_R\subset\R^d$ is a cube of volume $R^d=N\in \N$ and $d-2\le s\le d-\epsilon$, then we have 
\begin{equation}\label{drnos}
 \min_{i\neq j}|x_i^{N,s}-x_j^{N,s}|\ge r_\mathrm{sep}(\epsilon).
\end{equation}

\end{itemize}
\end{lemma}
\begin{proof}
The proof of the present lemma is largely overlapping with the one of \cite[Thm. 5]{PS}, therefore we only sketch the parts in which it differs somewhat from the one present in \cite{PS}. 

To begin with, we introduce some more notations. With the notations introduced in Section \ref{Jelintro}, we may write for the Jellium energy, similarly to Proposition 1.6 from \cite{PS} 
\begin{equation}\label{rewrite_jellium}
E_\mathrm{Jel}(K_R,\nu_{\vec x}) = \lim_{\eta\to 0}\frac{1}{c_{s,d}}\left(\int_{\R^{d+1}}|y|^\gamma\verti{\nabla h_{N,\eta}}^2 - c_{s,d}N \textsf{c}(\eta)\right),
\end{equation}
where $h_N$ is the unique decaying at infinity solution to the equation (with $1_{K_R}$ here meaning the extended measure $\bar \mu$ for $d\mu(x)=1_{K_R}(x)\diff x$ and $\delta_{x_i}$ meaning $\bar \delta_{x_i}=\delta_{(x_i,0)}$, by abuse of notation)
\begin{equation}\label{pde_hn}
-\op{div}(\verti{y}^\gamma\nabla h_N) = c_{s,d}\left(\sum_{i=1}^N\delta_{x_i} - 1_{K_R}\right),
\end{equation}
and with the same notations as in Section \ref{sec:truncation}, we may define $h_{N,\eta}$ for $\eta>0$ as the regularization of $h_N$ at scale $\eta$, given by 
\begin{equation}\label{regular_hn}
h_{N,\eta}(X):=h(X)- \sum_{p\in\textsf{set}(\nu_{\vec x})\times\{0\}}\textsf{f}_\eta(p-X)\ .
\end{equation}
The re-expression \eqref{rewrite_jellium} follows from the properties \eqref{fundsolrd+1}, and the definition \eqref{pde_hn} by integration by parts, precisely with the same reasoning as for the splitting formula from \cite[\S 2.1]{PS}.

\medskip

(a) We start with the formulas \eqref{rewrite_jellium} and \eqref{pde_hn} for the expression of $E_{\mathrm{Jel},s}(K_R,\nu_{\vec x})$. We may also write, due to formula \eqref{fundsolrd+1},
\begin{equation}\label{eqn_hn}
h_N= \textsf{c}*\left(\sum_{i=1}^N \delta x_i - 1_{K_R}\right).
\end{equation}
By comparing the energy of the configuration $(x_1,\ldots, x_N)$ to the one of $(x,x_2,\ldots, x_N)$, in which we exchanged the point $x_1$ with an arbitrary other choice $x\in\R^d$, we find precisely like in \cite[Lem. 4.1]{PS}, that $x_1$ must be a minimum point of the function 
\begin{equation}\label{pot_allbutone}
x\mapsto\mathcal U(x):=\textsf{c}*\left(\sum_{i=2}^N\delta_{x_i} - 1_{K_R}\right)(x).
\end{equation}
Then we recall the result of \cite[Lem. 4.2]{PS}, namely that if for some open $A\subset\R^{d+k}$ and for a given function $h:\R^{d+k}\to \R$ we have
\begin{equation}\label{max_princ}
-\op{div}\left(\verti{y}^\gamma \nabla h\right)\ge 0\quad \mbox{on}\quad A,\quad \int_A\verti{y}^\gamma\verti{\nabla h}^2<+\infty,
\end{equation}
then $h$ has no local minimum in $A$. Applied for $A=\op{int}(K_R)$, the above criterion automatically shows that the minimizers for the problem on the right hand side of \eqref{reformulc1aa} will stay inside $K_R$, proving the first item of our lemma. 

\medskip

(b) Next, using the above criterion, we proceed as in the proof of \cite[Thm. 5]{PS} in order to show the point separation result from the second item of the lemma. Up to renaming the points $x_1,\ldots,x_N,$ we may assume that 
\begin{equation}\label{x1x2_closest}
\min_{i\neq j}\verti{x_i-x_j}=\verti{x_1-x_2}.
\end{equation}
thus it suffices to prove that $x_1$ is separated from $x_2$. We first separate the contribution of $\mathcal U$ coming from charges near $x_2$ and that coming from the remaining charges. To do this, we use the formula of $\mathcal U$ from \eqref{pot_allbutone} and we write:
\begin{equation}\label{near_far}
\mathcal U(x)=\underbrace{\textsf c*\left(\delta_{x_2} - 1_B\right)(x)}_{:=\mathcal U^\mathrm{near}(x)} + \underbrace{\textsf{c}*\left(\sum_{i=3}^N \delta_{x_i} - 1_{K_R} + 1_B\right)(x)}_{:=\mathcal U^\mathrm{rem}(x)},
\end{equation}
where $B\subset\R^d$ is a ball centered at $x_2$ of volume $1$. Since $\mathcal U^\mathrm{near}$ and $\mathcal U^\mathrm{rem}$ have expressions of the form \eqref{eqn_hn}, they also satisfy analogues of \eqref{pde_hn}. We find, by the same discussion as in the analogous passage from \cite[proof of Thm. 5]{PS}, that conditions \eqref{max_princ} are valid for the choice $A=\R^{d+k}\setminus((K_R\setminus B)\times \R^k)$ and 
\begin{equation}\label{conclude_urem}
\begin{array}{c}\mbox{the minimum of}\quad \mathcal U^\mathrm{rem}\quad\mbox{is achieved on}\quad (K_R\setminus B)\times\{0\},\\\mbox{and it is not achieved on} \quad \left(\mathbb R^d\setminus(K_R\setminus B)\right)\times\{0\}.
\end{array}
\end{equation}
Next, we can write, with $r_B:=(d/\verti{\mathbb S^{d-1}})^{\frac1{d}}$ denoting the radius of $B$, and using the inclusion $B_{r_B}(x_2)\subset B_{r_B+\verti{x-x_2}}(x)$ and a translation of the coordinates
\begin{eqnarray*}
\int_B \textsf{c}(y-x)\diff y&=&\int_{B_{r_B}(x_2)}\verti{x-y}^{-s}\diff y\le\int_{B_{r_B+\verti{x-x_2}}(x)}\verti{x-y}^{-s}\diff y 
\\
&=&\int_{B_{r_B+\verti{x-x_2}}(0)}\verti{y}^{-s}\diff y = \frac{\verti{\mathbb S^{d-1}}}{d-s}\left(r_B +\verti{x-x_2}\right)^{d-s},
\end{eqnarray*}
and we find
\begin{equation}\label{eq_near}
\mathcal U^\mathrm{near}(x) = \textsf{c}(x-x_2) - \int_B \textsf{c}(y-x)\diff y\ge 
\verti{x-x_2}^{-s} - \frac{\verti{\mathbb S^{d-1}}}{d-s}\left(r_B +\verti{x-x_2}\right)^{d-s}.
\end{equation}
Then, as $d>s$, for $\verti{x-x_2}\le r_B$ we can bound 
\begin{equation}\label{firstbound}
\verti{x-x_2}^{-s}- \frac{\verti{\mathbb S^{d-1}}}{d-s}\left(r_B +\verti{x-x_2}\right)^{d-s} \ge\verti{x-x_2}^{-s}- \frac{\verti{2^{d-s}\mathbb S^{d-1}}}{d-s}r_B^{d-s},
\end{equation}
and by a direct computation we see that the right-hand side of $\eqref{firstbound}\ge r_B^{-s}$ if and only if
\begin{equation}\label{condition}
\verti{x-x_{2}}\le \left(\frac{2^{d-s}\verti{\mathbb S^{d-1}}}{d-s}r_B^{d-s} + r_B^{-s}\right)^{-\frac1{s}}=\left(\frac{2^{d-s} d}{d-s} + 1\right)^{-\frac1{s}}r_B,
\end{equation}
We note that 
\begin{equation}\label{condition_weak}
r_B\min_{s\in[d-2,d-\epsilon]}\left(\frac{2^{d-s} d}{d-s} + 1\right)^{-\frac1{s}}\ge r_\mathrm{sep}(\epsilon):=r_B\left(\frac{4d}{\epsilon}+1\right)^{-\frac1{d-2}},
\end{equation}
and then under the condition $\verti{x-x_2}\le r_\mathrm{sep}(\epsilon)$ condition \eqref{condition} holds for all $s\in[d-2,d-\epsilon]$, and in this case we can continue the chain of inequalities started in \eqref{eq_near}. We find that
\begin{equation}\label{conclude_unear}
\verti{x-x_2}^{-s} - \frac{\verti{\mathbb S^{d-1}}}{d-s}\left(r_B +\verti{x-x_2}\right)^{d-s}\ge r_B^{-s}\ge\max_{x\in K_R\setminus B}\mathcal U^\mathrm{near}(x)\quad\mbox{for}\quad|x-x_2|\le r_\mathrm{sep}(\epsilon).
\end{equation}
Here the second inequality follows by examining the contributions to $\mathcal U^\mathrm{near}(x)$ for $\verti{x-x_2}>r_B$: the first term $\textsf{c}(x-x_2)$ in the left-hand side of (\ref{eq_near}) is decreasing in $\verti{x-x_2}$ and thus it is bounded by $r_B^{-s}$ on $K_R\setminus B$, while the second one is negative, giving the desired inequality in \eqref{conclude_unear}. 

\medskip

As a consequence of \eqref{conclude_urem} and \eqref{conclude_unear}, we find that there exists a point $\bar x\in K_R\setminus B$ such that for all $\tilde x_1\in B_{r_\mathrm{sep}(\epsilon)}(x_2)$ we have 
\begin{itemize}
\item $\mathcal U^\mathrm{rem}(\bar x) < \mathcal U^\mathrm{rem}(\tilde x_1)$
\item $\mathcal U^\mathrm{near}(\bar x) \le \mathcal U^\mathrm{near}(\tilde x_1)$
\end{itemize}
and summing the above inequalities, we find that the minimum of $\mathcal U$ is not achieved on $B_{r_\mathrm{sep}(\epsilon)}(x_2)$.
\end{proof}

\subsection{Lower bound for the minimum Jellium energy, uniformly in $0<s<d$}
\label{ApB.2}
For completeness purposes, and since this is needed in \eqref{existlim} above, we next prove in Lemma \ref{LNequiv} a uniformly in $\epsilon<s<d-\epsilon$ lower bound for $\Xi_{N,s}(\mu)/N$ and for $E_{\mathrm{Jel},w}(K_R,\nu_{\vec x})/N$, where $0<\epsilon\le d/2$, and where $w$ can be the kernel appearing as an error term in \eqref{decomptouse} or another kernel with similar rough bounds as in Proposition \ref{prop3ws}. Our lower bound is the equivalent, for the Jellium energy, to the uniformly in $\epsilon<s<d-\epsilon$ generalised Lieb-Oxford lower bound for Uniform Electron Gas from Lemma C.1 in \cite{cotpet}. (See also, for example, \cite{Lieb79, Li83, LO, liebsolovejyngvason, LundNamPort} for Lieb-Oxford inequalities with sharper constants under different conditions on the exponents $s$.) 

Note that while one can find in the existing literature lower bounds for the minimum Jellium energy for $0<d-2\le s<d-\epsilon'$, where $0\le \epsilon'\le \min(2,d/2)$ (as described for Coulomb and Riesz gases in \cite{rs} for $s=d-2$ and in \cite{PS} for general $s\in[d-2,d)$), they are not uniform in the $0<d-2\le s<d-\epsilon'$ parameter. 

We will also obtain upper/lower bounds for $E_{\mathrm{Jel},w}(K_R,\nu_{\vec x})$, uniformly in $\epsilon\le s<d-\epsilon$, which bounds do not seem easily amenable to the methods in \cite{rs} and in \cite{PS} for $s\in[d-2,d)$. Note that a generalization of our method below, which method is to the best of our knowledge new and of independent interest, allows us to study much more general costs and densities in the forthcoming paper \cite{CotGnPet}.

Before we proceed with the proof of Lemma \ref{LNequiv}, we will first need to show in Lemma \ref{lem:renorm} for $E_{\mathrm{Jel},\mathsf{c}}(K_R,\nu_{\vec x})$, with $\mathsf{c}$ satisfying either (\ref{value c}) or being equal to $w$, the analogue of the monotonicity formula \cite[Lem. 2.3]{PS}.

To prove the result of the lemma, we are going to use the following general representation as proved in \cite[Thm. 1]{HS}: Let $V:\mathbb{R}^d\to\mathbb{R}$ be a radial function which is $d+1$ times differentiable away from $x=0$. Assume also  that $\lim_{|x|\to\infty}|x|^m \partial_{|x|}^m V(x)=0$ for all $0\le m\le [d/2]+1$.  Then
\begin{equation}\label{HSrep}
V(x)=\int_0^\infty {\mathds{1}}_{B_{\frac r 2}}*{\mathds{1}}_{B_{\frac r 2}}(x) f(r)\diff r\ =\int_0^\infty h_r(|x|) f(r)\diff r\ ,
\end{equation}
where we recall that $B_{r/2}:=\{x\in\mathbb{R}^d:|x|<r/2\}$, we define
\begin{equation}\label{hatr}
 h_r(|x|) := {\mathds{1}}_{B_{\frac r 2}}*{\mathds{1}}_{B_{\frac r 2}}(x)=
 \bigg\{\begin{array}{ll}\frac{1}{\Gamma\left(\frac{d+1}{2}\right)}\left(\frac \pi 4\right)^{\frac{d-1}2}\int_{|x|}^r\left(r^2-y^2\right)^{\frac{d-1}2}\, \diff y&\mbox{ if }\quad |x|\le r,\\[3mm]
 0&\mbox{ if }\quad |x|>r,\end{array}
\end{equation}
and we have
\begin{equation}\label{HSf}
f(r)=\frac{(-1)^{d+1}}{\Gamma([d/2]+2)}\frac{2}{(\pi r^2)^{(d-1)/2}}\int_r^\infty V^{(d+1)}(v)v (v^2-r^2)^{(d-3)/2}\diff v~~~\mbox{and}~~~ h_r(|x|):={\mathds{1}}_{B_{\frac r 2}}*{\mathds{1}}_{B_{\frac r 2}}(x)\ ,
\end{equation}
and where by abuse of notation $V(v)=V(|x|),$ with $|x|=v$.

Note that our costs $\mathsf{c}$ of interests, the Coulomb and Riesz costs from (\ref{value c}) and $w$, are positive definite kernels which satisfy (\ref{HSrep}). More precisely, they are radial functions, differentiable $d+1$ times and with $\lim_{|x|\to\infty}|x|^m \partial_{|x|}^m V(x)=0$ for all $0\le m\le [d/2]+1$. (For a justification of these properties for $w$, see Propositions A.4, A.5 and A.6 from \cite{cotpet}, and also the decomposition from Section A.1.2 therein). This means that they can be written as
\begin{equation}\label{expr_cost}
 \mathsf{c}(x)=\int_0^\infty h_r(|x|)f(r)\diff r,~~f(r)\ge 0.
\end{equation}
In particular, if $f_1(r)$ is the weight $f$ corresponding to $V(x)=|x|^{-s}$ as obtained from \eqref{HSf}, we also have
\[
f_1(r)=c(s,d)\int_r^\infty\frac{1}{v^{s+d+1}}v(v^2 - r^2)^{\frac{d-3}{2}}\diff v\ge 0\ ,
\]
because of the fact that the $d+1$-th derivative of $|x|^{-s}$ has sign $(-1)^{d+1}$. Furthermore, if $w_*(x):=w(x)-\frac{C}{|x|^s}$, then it is a radial function satisfying the properties of \cite[Thm. 1]{HS}, so we can use the same arguments as in the Proof of Lemma A.7 from \cite{cotpet} to write 
\[
w_*(|x|)=\int_0^\infty \int_{\mathbb{R}^d}1_{B_{r/2}}(u)1_{B_{r/2}}(x-u) f_2(r)\diff u\, \diff r\,
\]
where we have for a constant $c(d)>0$ depending only on the dimension that
\begin{equation}\label{boundf2}
f_2(r) = (-1)^{d+1}c(d)\int_r^\infty w_*^{(d+1)}(v)v(v^2 - r^2)^{\frac{d-3}{2}}\diff v\ .
\end{equation}
Therefore
\begin{multline*}
w(x,y)=\frac{C}{M}|x-y|^{-s}+\frac{1}{M} w_*(x-y)
=\frac{1}{M}\int_0^\infty 1_{B_{r/2}}(x-u)1_{B_{r/2}}(y-u) (Cf_1(r)+f_2(r))\diff u \diff r\ ,
\end{multline*}
where
\begin{equation}\label{exprf1f2}
Cf_1(r)+f_2(r)=c(d)\int_r^\infty\left(Cv^{-s-d-1} +(-1)^{d+1}w_*(v)^{(d+1)}\right)v(v^2- r^2)^{\frac{d-3}{2}}\, \diff v\ .
\end{equation}
Due to the bound (A.6) from Lemma A.4 in \cite{cotpet} for $|\beta|\le d+1$, we find that up to enlarging the above constant $C$ by a factor depending only on $d,\epsilon$ for our choice of $s$, there holds for a constant $\widetilde C(d,\epsilon,C)>0$ 
\begin{eqnarray}\label{intervalderivatives}
\widetilde{C} v^{-s-d-1}\ge Cv^{-s-d-1} +(-1)^{d+1}w_*(v)^{(d+1)}\ge C v^{-s-d-1} -|w_*(v)^{(d+1)}|\ge 0\ ,
\end{eqnarray}
which gives $Cf_1(r)+f_2(r)\ge 0$.

\par For $\mathsf{c}$ as in \eqref{expr_cost} it is convenient to introduce a regularized version according to
\begin{equation}\label{cost_epsilon}
\mathsf{c}_\eta(x) := \int_\eta^\infty h_r(|x|)f(r)\, \diff r,
\end{equation}
and it follows that $c_\eta$ and $c-c_\eta$ are positive definite for $\epsilon>0$ since the kernel $h_r(|x|)$ is positive definite for each $r>0$ and $f(r)\ge 0$.

We also note the following result
\begin{lemma}[Renormalized energy]\label{lem:renorm}
Assume that $\mathsf{c}$ has the expression \eqref{expr_cost}. Suppose $N \ge 2$, and assume that at least two of the terms in parenthesis from the second line in \eqref{Ejel_ren} below give finite integrals when we replace $\mathsf{c}_\alpha$ by $\mathsf{c}$. Then we have with the notation $E_{\mathrm{Jel}}[\mathsf{c},\alpha] := E_{\mathrm{Jel},c_\alpha}$ (where here we recall the definition given in (\ref{defjellium}) above)
\begin{multline}\label{Ejel_ren}
E_{\mathrm{Jel},\mathsf{c}}(K_R,\nu_{\vec x})=\lim_{\alpha\to 0} E_\mathrm{Jel}[\mathsf{c},\alpha](K_R,\nu_{\vec x})\\
:=\lim_{\alpha\to 0}\left(\sum_{p\neq q\in\mathsf{set}(\nu)}\mathsf{c}_\alpha(p-q) - 2\sum_{p\in\mathsf{set}(\nu)}\int_{K_R}\mathsf{c}_\alpha(p-y) \, \diff y + \int_{K_R}\int_{K_R} \mathsf{c}_\alpha(x-y) \, \diff x \, \diff x\right).
\end{multline}
\end{lemma}
\begin{proof}
This is immediate by applying the monotone convergence theorem to each term in the above.
\end{proof}
\begin{lemma}[Monotonicity of the regularized Jellium energy]\label{lem:monot}
Assume that $\mathsf{c}$ has the expression \eqref{expr_cost} and that the quantities below are finite.
For any $x_1, \dots, x_N \in \R^d$, and any $1>\eta>\alpha>0$,  we have, with the notation \eqref{cost_epsilon},
\begin{multline}\label{monot}
-2 N\int_{B_\eta}|c_\eta-c_\alpha|(x) \diff x \le E_\mathrm{Jel}[\mathsf{c},\alpha](K_R,\nu_{\vec x}) -E_\mathrm{Jel}[\mathsf{c},\eta](K_R,\nu_{\vec x})\\
\le \int_{|x-y|\le\eta}\left(\mathsf{c}_{\alpha} -\mathsf{c}_\eta\right)(x-y)\, \diff x \, \diff y  + \sum_{i\neq j,  |x_i-x_j|\le \eta}\left(\mathsf{c}_{\alpha} -\mathsf{c}_\eta\right)(x_i-x_j).
\end{multline}
\end{lemma}

\begin{proof}
Suppose $0 < \alpha \le \eta \le 1$. With the notation \eqref{expr_cost}, we have
\[
(\mathsf{c}_{\alpha} - \mathsf{c}_\eta)(x) \stackrel{\eqref{cost_epsilon}}{=} \int_{\alpha}^\eta h_r(|x|)f(r)  \, \diff r.
\]
We distinguish between different cases using \eqref{hatr}, which in particular gives that $h_r(1)=0$ for $r\le\eta$, and we obtain
\[
(\mathsf{c}_{\alpha} - \mathsf{c}_\eta)(x) = \begin{cases}
0& \mbox{ for } \lvert x \rvert >\eta, \\
\int_{\lvert x \rvert}^\eta h_r(|x|)f(r)\, \diff r & \mbox{ for } \alpha \le |x| \le \eta, \\
\int_{\alpha}^\eta h_r(|x|)f(r)\, \diff r& \mbox{ for } \lvert x \rvert < \alpha.
\end{cases}
\]
Hence $(\mathsf{c}_{\alpha} - \mathsf{c}_\eta)(x) \ge 0$, and we obtain
\begin{align}
& E_\mathrm{Jel}[\mathsf c, \alpha](K_R,\nu_{\vec x}) -E_\mathrm{Jel}[\mathsf{c}, \eta](K_R,\nu_{\vec x}) \nonumber \\
& \quad =\sum_{i\neq j}\left(\mathsf{c}_{\alpha} -\mathsf{c}_\eta\right)(x_i-x_j) - 2\sum_{i=1}^N\int_{K_R}\left(\mathsf{c}_{\alpha} -\mathsf{c}_\eta\right)(x_i-x)\, \diff x +\int_{K_R}\int_{K_R}\left(\mathsf{c}_{\alpha} -\mathsf{c}_\eta\right)(x-y)\, \diff x \, \diff y \nonumber \\
& \quad \ge - 2\sum_{i=1}^N\int_{K_R}\left(\mathsf{c}_{\alpha} -\mathsf{c}_\eta\right)(x_i-x) \, \diff x
= - 2\sum_{i=1}^N\int_{K_R\cap |x_i-x|\le\eta}\left(\mathsf{c}_{\alpha} -\mathsf{c}_\eta\right)(x_i-x) \, \diff x\ge -2 N\int_{B_\eta}|c_\eta-c_\alpha|(x) \diff x,
\label{first_reexpr}
\end{align}
which gives the lower bound in (\ref{monot}). Since  $(\mathsf{c}_{\alpha} - \mathsf{c}_\eta)(x) \ge 0$, the upper bound is immediate.
\end{proof}
\begin{lemma}(Upper/lower bounds for the minimum Jellium energy, uniformly in $\epsilon\le s<d-\epsilon$)
\label{LNequiv}

Let $0<\epsilon\le d/2$. For all $\epsilon< s<d-\epsilon$ and for $R=N^d$, there exists a constant $c_{\mathrm{LN}}(\epsilon,d)<0$ such that for all $N$, we have
\begin{equation}\label{loterm1LN}
-\infty<c_{\mathrm{LN}}(\epsilon,d) \le \frac{\Xi_{N,s}(K_R)}{N}\le 0~~~\mbox{and}~~~\frac{1}{M}c_{\mathrm{LN}}(\epsilon,d)\le \frac{\Xi_{N,w}(K_R)}{N}\le 0.
\end{equation}
Furthermore, let $\vec{x}=(x_1,\ldots,x_N)\in K_R$ be a minimizer for $E_{\mathrm{Jel},\mathsf{c}}(K_R,\nu_{\vec x})$. Then for all $0<\epsilon<\min(2,d/2)$ and $0<d-2\le s\le d-\epsilon$, there exists a constant $C_{\mathrm{up}}(\epsilon,d)$ such that there holds
\begin{equation}
\label{loterm1LNupbd}
\frac{E_{\mathrm{Jel},w}(K_R,\nu_{\vec x})}{N}\le \frac{1}{M}C_{\mathrm{up}}(\epsilon,d).
\end{equation}
\end{lemma}
\begin{proof}
We start by proving the bounds for $\Xi_{N,s}(K_R)/N$. We will show to begin with the upper bound. We have for all $(x_1,\ldots,x_N)\in K_R$ such that $x_i\neq x_j$ for $i\neq j$,
\[
\Xi_{N,s}(K_R)\le \sum_{i,j=1,i\neq j}^N\frac{1}{|x_i-x_j|^s} - 2\sum_{i=1}^N\int_{K_R}\frac{1}{|x_i-y|^s} \diff y + \int_{K_R}\int_{K_R} \frac{1}{|x-y|^s}\diff x\ \diff y.
\]
Integrating in the above on both sides of the inequality with respect to $\mu^{\otimes N}$ for $d\mu(x)=\tfrac{1}{N}1_{K_R}(x)\diff x$ gives 
\[
\Xi_{N,s}(K_R)\le -\frac{1}{N}\int_{K_R}\int_{K_R}\frac{1}{\verti{x-y}^s}\diff x\ \diff y\le 0.
\]
\medskip

From the lower bound in \eqref{monot} by using $\eta = 1$ and letting $\alpha \to 0$, we get for $\mathsf{c}$ satisfying (\ref{value c})
\begin{eqnarray}\label{separation_used}
E_{\mathrm{Jel},\mathsf{c}}(K_R,\nu_{\vec x})&=& \lim_{\alpha \to 0}E_\mathrm{Jel}[\mathsf{c},\alpha](K_R,\nu_{\vec x})\ge E_\mathrm{Jel}[\mathsf{c},1](K_R,\nu_{\vec x}) - 2 N\int_{B_1}|c_1-c|(x) \diff x\nonumber\\
&=&\int\int\mathsf{c}_1(x-y)\diff(\mathds{1}_{K_R}-\nu_{\vec x})(x)\diff(\mathds{1}_{K_R}-\nu_{\vec x})(y) - 2 N\int_{B_1}|c_1-c|(x) \diff x\nonumber\\
&\ge&-4N\int_{B_1}\verti{x}^{-s}\diff x = -4N\verti{\mathbb S^{d-1}}\frac{1}{(d-s)}\ge c_{\mathrm{LN}}(\epsilon,d) N,
\end{eqnarray}
where for the second equality we have used the definition \eqref{Ejel_ren} and summed and subtracted the $p=q\in\mathsf{set}(\nu)$ terms from the first sum in \eqref{Ejel_ren}. This proves the lower bound in  $E_{\mathrm{Jel},\mathsf{c}}(K_R,\nu_{\vec x})$, as desired.
We move now to the proof of the bounds for the $w$ energy terms. To this purpose we use that, by Lemmas A.1-A.6 in \cite{cotpet}, the error term $w$ has an expression of the type \eqref{expr_cost} and we can chose $C(\epsilon,d)>0$ in the definition of $w$ from Proposition \ref{prop3ws}, such that there exist $C_1(\epsilon,d), C_2(\epsilon,d)$ for which the following holds:
$$
w(x,y)-\frac{C_1(\epsilon,d)}{M|x-y|^s}~~~\mbox{and}~~~\frac{C_2(\epsilon,d)}{M|x-y|^s}-w(x,y) \quad \mbox{have expressions of the form \eqref{expr_cost}.}
$$
We can then apply Lemma \ref{lem:monot} to $w$. From the upper bound in \eqref{monot} by using $\eta = r_{\mathrm{sep}}(\epsilon)$ from (\ref{drnos}), and letting $\alpha \to 0$, and further using the positivity and positive definiteness of $\frac{C_2(\epsilon,d)}{M|x-y|^s}-w(x,y)$, we now obtain
\begin{eqnarray}
\label{wsepnbd}
E_{\mathrm{Jel},w}(K_R,\nu_{\vec x})&\le&\int\int w_{r_{\mathrm{sep}}(\epsilon)}(x-y)\diff(\mathds{1}_{K_R}-\nu_{\vec x})(x)\diff(\mathds{1}_{K_R}-\nu_{\vec x})(y)\nonumber\\
&&+\int\int_{|x-y|\le r_{\mathrm{sep}}(\epsilon)}\left(w -w_{r_{\mathrm{sep}}(\epsilon)}\right)(x-y)\, \diff x \, \diff y+\sum_{i\neq j,  |x_i-x_j|\le r_{\mathrm{sep}}(\epsilon)}\left(w -w_{r_{\mathrm{sep}}(\epsilon)}\right)(x_i-x_j)\nonumber\\
&\le&\int\int w(x-y)\diff(\mathds{1}_{K_R}-\nu_{\vec x})(x)\diff(\mathds{1}_{K_R}-\nu_{\vec x})(y)+\int\int_{|x-y|\le r_{\mathrm{sep}}(\epsilon)}w(x-y)\, \diff x \, \diff y\nonumber\\
&\le&\int \frac{C_2(\epsilon,d)}{M|x-y|^s}\diff(\mathds{1}_{K_R}-\nu_{\vec x})(x)\diff(\mathds{1}_{K_R}-\nu_{\vec x})(y)+\int\int_{|x-y|\le r_{\mathrm{sep}}(\epsilon)}\frac{C_2(\epsilon,d)}{M|x-y|^s}\, \diff x \, \diff y\nonumber\\
&\le&\frac{2C_2(\epsilon,d)N}{M}\int_{B_{ r_{\mathrm{sep}}(\epsilon)}}|c_{ r_{\mathrm{sep}}(\epsilon)}-c|(x) \diff x\le \frac{4C_2(\epsilon,d)N}{M}\sup_{s\in [d-2,d-\epsilon]} \int_{\mathbb{R}^d}\frac{1}{|x|^s}\diff x<\infty,
\end{eqnarray}
where for the second inequality we applied (\ref{drnos}), and for the fourth inequality we used (\ref{separation_used}). The lower bound follows a similar argument to the one in (\ref{separation_used}), and will be omitted.
\end{proof}

\section*{Acknowledgments}

Both authors gratefully acknowledge the support of a Royal Society International Exchanges Grant, and of an IHP Research-in-Pairs grant. They also thank the organisers of the Fields Institute 2014 Thematic Program on
'Variational Problems in Physics, Economics and Geometry', where this research was first started. MP was funded by an EPDI fellowship.

Both authors are grateful to the organizers of the 'Scaling limits, rough paths, quantum field theory' program at the Isaac Newton Institute for the invitation to the program, at which the $s=d-2$ case was elucidated.
%
\bibliographystyle{siam}

\end{document}